\documentclass[aps,pra,twocolumn,superscriptaddress,longbibliography]{revtex4-2} 

\usepackage{thmtools}
\usepackage{amsmath}
\usepackage{amsfonts}
\usepackage{amssymb}
\usepackage{graphicx}
\usepackage{color}
\usepackage{accents}
\usepackage[colorlinks=true, citecolor=cyan]{hyperref}
\usepackage{enumitem}

\newcommand{\ket}[1]{|#1\rangle}
\newcommand{\bra}[1]{\langle #1|}

\usepackage{tikz}

\makeatletter 
\newsavebox{\@brx}
\newcommand{\llangle}[1][]{\savebox{\@brx}{\(\m@th{#1\langle}\)}%
  \mathopen{\copy\@brx\kern-0.5\wd\@brx\usebox{\@brx}}}
\newcommand{\rrangle}[1][]{\savebox{\@brx}{\(\m@th{#1\rangle}\)}%
  \mathclose{\copy\@brx\kern-0.5\wd\@brx\usebox{\@brx}}}
\makeatother

\newlength{\dhatheight} 

\newtheorem{theorem}{Theorem} 
\newtheorem{lemma}[theorem]{Lemma}
\newtheorem{proposition}[theorem]{Proposition}

\newcommand{\qed}{\nobreak \ifvmode \relax \else
      \ifdim\lastskip<1.5em \hskip-\lastskip
      \hskip1.5em plus0em minus0.5em \fi \nobreak
      \vrule height0.75em width0.5em depth0.25em\fi}

\usepackage{amsmath}
\usepackage{enumitem}
\usepackage{tikz}

\begin{document}

\title{Dynamical signatures of conventional and asymptotic quantum many-body scars \\ on a trapped ion simulator}

\author{Leonard Logari\'c} 
\email[]{logaricl@tcd.ie}
\affiliation{Department of Physics, Trinity College Dublin, Dublin 2, Ireland} 
\affiliation{Trinity Quantum Alliance, Unit 16, Trinity Technology and Enterprise Centre, Pearse Street, Dublin 2, D02 YN67, Ireland}
\author{John Goold} 
\email[]{gooldj@tcd.ie}
\affiliation{Department of Physics, Trinity College Dublin, Dublin 2, Ireland}
\affiliation{Trinity Quantum Alliance, Unit 16, Trinity Technology and Enterprise Centre, Pearse Street, Dublin 2, D02 YN67, Ireland}
\affiliation{Algorithmiq Limited, Kanavakatu 3C 00160 Helsinki, Finland}
\author{Shane Dooley} 
\email[]{dooleysh@gmail.com}
\affiliation{Dublin Institute for Advanced Studies, School of Theoretical Physics,
10 Burlington Road, Dublin, D04 C932, Ireland}
\affiliation{Trinity Quantum Alliance, Unit 16, Trinity Technology and Enterprise Centre, Pearse Street, Dublin 2, D02 YN67, Ireland}

\date{\today}
 
\begin{abstract}
One of the promising applications of digital quantum processors is the simulation of many-body quantum systems. They have been already used to investigate several ergodicity violating mechanisms, such as many-body localisation, Hilbert space fragmentation and quantum many-body scars (QMBS). In addition to conventional QMBS, a recently discovered mechanism for ergodicity violation are the so-called asymptotic quantum many-body scars (AQMBS). These become more stable as system size is increased, leading to progressively longer thermalisation timescales. In this work, we show a connection between gapless excitations of certain ``reference'' Hamiltonians and AQMBS in a related set of Hamiltonian and circuit models. We use this connection to construct a 2-local model hosting both conventional and asymptotic scars. By exploiting the structure of the AQMBS states and the all-to-all connectivity of the Quantinuum H1-1 quantum processor, we prepare these states in logarithmic circuit depth, and probe their thermalisation under Floquet circuit dynamics. Performing simulations with up to $418$ entangling $ZZ$ gates, we find slower thermalisation times as the system size is increased, providing the first experimental signatures of AQMBS.

\end{abstract}

\maketitle

\section{Introduction}

Rapid advances in the laboratory have made it possible to experimentally explore the unitary dynamics of many-body quantum systems. In particular, state-of-the-art quantum simulators, including platforms based on Rydberg atoms \cite{Ada-19a}, superconducting circuits \cite{Kja-20a} and trapped ions \cite{Fos-24a-arxiv}, now routinely achieve high-fidelity preparation of entangled states, programmable interactions, and dynamics with relatively large circuit depth. These capabilities make it possible to investigate fundamental questions, such as the mechanisms responsible for thermalisation, or failure to thermalise, in isolated many-body systems systems \cite{Mor-18,Fau-24a,Fis-26a,Yan-25a-arxiv}. 

A phenomenon which has garnered significant attention is weak ergodicity breaking, in which a non-integrable system thermalises for almost all initial states but exhibits long-lived nonthermal dynamics for a restricted set of initial states. This is due the presence of quantum many-body scars (QMBS): a small number of atypical eigenstates embedded within an otherwise thermalising many-body spectrum \cite{Tur-18a, Doo-21a,Mou-22a,Cha-23b,Logaric-24}. When an initial state overlaps significantly with these anomalous eigenstates, the ensuing time evolution shows dynamical signatures of the QMBS that can be measured experimentally, such as long-lived coherent oscillations or anomalously slow decay of observables \cite{Ber-17,Blu-21,Su-23a,Zha-25a,Zha-23a}. Very recently, the existence of so-called asymptotic QMBS (AQMBS) was also proven theoretically \cite{Got-23}. These are special states, which are not eigenstates of a system at finite size, but emerge as QMBS only in the thermodynamic limit. As a consequence, AQMBS have the unusual dynamical property that they have relaxation times that \emph{diverge} with increasing system size.
 
Both QMBS and AQMBS have been shown theoretically to arise in a wide variety of many-body models \cite{Kun-25, Gio-26}. Nevertheless, experimental realisations of QMBS remain confined to a narrow class of systems \cite{Ber-17,Blu-21,Su-23a,Zha-25a,Zha-23a}, most notably those effectively described by PXP-like Hamiltonians \cite{Ber-17,Blu-21,Su-23a,Zha-25a}, while no experiment to date has reported observations of AQMBS. Bridging this gap between theory and experiment is challenging because many known constructions of QMBS rely either on local Hilbert spaces of dimension greater than two (such as spin-1 models including the $XY$ magnet \cite{Sch-19, Cha-20, Cha-23b, Doo-23a}) or on interaction terms that extend beyond simple two-body couplings (such as the multi-body constraints underlying the PXP-like models \cite{Tur-18a, Doo-22a, Zha-25a}). By contrast, the most widely available ingredients across current quantum simulation platforms are local Hilbert spaces of dimension two (i.e., qubits) with two-body interactions. An experimental protocol for realising a single conventional QMBS in such models has been proposed in \cite{Lar-26a}. 

In this work, we outline a general theoretical construction for models with QMBS and AQMBS as a deformation of a ``reference'' Hamiltonian. Using this construction, we introduce a model on a chain of qubits with experimentally realisable nearest-neighbour two-body interactions that hosts both QMBS and AQMBS. The model is parametrised by a complex variable $g$, and hosts a QMBS that is an edge-localised state: for $|g| \gg 1$ it is confined to the left boundary of the chain, for $|g| \ll 1$ it localises at the right boundary, and at the critical point $|g|=1$ it becomes fully delocalised across the chain. At this critical point we also observe AQMBS, whose existence can be directly linked to gapless excitations of the reference Hamiltonian, following a connection first identified in Ref. \cite{Kun-25}.
A key advantage of our model is its straightforward implementation on current digital quantum simulators. We demonstrate this explicitly using the Quantinuum H1-1 trapped-ion quantum computer and probe the dynamics of the total magnetisation starting from several classes of initial states: conventional QMBS, non-QMBS, asymptotic QMBS, and local edge excitations. As expected, non-QMBS initial states exhibit rapid thermalisation, while dynamics initiated in a conventional QMBS display a markedly slower decay that is essentially independent of system size. By contrast, when the system is initialised in an asymptotic QMBS, the relaxation time increases with system size as the initial state approaches a stationary eigenstate, providing the first experimental evidence for asymptotic QMBS.

\section{Results} \label{sec:results}

\subsection{General theoretical framework}
\label{subsec: general framework}

In this section, we present a framework for identifying models that host conventional as well as asymptotic QMBS. Consider a system of $N$ qudits (i.e., particles with Hilbert space dimension $d$), labelled $n \in \{0,1,\hdots,N-1 \}$, arranged on a $D$-dimensional lattice $\Lambda$. We also introduce a set of local interactions $\{ \hat{H}_{X_n} \}_{n=0}^{N-1}$, where each $\hat{H}_{X_n}$ is a Hermitian operator acting non-trivially only on a finite number of qudits in a local neighbourhood $X_n \subset \Lambda$ of the qudit $n$. These local interactions serve as the fundamental building block for both circuit and time-independent Hamiltonian models: in the circuit setting, they generate local unitary gates $\hat{U}_{X_n} = \exp(i \hat{H}_{X_n})$, which are composed into a many-body circuit unitary $\hat{\mathbb{U}}$, while in the Hamiltonian setting they define a many-body Hamiltonian $\hat{\mathbb{H}} = \sum_{n} \hat{H}_{X_n}$.

To construct models with QMBS and AQMBS it is convenient to restrict attention to interactions of the form: \begin{equation} \hat{H}_{X_n} = \hat{P}_{X_n} \hat{h}_{X_n} \hat{P}_{X_n} , \end{equation} where $\hat{P}_{X_n} \neq \hat{I}^{\otimes N}$ is a local projector and $\hat{h}_{X_n}$ is an arbitrary Hermitian operator, both acting non-trivially only on the qudits in $X_n$. We also define a local Hamiltonian: \begin{equation} \hat{\mathbb{H}}_+ = \sum_{n} \hat{P}_{X_n} , \end{equation} which we call the ``reference'' Hamiltonian. This allows us to write the following result, variants of which are widely used in the literature to construct models with conventional QMBS \cite{Shi-17,Mou-20a,Lar-24,Lar-26a}:

\begin{theorem} \label{theorem:QMBS}
  If the reference Hamiltonian $\hat{\mathbb{H}}_+ \equiv \sum_{n} \hat{P}_{X_n}$ is \emph{frustration-free}, that is, if a ground state $\ket{\mathcal{G}}$ is simultaneously a ground state of each individual projector $\hat{P}_{X_n}$, then $\ket{\mathcal{G}}$ remains an exact eigenstate of any corresponding circuit or Hamiltonian model constructed with the more general local interactions $\hat{H}_{X_n} = \hat{P}_{X_n} \hat{h}_{X_n} \hat{P}_{X_n}$.
\end{theorem}
The proof is straightforward: A ground state $\ket{\mathcal{G}}$ of the projectors obeys $\hat{P}_{X_n} \ket{\mathcal{G}} = 0$, which implies that it is also a ground state of the reference Hamiltonian ($\hat{\mathbb{H}}_{+} \ket{\mathcal{G}} = 0$), since it is non-negative ($\hat{\mathbb{H}}_+ \geq 0$). It follows that $\hat{H}_{X_n} \ket{\mathcal{G}} = \hat{P}_{X_n} \hat{h}_{X_n} \hat{P}_{X_n} \ket{\mathcal{G}} = 0$ for any choice of $\hat{h}_{X_n}$. Therefore, any many-body circuit or Hamiltonian constructed from the interactions $\hat{H}_{X_n}$ has $\ket{\mathcal{G}}$ as an eigenstate. In such cases, $\ket{\mathcal{G}}$ typically appears as a QMBS rather than as a ground state. In other words, as we deform the reference Hamiltonian $\hat{\mathbb{H}}_+$ into either the non-positive Hamiltonian $\hat{\mathbb{H}}$ or the circuit model $\hat{\mathbb{U}}$, the state $\ket{\mathcal{G}}$ is transformed from a ground state to a QMBS. 

Suppose that the reference Hamiltonian $\hat{\mathbb{H}}_+$ has gapless excitations, which become zero-energy ground states only asymptotically, in the thermodynamic limit. One might expect that such states become QMBS in the thermodynamic limit, i.e., are AQMBS, as was conjectured in Ref. \cite{Mou-24}. This intuition is supported by the following result, which we prove in the Supplemental Material (SM), Sec. \ref{app: general_proof_aqmbs} \footnote{We note that the general theoretical framework of Theorem 2 complements the results of Refs. \cite{Ren-24, Kun-25, Gio-26}, which provided alternative constructions of models hosting AQMBS.}:

\begin{restatable}{theorem}{aqmbsthm}
\label{theorem:AQMBS_main_text}
  If the reference Hamiltonian $\hat{\mathbb{H}}_+ = \sum_{n} \hat{P}_{X_n}$ has positive-energy eigenstates $\ket{\mathcal{A}}$ for finite $N$ that become zero-energy ground states in the thermodynamic limit $\hat{\mathbb{H}}_+ \ket{\mathcal{A}} \stackrel{N\to\infty}{\longrightarrow} 0$ (e.g., gapless excitations), then the corresponding finite-depth circuit or Hamiltonian model constructed from the interactions $\hat{H}_{X_n} = \hat{P}_{X_n} \hat{h}_{X_n} \hat{P}_{X_n}$ will host AQMBS $\ket{\mathcal{A}}$.
\end{restatable}
\noindent While the conventional QMBS state $\ket{\mathcal{G}}$ will be stationary for time evolution generated by the Hamiltonian $\hat{\mathbb{H}}$ or the circuit $\hat{\mathbb{U}}$ for \emph{any} system size $N$, the AQMBS state $\ket{\mathcal{A}}$ will thermalise for finite $N$, but will approach a stationary state as the system size $N$ increases. More precisely, since the energy variance of $\ket{\mathcal{A}}$ decreases with $N$, the fidelity relaxation time will increase due to the Mandelstam-Tamm bound \cite{Man-45,Got-23}. As a consequence, the thermalisation of local observables, such as the total magnetisation, will be suppressed as the system size increases, which is an experimentally observable signature of AQMBS.

The general results in Theorems \ref{theorem:QMBS} and \ref{theorem:AQMBS_main_text} suggest a concrete strategy for constructing models with QMBS or AQMBS, with a set of desired properties:
If one can find a reference Hamiltonian $\hat{\mathbb{H}}_+$ with the desired properties that is also:
\begin{enumerate}[label=(\roman*)]
\item Frustration-free, and/or
\item Gapless or asymptotically frustration-free, 
\end{enumerate}
then the corresponding circuit or Hamiltonian model also obeys those properties and hosts (i) QMBS and/or (ii) AQMBS. 
Specifically, for our goal of constructing a model that is amenable to experiment on existing digital quantum simulators, we would like to find a reference Hamiltonian for a qubit system ($d=2$) that is built exclusively from local two-qubit projector interactions $\hat{P}_{X_n}$. In the next section, we identify a reference Hamiltonian that satisfies these conditions.

\subsection{PVBS reference Hamiltonian} \label{subsec:PVBS}

In Ref. \cite{Bac-12}, Bachmann and Nachtergaele introduced a class of Hamiltonians on an $N$ qubit chain with open boundary conditions, called product vacua with boundary state (PVBS) models. An example is given by: \begin{equation} \hat{\mathbb{H}}^{(g)}_+ = \sum_{n=0}^{N-2} \hat{P}^{(g)}_{n,n+1}, \quad \hat{P}^{(g)} = \ket{11}\bra{11} + \ket{\psi^{(g)}} \bra{\psi^{(g)}} , \label{eq:two_qubit_projector} \end{equation} where $\hat{P}_{n,n+1}^{(g)}$ is two-qubit projector acting non-trivially only on qubits $n$ and $n+1$, and: \begin{equation} |\psi^{(g)}\rangle = \frac{1}{\sqrt{1+|g|^2}} \big(g^*|01\rangle - |10\rangle \big) , \end{equation} is a two-qubit state parameterised by a complex variable $g$. Here, $\{ \ket{0}, \ket{1} \}$ is the single-qubit computational basis. The PVBS reference Hamiltonian $\hat{\mathbb{H}}^{(g)}_+$ is frustration-free, since it has two zero-energy ground states, the \emph{product vacuum} state: \begin{equation} \ket{\mathcal{V}} = \ket{0}^{\otimes N} , \end{equation} and the \emph{boundary} state: \begin{equation} \ket{\mathcal{B}^{(g)}} = \frac{1}{\mathcal{N}^{(g)}} \sum_{n=0}^{N-1} g^{-n} \ket{n} , \label{eq:BQMBS} \end{equation} which give their names to the PVBS class of models \cite{Bac-12} and are both annihilated by every projector $\hat{P}_{n,n+1}^{(g)}$ (see SM, Sec. \ref{app: edge_localised_state} for details). In Eq. \ref{eq:BQMBS}, $\mathcal{N}^{(g)}$ is a normalisation factor and $\ket{n} \equiv \ket{0_0 0_1 \hdots 0_{n-1} 1_n 0_{n+1} \hdots 0_{N-1}}$ is the computational basis state in which a single ``excitation'' is located on site $n$ of the chain. For $|g| > 1$ the complex amplitudes in Eq. \ref{eq:BQMBS} are decreasing exponentially as $n$ increases, so that the excitation is exponentially localised at the left side of the chain, while for $|g| < 1$ it is exponentially localised at the right side of the chain [see Fig. \ref{fig:circuits_and_initial_states}(b)]. The value $|g|=1$ is a critical point, at which the excitation is delocalised across the entire system. In SM, Sec. \ref{app: edge_localised_state} we highlight other critical features of the state $\ket{\mathcal{B}^{(g)}}$, including its half-chain entanglement entropy $S = -p \log p - (1-p) \log (1-p)$, $p = 1/(1+|g^{-1}|^N)$, which is singular at $|g|=1$ in the thermodynamic limit $N \to \infty$.

The PVBS reference Hamiltonian $\hat{\mathbb{H}}_+^{(g)}$ is also known to be gapless at $|g|=1$ \cite{Bra-15}. This is because one can find a set of eigenstates $\hat{\mathbb{H}}^{(|g|=1)}_+ \ket{\mathcal{A}_k} = \varepsilon_k \ket{\mathcal{A}_k}$ given by (see SM, Sec. \ref{app: aqmbs_concrete_model}):
\begin{eqnarray}
    \ket{\mathcal{A}_k} &=& \frac{1}{\mathcal{N}_k} \sum_{n=0}^{N-1} e^{in\phi} \cos\left[\left(N-n-\frac{1}{2}\right)\frac{k \pi}{N} \right] \ket{n}
    \label{eq: aqmbs_amplitudes} \\
    \varepsilon_k &=& 1 - \cos\frac{k\pi}{N}, \label{eq: aqmbs_energies}
\end{eqnarray}
where $k \in \{ 0, 1,..., N - 1\}$, $g=e^{i\phi}$, and $\mathcal{N}_k$ is a normalisation factor. The lowest energy eigenstate ($k=0$) in this set corresponds to the boundary state $\ket{\mathcal{A}_0} =  \ket{\mathcal{B}^{(|g|= 1)}}$, which as mentioned above, is a ground state with eigenvalue $\varepsilon_0 = 0$. However, if $k > 0$ is held fixed as $N$ increases, the energy gap $\varepsilon_k$ closes, signalling gapless excitations of $\hat{\mathbb{H}}_+^{(|g|=1)}$ in the thermodynamic limit.

Since the PVBS reference Hamiltonian $\hat{\mathbb{H}}_+^{(g)}$ is frustration-free, it follows from Theorem \ref{theorem:QMBS} that any circuit model or time-independent Hamiltonian model based on the two-qubit interactions $\hat{H}_{n,n+1}^{(g)} = \hat{P}^{(g)}_{n,n+1} \hat{h}_{n,n+1} \hat{P}^{(g)}_{n,n+1}$ will have the product state $\ket{\mathcal{V}}$ and the boundary state $\ket{\mathcal{B}^{(g)}}$ as conventional QMBS. Likewise, since the PVBS reference Hamiltonian has the gapless excitations $\ket{\mathcal{A}_{k>0}}$, by Theorem \ref{theorem:AQMBS_main_text}, these are AQMBS of any Hamiltonian or finite-depth circuit constructed from the two-qubit interactions $\hat{H}_{n,n+1}^{(|g|=1)} = \hat{P}_{n,n+1}^{(|g|=1)} \hat{h}_{n,n+1} \hat{P}_{n,n+1}^{(|g|=1)}$. We present numerical evidence confirming the presence of QMBS and AQMS in such models in SM, Sec. \ref{app:additional_classical_results}. However, in the next section, based on the PVBS reference Hamiltonian, we construct a Floquet circuit model hosting QMBS and AQMBS.


\begin{figure*}
    \centering
    \includegraphics[width=\linewidth]{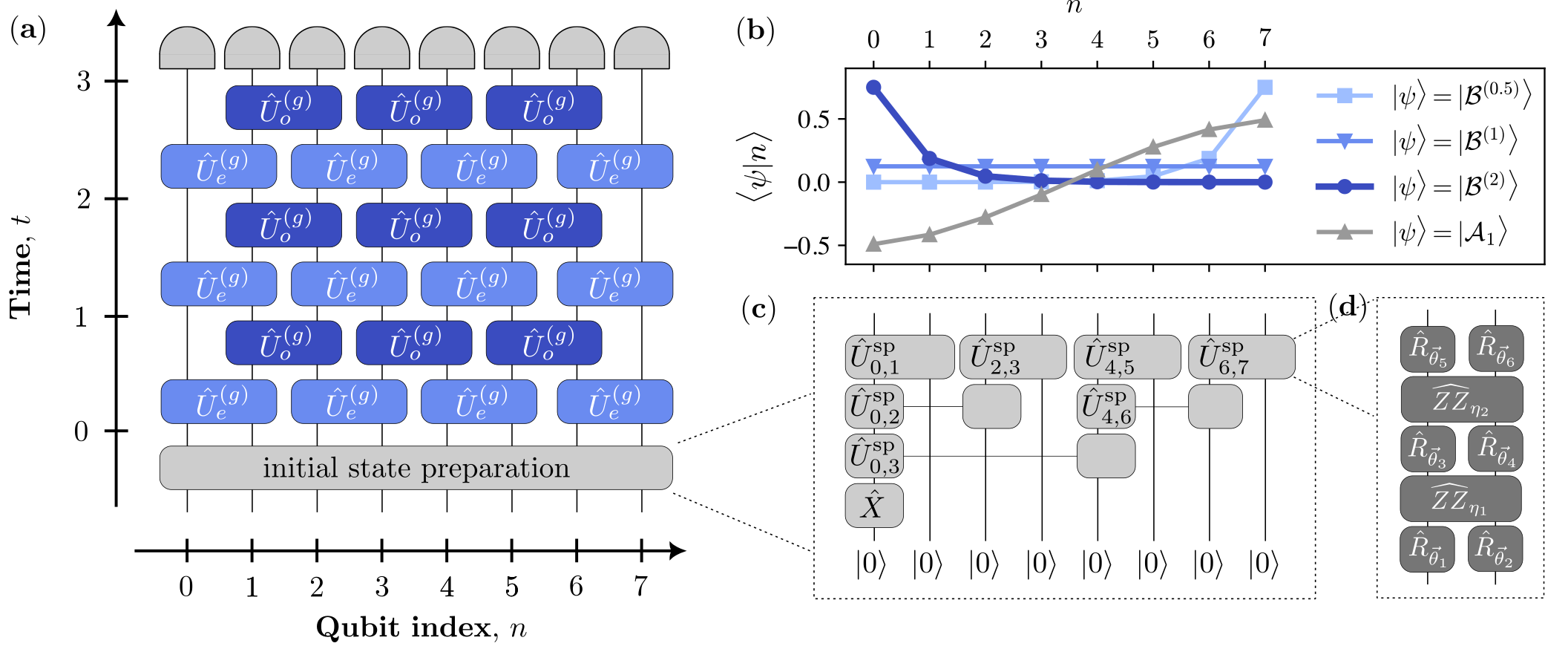}
    \caption{(a) The Floquet brickwork circuit of our experiment in the Quantinuum H1-1 processor. (b) The probability amplitudes of some examples of the boundary QMBS $\ket{\mathcal{B}^{(g)}}$ (blue lines) and the asymptotic QMBS $\ket{\mathcal{A}_1}$ (gray line), in the single-excitation basis $\{ \ket{n} \}_{n=0}^{N-1}$, $N=8$. (c) The circuit to prepare the asymptotic QMBS $\ket{\mathcal{A}_1}$ for $N=8$. (d) The decomposition of the two-qubit gates in terms of the single-qubit rotations $\hat{R}
    _{\vec{\theta}} = e^{i\theta^{(1)}\hat{Z}/2} e^{i\theta^{(2)}\hat{X}/2} e^{i\theta^{(3)}\hat{Z}/2}$ and the entangling gates $\widehat{ZZ}_\eta = e^{i\eta \hat{Z}\otimes \hat{Z}/2}$ that are native to the Quantinuum H1-1 device, where $\hat{Z} = \ket{0}\bra{0} - \ket{1}\bra{1}$. Appropriate choices of the gate parameters $\vec{\theta}_i$, $\eta_i$ give the required two-qubit gates $\hat{U}_{e/o}^{(g)}$ [in the brickwork circuit (a)] or $\hat{U}_{n,n'}^{\rm sp}$ [in the state preparation circuit (c)].}
    \label{fig:circuits_and_initial_states}
\end{figure*}

\subsection{Model and experimental setup} \label{subsec:model}

Although Theorems \ref{theorem:QMBS} and \ref{theorem:AQMBS_main_text} apply to both Hamiltonian and circuit models, throughout this work we focus primarily on circuit models because of their convenience for experimental implementation on available digital quantum simulators. Assuming that the number of qubits $N$ is even, we construct a time-periodic (Floquet) brickwork circuit, with a Floquet unitary $\hat{\mathbb{U}}^{(g)} = \hat{\mathbb{U}}_o^{(g)} \hat{\mathbb{U}}_e^{(g)}$ given by an ``even layer'' of two-qubit gates $\hat{\mathbb{U}}_e^{(g)} = [\hat{U}_e^{(g)}]^{\otimes N/2}$ followed by an ``odd layer'' $\hat{\mathbb{U}}_o^{(g)} = \hat{I} \otimes [\hat{U}_o^{(g)}]^{\otimes (N/2-1)} \otimes \hat{I}$ [see Fig. \ref{fig:circuits_and_initial_states}(a)]. Here, the two-qubit gates in the even and odd layers are given by $\hat{U}_e^{(g)} = \exp[i \hat{P}^{(g)} \hat{h}_e \hat{P}^{(g)}]$ and $\hat{U}_o^{(g)} = \exp[i \hat{P}^{(g)} \hat{h}_o \hat{P}^{(g)}]$, respectively, where $\hat{P}^{(g)}$ is the projector in Eq. \ref{eq:two_qubit_projector} and: \begin{eqnarray}  \hat{h}_{e}  &=& (1+2i)\ket{11}\bra{\psi^{(g)}} + (1-2i)\ket{\psi^{(g)}}\bra{11} ,\notag \\ \hat{h}_{o}  &=& \frac{6 + i\pi}{2} \ket{11}\bra{\psi^{(g)}} + \frac{6 - i\pi}{2} \ket{\psi^{(g)}}\bra{11}. \label{eq:h_e_h_o} \end{eqnarray}
We enforce open boundary conditions in the chain, which is necessary to observe the boundary QMBS and AQMBS. In SM Sec. \ref{app: level_spacing} we numerically compute level spacing statistics to verify that the Floquet unitary $\hat{\mathbb{U}}^{(g)}$ is non-integrable. We also numerically compute the half-system entanglement entropy of each eigenstate of the Floquet unitary, to show that the product vacuum state $\ket{\mathcal{V}}$ and the boundary state $\ket{\mathcal{B}^{(g)}}$ are low-entanglement outliers, i.e. exact QMBS, in contrast to the surrounding thermal eigenstates which obey a volume-law scaling of entanglement entropy.

We implement this circuit model on a trapped ion digital quantum computer -- the Quantinuum H1-1 processor \cite{quantinuum_h1_1_access, quantinuum_h1_datasheet_2023}. The two-qubit gates $\hat{U}_e^{(g)}$ and $\hat{U}_o^{(g)}$ that make up the even and odd layers are decomposed into single-qubit rotations and $\widehat{ZZ}_\eta = e^{i\frac{\eta}{2} \hat{Z} \otimes \hat{Z}}$, the native entangling operation of the Quantinuum H1-1 processor, as shown in Fig. \ref{fig:circuits_and_initial_states}(d). The most detrimental experimental errors during time evolution are due to the $\widehat{ZZ}_\eta$ gates. We note that a two-qubit unitary generally require 3 such entangling gates, but if $\hat{h} = \alpha \ket{11}\bra{\psi^{(g)}} + \alpha^* \ket{\psi^{(g)}}\bra{11}$ for $\alpha \in \mathbb{C}$, it can be decomposed with only 2. In order to minimise the accumulated errors, we therefore set the local generators as shown in Eq. \ref{eq:h_e_h_o}.

\begin{figure*}
    \centering
    \includegraphics[width=\linewidth]{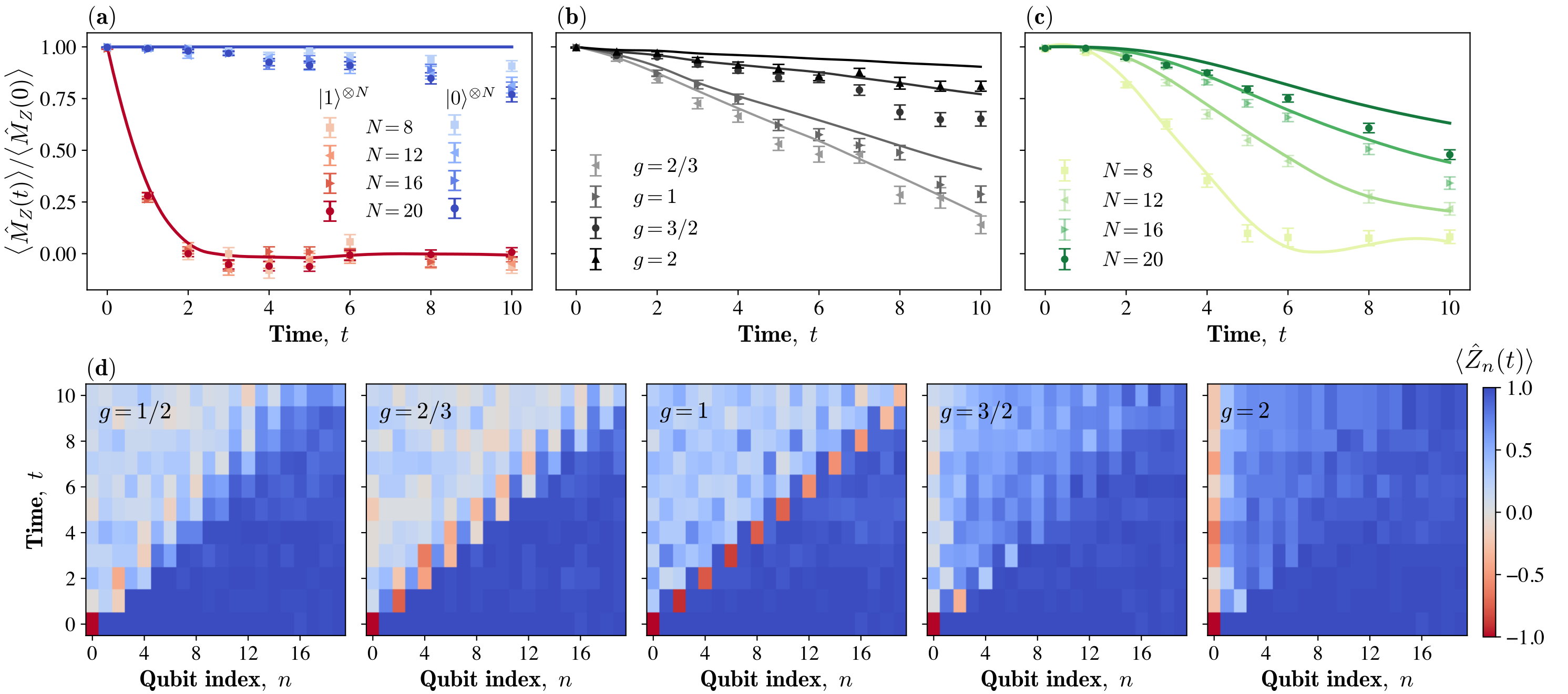}
    \caption{{\bf (a)} The time-evolved expectation value of the total magnetisation $\hat{M}_Z = \sum_{n=0}^{N-1}\hat{Z}_n$ (normalised by its initial value) for the exact QMBS initial state $\ket{\mathcal{V}} = \ket{0}^{\otimes N}$ (blue) and the non-QMBS initial state $\ket{1}^{\otimes N}$ (red). The non-QMBS initial state rapidly thermalises, whereas the exact QMBS initial state remains stationary, up to experimental imperfections. {\bf (b)} The evolution of the total magnetisation for the initial state $\ket{L}$, which has a large overlap with the boundary QMBS $\ket{\mathcal{B}^{(g)}}$ for $|g| \gg 1$. All of the results displayed in \textbf{(b)} are for simulations at system size $N = 20$. As $|g|$ increases, the overlap between $\ket{L}$ and $\ket{\mathcal{B}^{(g)}}$ grows, resulting in increasingly slow decay of the total magnetisation. {\bf (c)} The evolution of the total magnetisation for the asymptotic QMBS $\ket{\mathcal{A}_1}$ as the initial state. As the system size $N$ increases, $\ket{\mathcal{A}_1}$ the total magnetisation decays more slowly, a direct consequence of $\ket{\mathcal{A}_1}$ becoming more stable with increased system size. {\bf (d)} The spatially resolved magnetisation for the initial state $\ket{L}$. As $|g|$ increases, $\ket{L}$ approaches a stationary state.}
    \label{fig:exp_data}
\end{figure*}

\subsection{Experimental results}

First, we would like to experimentally observe the product QMBS $\ket{\mathcal{V}} = \ket{0}^{\otimes N}$ on the trapped-ion simulator. Since it is an eigenstate of our Floquet unitary, in principle it is a stationary state of the dynamics and any observable evolving from that initial state will be stationary, thus failing to thermalise. The product QMBS $\ket{\mathcal{V}}$ is easily prepared as an initial state on the trapped-ion device. In Fig. \ref{fig:exp_data}(a), we plot the experimental data, corresponding to evolution of the total magnetization $\sum_n \langle \hat{Z}_n(t) \rangle$ for this initial state, as a function of the number of Floquet cycles $t$ for system sizes $N= 8, 12, 16, 20$. For comparison, we also show results for the initial product state $\ket{1}^{\otimes N}$, which rapidly decays towards $\sum_n \langle \hat{Z}_n (t) \rangle \approx 0$, consistent with each local subsystem thermalising to its infinite-temperature density matrix. The very slow decay from the QMBS initial state, compared to the rapid thermalisation from the non-QMBS initial state is a clear signature of the conventional QMBS $\ket{\mathcal{V}} = \ket{0}^{\otimes N}$, which persists in spite of the experimental errors. We note that increasing the system size leads to greater accumulated errors, primarily due to the increased number of entangling gates, which is observed in the slightly increased rate of decay of the magnetisation.

Next, we would like to observe dynamical signatures of the boundary QMBS $\ket{\mathcal{B}^{(g)}}$ on the trapped-ion simulator. Instead of preparing the boundary QMBS as an initial state, we prepare $\ket{L} = \ket{1} \otimes \ket{0}^{\otimes (N-1)}$, which is much easier experimentally and has a large overlap with the boundary QMBS when $|g| \gg 1$ and a small overlap when $|g| \lesssim 1$. Time-evolution starting from the initial state $\ket{L}$ should therefore have a large component that is approximately stationary when $|g| \gg 1$, with the approximation improving as $|g|$ increases. We therefore investigate how the site-resolved and total magnetisation evolve for different $g$ values when evolving from the initial state $\ket{L}$, with the results shown in Fig.  \ref{fig:exp_data}(b,d). We probe the dynamics for $N = 20$ qubits, $g \in \{\frac{2}{3}, \, 1, \, \frac{3}{2}, \, 2 \}$, up to $t = 10$ time-steps. Fig. \ref{fig:exp_data}(d) shows the experimental data corresponding to the local expectation values $\langle \hat{Z}_n(t)\rangle$ as they evolve in space $n$ and time $t$. For $|g| \lesssim 1$ the excitation, initially localised at $n=0$, spreads rapidly through the system. However, for $g = 2$ the excitation remains well localised, due to its large overlap with the boundary QMBS. Interestingly, at $g=1$, where the initial state has a small overlap $|\langle L | \mathcal{B}^{(g=1)} \rangle | = 1/\sqrt{N}$ with the completely delocalised QMBS, we observe a strong signal along the ``light-cone'', showing transport of the spin excitation close to the maximum possible velocity. In Fig. \ref{fig:exp_data}(b) we also show the evolution of the total magnetization $\sum_n \langle \hat{Z}_n(t) \rangle$. The decay of the observable slows as $g$ is increased, as a result of the increasing overlap with the stationary boundary QMBS. Taken together, the experimental results in Fig. \ref{fig:exp_data}(b,d) are strong evidence of the presence of the boundary QMBS in the Floquet circuit model.

Finally, to experimentally explore AQMBS in our model, we would like to be able to prepare $\ket{\mathcal{A}_k}$ as an initial state on the trapped-ion device. While the product QMBS $\ket{\mathcal{V}}$ and state $\ket{L}$ used to probe the boundary QMBS can be constructed almost trivially on any digital simulator, the preparation of the asymptotic QMBS requires a more complex state preparation procedure. Fortunately, because they are constrained to the single excitation subspace, $\ket{\mathcal{A}_k}$ can be prepared in logarithmic circuit depth on devices with all-to-all connectivity. The state preparation protocol, for $N = 8$, is shown in Fig. \ref{fig:circuits_and_initial_states}(c). This construction can be easily generalised to any $N = 2^k$, requiring $\sim N -1$ entangling gates and a depth of $\sim \log _2 N = k$. For $N \neq 2^k$, the state preparation circuit depth will be $\sim \left \lceil{\log _2 N}\right \rceil $. These state preparation circuits are described in more detail in SM, Sec. \ref{app: state_prep}, both for general $N$ as well as the specific circuits deployed in our experiments.

In our experiment, we prepare the asymptotic QMBS $\ket{\mathcal{A}_1}$, and observe the subsequent dynamics for a range of system sizes $N \in \{8, 12, 16, 20 \}$, up to $t = 10$ time-steps. The total magnetisation is shown in Fig. \ref{fig:exp_data}(c).
For $N = 8, \; 12$ we see strong agreement between the classical simulation and experimental results. For greater system sizes, $N = 16, \; 20$, we see a more noticeable deviation due to greater accumulation of errors, as already observed for the $\ket{\mathcal{V}}$ state. Despite the experimental errors, one can still observe slower thermalisation as system size is increased, a distinct feature of AQMBS states. This should be contrasted with the lack of thermalisation for $\ket{\mathcal{V}}$ and the rapid decay in the case of $\ket{1}^{\otimes N}$, both of which are independent of system size. As $N$ is increased, the state $\ket{\mathcal{A}_1}$ approaches a stationary state. This is, to the best of our knowledge, the first direct experimental evidence for asymptotic QMBS.

\section{Conclusion} \label{sec:conclusion}

To summarise, we first showed a direct relationship between gapless excitations in a frustration-free (or, alternatively, asymptotically frustration-free) reference Hamiltonian and the existence of AQMBS states in a related family of Hamiltonian and circuit models. Based on this, we constructed a 2-local model hosting both exact and asymptotic QMBS states. One of the QMBS states displays edge localisation, and becomes fully delocalised over the system at $|g|= 1$. At this critical value, we also observe AQMBS states, corresponding to gapless excitations in the reference Hamiltonian where the exact QMBS states are embedded as ground states. 

The specific features of the model were deliberately constructed to make digital quantum simulation feasible. We exploited the freedom in the 2-local unitaries to find efficient gate decompositions, 
and the structure of the probed AQMBS state $\ket{\mathcal{A}_1}$ to prepare it in logarithmic depth. Performing simulations starting from $\ket{\mathcal{A}_1}$ at various system sizes on the Quantinuum H1-1 trapped ion simulator, we observe longer thermalisation timescales as $N$ is increased. This is the first direct experimental detection of anomalous thermalisation due to AQMBS states and one of the rare occasions where a new many-body phenomenon is observed for the first time in a digital simulator, rather than on an analogue platform. Furthermore, we also observed direct evidence of the product scar $\ket{\mathcal{V}}$ and the the edge localised scar state $\ket{\mathcal{B}^{(g)}}$. 

This work opens up several research pathways. We have considered a relatively simple model, with solely nearest neighbour interactions in one dimension, in order to facilitate the digital quantum simulation. However, we expect that the same arguments could be generalised to higher order interactions of similar form, as well as to two or three spatial dimensions \cite{Bac-15}, giving rise to a wide family of experimentally amenable models that could potentially host AQMBS states. Furthermore, for some sets of local generators we can observe an exponential degeneracy in the middle of the spectrum, which might include  additional QMBS states, as found in the spin-$1$ $XY$ model \cite{Moh-25}. On the experimental side, with more resources, it would also be possible and of great interest to investigate the fidelity decay or growth of Renyi entropies from AQMBS via randomised measurement techniques \cite{Fla-11, Bry-19}. Finally, we have demonstrated the direct link between AQMBS states and gapless excitations in a relatively simple reference Hamiltonian. This connection may also be exploited in other candidate reference Hamiltonians such as the ferromagnetic Heisenberg model, Fredkin and Motzkin spin chains \cite{Sal-16, Mov-16, Mov-18}, and the Rokhar-Kivelson Hamiltonian \cite{Rok-88, Fel-26-arxiv}.

\section{Methods} \label{sec: Methods}

\subsection{Ion Trap Processor Details}

The device specifications of the Quantinuum H1-1 device are available online at Ref. \cite{quantinuum_h1_datasheet_2023}. This quantum processing unit is based on trapped ions, and supports simulations on up to $20$ qubits. The native gate set is based two parameterised single qubit gates:
\begin{eqnarray}
    \hat{U}_{1q}(\theta, \; \phi) &=& \exp \left[-i \frac{\theta}{2} (\cos (\phi)\hat{X} +  \sin (\phi)\hat{Y}) \right], \\
    \hat{R}_z(\lambda) &=& \exp \left[-i \frac{\lambda}{2} \hat{Z} \right],
\end{eqnarray}
and one entangling 2-qubit $ZZ$ gate:
\begin{equation}
    \widehat{ZZ}_\eta = \exp \left[ -i \frac{\eta}{2} \hat{Z} \otimes \hat{Z} \right] .
\end{equation}

The device supports all-to-all connectivity, allowing the execution of up to five $2$-qubit gates in parallel. 

The typical single qubit and $2$-qubit gate error rates are $1 \times 10^{-5}$ and $2 \times 10^{-5}$, respectively. The typical state-preparation and measurement (SPAM) errors are $3 \times 10^{-5}$. For relatively shallow circuits, the SPAM errors are the dominant sources of error, while for greater circuit depth the $2$-qubit gates infidelities become the prohibitive factor.

\subsection{Additional Experimental Details}

In order to execute the model specific $2$-qubit generators $\hat{h}_{e}, \,\hat{h}_{o}$ as written in Eq. ~\eqref{eq:h_e_h_o} we use the gate decomposition depicted in Fig. \ref{fig:circuits_and_initial_states}. In Table \ref{tab: gate_params} we present the gate decomposition parameters for the five different values of $g$ used in the experiments.

\begin{table*}[]
    \centering
    \begin{tabular}{|c|c|c|c|c|c|}
    \hline 
    $g$ & $0.5$ & $\frac{2}{3}$ & $1.0$ & $1.5$ & $2.0$ \\ \hline
    
    $\eta_1^e$, $\eta_2^e$ & 1.2726, 0.82774   & 1.4509, 1.0444 & 1.3782, 1.3782 & 1.4509, 1.0444 & 1.2726, 0.8277 \\ \hline

    $\vec{\theta}_1^e$ & 3.6052, $\pi / 2$, $1.0727$ & 3.6052, $\pi / 2$, 4.1838 & 1.1528, 1.1584, 2.0236 & 2.0344, 0.8508, $3 \pi / 2$ & 2.0344, 0.7434, $3 \pi / 2$ \\ \hline

    $\vec{\theta}_2^e$ & 5.1758, $0.7434$, $\pi/2$ & 5.1758, 0.8508, $3\pi / 2$ & 4.2948, 1.1584, 5.1648 & 0.46365, $\pi / 2$, 4.1833 & 0.4636, $\pi / 2$, 4.2143 \\ \hline

    $\vec{\theta}_3^e$ & $0$, $\pi / 2$, $\pi / 2$ & $0$, $\pi / 2$, $\pi / 2$ & $0$, $\pi / 2$, $\pi / 2$ & $0$, $\pi / 2$, $\pi / 2$ & $0$, $\pi / 2$, $\pi / 2$ \\ \hline

    $\vec{\theta}_4^e$ & $0$, $\pi / 2$, $\pi / 2$ & $0$, $\pi / 2$, $\pi / 2$ & $0$, $\pi / 2$, $\pi / 2$ & $0$, $\pi / 2$, $\pi / 2$ & $0$, $\pi / 2$, $\pi / 2$ \\ \hline

    $\vec{\theta}_5^e$ & $0$, $1.0727$, 4.2488 & 0, 4.1833, 4.2488 & 2.6888, 1.1584, 5.1308 & 2.4216, $\pi / 2$, 5.8198 & 2.3142, $\pi / 2$, 5.8198 \\ \hline

    $\vec{\theta}_6^e$ & $0.8274$, $\pi / 2$, 5.8198 & 2.4216, $\pi / 2$, 2.6779 & 0.45276, 1.9832, 5.1304 & 0, 4.1838, 1.1071 & 0, 4.2138, 1.1071 \\ \hline

     $\eta_1^o$, $\eta_2^o$ & 0.95678, 0.91982  & 1.1983, 1.166 & 1.5559, 1.5559 & 1.1983, 1.166 & 0.9568, 0.9198  \\ \hline

    $\vec{\theta}_1^o$ & 1.0884, $\pi/2$, 1.4341 & 1.0884, $\pi/2$, 1.424 & 5.1088, 1.4377, 1.6084 & 5.8008, 1.3526, $\pi / 2$ & 5.8008, 1.3024, $3\pi / 2$ \\ \hline

    $\vec{\theta}_2^o$ & 2.6592, 1.3024, $\pi / 2$ & 2.6592, 1.3526, $\pi / 2$ & 1.766, 1.4622, 4.8458 & 4.23, $\pi / 2$, 1.424 & 4.2298, $\pi / 2$, 4.5758 \\ \hline

    $\vec{\theta}_3^o$ & $0$, $\pi / 2$, $\pi / 2$ & $0$, $\pi / 2$, $\pi / 2$ & $0$, $\pi / 2$, $\pi / 2$ & $0$, $\pi / 2$, $\pi / 2$ & $0$, $\pi / 2$, $\pi / 2$ \\ \hline

    $\vec{\theta}_4^o$ & $0$, $\pi / 2$, $\pi / 2$ & $0$, $\pi / 2$, $\pi / 2$ & $0$, $\pi / 2$, $\pi / 2$ & $0$, $\pi / 2$, $\pi / 2$ & $0$, $\pi / 2$, $\pi / 2$ \\ \hline
    
    $\vec{\theta}_5^o$ & 0, 1.4341, 0.4823 & 0, 1.424, 0.48235 &  3.0077, 1.4622, 1.3756 & 0.2182, $\pi / 2$, 5.1947  & 2.8732, $\pi / 2$, 2.0531 \\ \hline

    $\vec{\theta}_6^o$ & 0.2683, $\pi / 2$ & 0.2182, $\pi / 2$, 2.0531 & 0.10959, 1.7039, 1.1745 & 0, 1.424, 3.6239 & 0, 4.5757, 3.6239 \\ \hline

    \end{tabular}
    \caption{Gate parameters for the decomposition depicted in Fig. \ref{fig:circuits_and_initial_states} \textbf{d}, for the used values of $g$. Note that the superscripts $e/o$ are used to distinguish even and odd gates.}
    \label{tab: gate_params}
\end{table*}

Note that, when executing the job on the device, the single qubit operations from the different unitary blocks are converted to the native gate-set basis. There is also a small compilation error that arises due to this decomposition, as well as the finite precision of rotation angles.

Whenever performing simulations for $N = 8$ and $N = 12$ (for any of the states $\{\ket{\mathcal{V}}, \ket{1}^{\otimes N}, \ket{\mathcal{A}_1} \}$), we run them in parallel, as the system allows for the simultaneous usage of up to $20$ qubits.

By default, the qubit register is always initialised in $|0\rangle ^{\otimes N }$. Hence, when performing simulations for the QMBS state $|0\rangle ^{\otimes N }$, we do not have to perform any state preparation prior to the time evolution. The probed non-scar state $|1\rangle ^{\otimes N }$ can be prepared in a straightforward way, by applying the $\hat{X}$ (which is equivalent, up to a global phase, to $U_{1q}(\theta = \pi, \phi = 0)$) operation on all of the qubits. Similarly, in order to prepare $\ket{L}=  \ket{1} \otimes \ket{0}^{\otimes N- 1}$, we apply $\hat{X}$ on the first qubit in the chain.


Finally, in order to minimise the costs associated with performing the experiment, we have reused some data points used at $T = 0$. Since the circuit constructions are the same, we reused data for the following experimental data points:
\begin{enumerate}
    \item For the $\ket{0}^{\otimes N}$ state the data for $N = 20$ is the same as already used for $N = 8$ and $12$ (which were run in parallel), at $T = 0$
    \item For the $\ket{1}^{\otimes N}$ state the data for $N = 20$ is the same as already used for $N = 8$ and $12$ (which were run in parallel), at $T = 0$
    \item For the $\ket{1} \otimes \ket{0}^{\otimes N - 1}$ state we reused the data for all $g$ values at $T = 0$.
    
\end{enumerate}

\subsection{Classical Methods}

The results presented in the main text, obtained via classical methods, were computed either by exact brute force methods implemented in Python's numpy library \cite{harris2020numpy}, for system sizes up to $N = 12$, whereas for greater system sizes, the Julia package ITensor was used \cite{itensor}.

\newpage
\clearpage

\begin{acknowledgments}
J.G. and L.L. acknowledge financial support by Microsoft Ireland. S.D. was supported by Taighde \'{E}ireann – Research Ireland under Grant number 22/PATH-S/10812. J. G. thanks the Royal Society and Research Ireland for funding. All of the authors thank Microsoft Ireland for granting access to quantum simulators via the Microsoft Azure platform. We acknowledge helpful conversations with Federica Maria Surace, Nathan Keenan, and Sanjay Moudgalya.

\end{acknowledgments}

\bibliography{refs.bib}

@article{Doo-21a,
	author = {Dooley, Shane},
	doi = {10.1103/PRXQuantum.2.020330},
	issue = {2},
	journal = {PRX Quantum},
	month = {May},
	numpages = {12},
	pages = {020330},
	publisher = {American Physical Society},
	title = {Robust Quantum Sensing in Strongly Interacting Systems with Many-Body Scars},
	url = {https://link.aps.org/doi/10.1103/PRXQuantum.2.020330},
	volume = {2},
	year = {2021}}

@article{Doo-23a,
	author = {Dooley, Shane and Pappalardi, Silvia and Goold, John},
	doi = {10.1103/PhysRevB.107.035123},
	issue = {3},
	journal = {Phys. Rev. B},
	month = {Jan},
	numpages = {7},
	pages = {035123},
	publisher = {American Physical Society},
	title = {Entanglement enhanced metrology with quantum many-body scars},
	url = {https://link.aps.org/doi/10.1103/PhysRevB.107.035123},
	volume = {107},
	year = {2023}}

@article{Doo-22a,
	author = {Dooley, Shane and Kells, Graham},
	doi = {10.1103/PhysRevB.105.155127},
	issue = {15},
	journal = {Phys. Rev. B},
	month = {Apr},
	numpages = {9},
	pages = {155127},
	publisher = {American Physical Society},
	title = {Extreme many-body scarring in a quantum spin chain via weak dynamical constraints},
	url = {https://link.aps.org/doi/10.1103/PhysRevB.105.155127},
	volume = {105},
	year = {2022}}

@article{Mou-24,
  title = {Symmetries as Ground States of Local Superoperators: Hydrodynamic Implications},
  author = {Moudgalya, Sanjay and Motrunich, Olexei I.},
  journal = {PRX Quantum},
  volume = {5},
  issue = {4},
  pages = {040330},
  numpages = {41},
  year = {2024},
  month = {Nov},
  publisher = {American Physical Society},
  doi = {10.1103/PRXQuantum.5.040330},
  url = {https://link.aps.org/doi/10.1103/PRXQuantum.5.040330}
}

@misc{Fel-26-arxiv,
      title={Digital dissipative state preparation for frustration-free gapless quantum systems}, 
      author={Johannes Feldmeier and Yu-Jie Liu and Mikhail D. Lukin and Soonwon Choi},
      year={2026},
      eprint={2603.10119},
      archivePrefix={arXiv},
      primaryClass={quant-ph},
      url={https://arxiv.org/abs/2603.10119}, 
}

@article{Rok-88,
  title = {Superconductivity and the Quantum Hard-Core Dimer Gas},
  author = {Rokhsar, Daniel S. and Kivelson, Steven A.},
  journal = {Phys. Rev. Lett.},
  volume = {61},
  issue = {20},
  pages = {2376--2379},
  numpages = {0},
  year = {1988},
  month = {Nov},
  publisher = {American Physical Society},
  doi = {10.1103/PhysRevLett.61.2376},
  url = {https://link.aps.org/doi/10.1103/PhysRevLett.61.2376}
}

@article{Mov-16,
    author = {Ramis Movassagh  and Peter W. Shor },
    title = {Supercritical entanglement in local systems: Counterexample to the area law for quantum matter},
    journal = {Proceedings of the National Academy of Sciences},
    volume = {113},
    number = {47},
    pages = {13278-13282},
    year = {2016},
    doi = {10.1073/pnas.1605716113},
    URL = {https://www.pnas.org/doi/abs/10.1073/pnas.1605716113},
    eprint = {https://www.pnas.org/doi/pdf/10.1073/pnas.1605716113}}

@misc{Sal-16,
      title={Fredkin Spin Chain}, 
      author={Olof Salberger and Vladimir Korepin},
      year={2016},
      eprint={1605.03842},
      archivePrefix={arXiv},
      primaryClass={quant-ph},
      url={https://arxiv.org/abs/1605.03842}, 
}

@article{Mov-18,
   title={The gap of Fredkin quantum spin chain is polynomially small},
   volume={3},
   ISSN={2380-2898},
   url={http://dx.doi.org/10.4310/AMSA.2018.v3.n2.a5},
   DOI={10.4310/amsa.2018.v3.n2.a5},
   number={2},
   journal={Annals of Mathematical Sciences and Applications},
   publisher={International Press of Boston},
   author={Movassagh, Ramis},
   year={2018},
   pages={531–562} }

@article{Mou-20a,
	author = {Moudgalya, Sanjay and O'Brien, Edward and Bernevig, B. Andrei and Fendley, Paul and Regnault, Nicolas},
	doi = {10.1103/PhysRevB.102.085120},
	issue = {8},
	journal = {Phys. Rev. B},
	month = {Aug},
	numpages = {19},
	pages = {085120},
	publisher = {American Physical Society},
	title = {Large classes of quantum scarred Hamiltonians from matrix product states},
	url = {https://link.aps.org/doi/10.1103/PhysRevB.102.085120},
	volume = {102},
	year = {2020}}

@article{Fau-24a,
	author = {Fauseweh, Benedikt},
	date = {2024/03/08},
	doi = {10.1038/s41467-024-46402-9},
	id = {Fauseweh2024},
	isbn = {2041-1723},
	journal = {Nature Communications},
	number = {1},
	pages = {2123},
	title = {Quantum many-body simulations on digital quantum computers: State-of-the-art and future challenges},
	url = {https://doi.org/10.1038/s41467-024-46402-9},
	volume = {15},
	year = {2024}}

@misc{Yan-25a-arxiv,
	archiveprefix = {arXiv},
	author = {Zhiguang Yan and Zi-Yong Ge and Rui Li and Yu-Ran Zhang and Franco Nori and Yasunobu Nakamura},
	eprint = {2506.21061},
	primaryclass = {quant-ph},
	title = {Characterizing Many-body Dynamics with Projected Ensembles on a Superconducting Quantum Processor},
	url = {https://arxiv.org/abs/2506.21061},
	year = {2025}}

@article{Fis-26a,
	author = {Fischer, Laurin E. and Leahy, Matea and Eddins, Andrew and Keenan, Nathan and Ferracin, Davide and Rossi, Matteo A. C. and Kim, Youngseok and He, Andre and Pietracaprina, Francesca and Sokolov, Boris and Dooley, Shane and Zimbor{\'a}s, Zolt{\'a}n and Tacchino, Francesco and Maniscalco, Sabrina and Goold, John and Garc{\'\i}a-P{\'e}rez, Guillermo and Tavernelli, Ivano and Kandala, Abhinav and Filippov, Sergey N.},
	date = {2026/01/20},
	id = {Fischer2026},
	isbn = {1745-2481},
	journal = {Nature Physics},
	title = {Dynamical simulations of many-body quantum chaos on a quantum computer},
	url = {https://doi.org/10.1038/s41567-025-03144-9},
	year = {2026},
    }

@article{Doo-20,
  title = {Simulating quantum circuits by adiabatic computation: Improved spectral gap bounds},
  author = {Dooley, Shane and Kells, Graham and Katsura, Hosho and Dorlas, Tony C.},
  journal = {Phys. Rev. A},
  volume = {101},
  issue = {4},
  pages = {042302},
  numpages = {7},
  year = {2020},
  month = {Apr},
  publisher = {American Physical Society},
  doi = {10.1103/PhysRevA.101.042302},
  url = {https://link.aps.org/doi/10.1103/PhysRevA.101.042302}
}

@misc{Fos-24a-arxiv,
	archiveprefix = {arXiv},
	author = {Michael Foss-Feig and Guido Pagano and Andrew C. Potter and Norman Y. Yao},
	eprint = {2409.02990},
	primaryclass = {quant-ph},
	title = {Progress in Trapped-Ion Quantum Simulation},
	url = {https://arxiv.org/abs/2409.02990},
	year = {2024}}

@article{Kja-20a,
	author = {Kjaergaard, Morten and Schwartz, Mollie E. and Braum{\"u}ller, Jochen and Krantz, Philip and Wang, Joel I.-J. and Gustavsson, Simon and Oliver, William D.},
	doi = {https://doi.org/10.1146/annurev-conmatphys-031119-050605},
	issn = {1947-5462},
	journal = {Annual Review of Condensed Matter Physics},
	number = {Volume 11, 2020},
	pages = {369-395},
	publisher = {Annual Reviews},
	title = {Superconducting Qubits: Current State of Play},
	type = {Journal Article},
	url = {https://www.annualreviews.org/content/journals/10.1146/annurev-conmatphys-031119-050605},
	volume = {11},
	year = {2020}}

@article{Ada-19a,
	author = {Adams, C S and Pritchard, J D and Shaffer, J P},
	doi = {10.1088/1361-6455/ab52ef},
	journal = {Journal of Physics B: Atomic, Molecular and Optical Physics},
	month = {dec},
	number = {1},
	pages = {012002},
	publisher = {IOP Publishing},
	title = {Rydberg atom quantum technologies},
	url = {https://doi.org/10.1088/1361-6455/ab52ef},
	volume = {53},
	year = {2019}}

@article{Pan-20,
   title={Quantum East Model: Localization, Nonthermal Eigenstates, and Slow Dynamics},
   volume={10},
   ISSN={2160-3308},
   url={http://dx.doi.org/10.1103/PhysRevX.10.021051},
   number={2},
   journal={Physical Review X},
   publisher={American Physical Society (APS)},
   author={Pancotti, Nicola and Giudice, Giacomo and Cirac, J. Ignacio and Garrahan, Juan P. and Bañuls, Mari Carmen},
   year={2020},
   month=jun}

@article{Hor-15,
   title={Dynamics of many-body localization in a translation-invariant quantum glass model},
   volume={92},
   ISSN={1550-235X},
   url={http://dx.doi.org/10.1103/PhysRevB.92.100305},
   number={10},
   journal={Physical Review B},
   publisher={American Physical Society (APS)},
   author={van Horssen, Merlijn and Levi, Emanuele and Garrahan, Juan P.},
   year={2015},
   month=sep }

@article{Bri-24,
   title={Anomalous transport in the kinetically constrained quantum East-West model},
   volume={110},
   ISSN={2469-9969},
   url={http://dx.doi.org/10.1103/PhysRevB.110.L100304},
   number={10},
   journal={Physical Review B},
   publisher={American Physical Society (APS)},
   author={Brighi, Pietro and Ljubotina, Marko},
   year={2024},
   month=sep }

@article{harris2020numpy,
 title         = {Array programming with {NumPy}},
 author        = {Charles R. Harris and K. Jarrod Millman and St{\'{e}}fan J.
                 van der Walt and Ralf Gommers and Pauli Virtanen and David
                 Cournapeau and Eric Wieser and Julian Taylor and Sebastian
                 Berg and Nathaniel J. Smith and Robert Kern and Matti Picus
                 and Stephan Hoyer and Marten H. van Kerkwijk and Matthew
                 Brett and Allan Haldane and Jaime Fern{\'{a}}ndez del
                 R{\'{i}}o and Mark Wiebe and Pearu Peterson and Pierre
                 G{\'{e}}rard-Marchant and Kevin Sheppard and Tyler Reddy and
                 Warren Weckesser and Hameer Abbasi and Christoph Gohlke and
                 Travis E. Oliphant},
 year          = {2020},
 month         = sep,
 journal       = {Nature},
 volume        = {585},
 number        = {7825},
 pages         = {357--362},
 doi           = {10.1038/s41586-020-2649-2},
 publisher     = {Springer Science and Business Media {LLC}},
 url           = {https://doi.org/10.1038/s41586-020-2649-2}
}

@article{Cha-20,
   title={Quantum many-body scars from virtual entangled pairs},
   volume={101},
   ISSN={2469-9969},
   url={http://dx.doi.org/10.1103/PhysRevB.101.174308},
   number={17},
   journal={Physical Review B},
   publisher={American Physical Society (APS)},
   author={Chattopadhyay, Sambuddha and Pichler, Hannes and Lukin, Mikhail D. and Ho, Wen Wei},
   year={2020},
   month=may }

@article{Cha-23b,
	author = {Chandran, Anushya and Iadecola, Thomas and Khemani, Vedika and Moessner, Roderich},
	journal = {Annual Review of Condensed Matter Physics},
	number = {1},
	pages = {443-469},
	title = {Quantum Many-Body Scars: A Quasiparticle Perspective},
	url = {https://doi.org/10.1146/annurev-conmatphys-031620-101617},
	volume = {14},
	year = {2023},
	}

@article{Mou-22a,
	author = {Sanjay Moudgalya and B Andrei Bernevig and Nicolas Regnault},
	doi = {10.1088/1361-6633/ac73a0},
	journal = {Reports on Progress in Physics},
	month = {jul},
	number = {8},
	pages = {086501},
	publisher = {IOP Publishing},
	title = {Quantum many-body scars and Hilbert space fragmentation: a review of exact results},
	url = {https://dx.doi.org/10.1088/1361-6633/ac73a0},
	volume = {85},
	year = {2022}}

@article{Ber-17,
	author = {Bernien, Hannes and Schwartz, Sylvain and Keesling, Alexander and Levine, Harry and Omran, Ahmed and Pichler, Hannes and Choi, Soonwon and Zibrov, Alexander S. and Endres, Manuel and Greiner, Markus and Vuleti{\'c}, Vladan and Lukin, Mikhail D.},
	date = {2017/11/29/online},
	day = {29},
	journal = {Nature},
	l3 = {10.1038/nature24622;},
	m3 = {Article},
	month = {11},
	pages = {579 EP -},
	publisher = {Macmillan Publishers Limited, part of Springer Nature. All rights reserved. SN -},
	title = {Probing many-body dynamics on a 51-atom quantum simulator},
	ty = {JOUR},
	url = {https://doi.org/10.1038/nature24622},
	volume = {551},
	year = {2017}}

@article{itensor,
	title={{The ITensor Software Library for Tensor Network Calculations}},
	author={Matthew Fishman and Steven R. White and E. Miles Stoudenmire},
	journal={SciPost Phys. Codebases},
	pages={4},
	year={2022},
	publisher={SciPost},
	doi={10.21468/SciPostPhysCodeb.4},
	url={https://scipost.org/10.21468/SciPostPhysCodeb.4}
}

@article{Blu-21,
	author = {Bluvstein, D. and Omran, A. and Levine, H. and Keesling, A. and Semeghini, G. and Ebadi, S. and Wang, T. T. and Michailidis, A. A. and Maskara, N. and Ho, W. W. and Choi, S. and Serbyn, M. and Greiner, M. and Vuleti{\'c}, V. and Lukin, M. D.},
	doi = {10.1126/science.abg2530},
	elocation-id = {eabg2530},
	issn = {0036-8075},
	journal = {Science},
	number = {6536},
	pages = {1355-1359},
	publisher = {American Association for the Advancement of Science},
	title = {Controlling quantum many-body dynamics in driven {R}ydberg atom arrays},
	url = {https://science.sciencemag.org/content/early/2021/02/24/science.abg2530},
	volume = {371},
	year = {2021}}

@article{Su-23a,
	author = {Su, Guo-Xian and Sun, Hui and Hudomal, Ana and Desaules, Jean-Yves and Zhou, Zhao-Yu and Yang, Bing and Halimeh, Jad C. and Yuan, Zhen-Sheng and Papi\ifmmode \acute{c}\else \'{c}\fi{}, Zlatko and Pan, Jian-Wei},
	doi = {10.1103/PhysRevResearch.5.023010},
	issue = {2},
	journal = {Phys. Rev. Res.},
	month = {Apr},
	numpages = {13},
	pages = {023010},
	publisher = {American Physical Society},
	title = {Observation of many-body scarring in a Bose-Hubbard quantum simulator},
	url = {https://link.aps.org/doi/10.1103/PhysRevResearch.5.023010},
	volume = {5},
	year = {2023}}

@article{Zha-23a,
	author = {Zhang, Pengfei and Dong, Hang and Gao, Yu and Zhao, Liangtian and Hao, Jie and Desaules, Jean-Yves and Guo, Qiujiang and Chen, Jiachen and Deng, Jinfeng and Liu, Bobo and Ren, Wenhui and Yao, Yunyan and Zhang, Xu and Xu, Shibo and Wang, Ke and Jin, Feitong and Zhu, Xuhao and Zhang, Bing and Li, Hekang and Song, Chao and Wang, Zhen and Liu, Fangli and Papi{\'c}, Zlatko and Ying, Lei and Wang, H. and Lai, Ying-Cheng},
	date = {2023/01/01},
	doi = {10.1038/s41567-022-01784-9},
	id = {Zhang2023},
	isbn = {1745-2481},
	journal = {Nature Physics},
	number = {1},
	pages = {120--125},
	title = {Many-body Hilbert space scarring on a superconducting processor},
	url = {https://doi.org/10.1038/s41567-022-01784-9},
	volume = {19},
	year = {2023}}

@article{Zha-25a,
	author = {Zhao, Luheng and Datla, Prithvi Raj and Tian, Weikun and Aliyu, Mohammad Mujahid and Loh, Huanqian},
	doi = {10.1103/PhysRevX.15.011035},
	issue = {1},
	journal = {Phys. Rev. X},
	month = {Feb},
	numpages = {18},
	pages = {011035},
	publisher = {American Physical Society},
	title = {Observation of Quantum Thermalization Restricted to Hilbert Space Fragments and ${\mathbb{Z}}_{2k}$ Scars},
	url = {https://link.aps.org/doi/10.1103/PhysRevX.15.011035},
	volume = {15},
	year = {2025}}

@article{Lar-24,
	author = {Larsen, Peter Gr{\ae}ns and Nielsen, Anne E. B.},
	issn = {2643-1564},
	journal = {Physical Review Research},
	month = oct,
	number = {4},
	publisher = {American Physical Society (APS)},
	title = {Phase transitions in quantum many-body scars},
	url = {http://dx.doi.org/10.1103/PhysRevResearch.6.L042007},
	volume = {6},
	year = {2024},
    }

@article{Got-23,
	author = {Gotta, Lorenzo and Moudgalya, Sanjay and Mazza, Leonardo},
	issn = {1079-7114},
	journal = {Physical Review Letters},
	month = nov,
	number = {19},
	publisher = {American Physical Society (APS)},
	title = {Asymptotic Quantum Many-Body Scars},
	url = {http://dx.doi.org/10.1103/PhysRevLett.131.190401},
	volume = {131},
	year = {2023}}

@article{Ren-24,
	author = {Ren, Jie and Wang, Yu-Peng and Fang, Chen},
	doi = {10.1103/PhysRevB.110.245101},
	issue = {24},
	journal = {Phys. Rev. B},
	month = {Dec},
	numpages = {25},
	pages = {245101},
	publisher = {American Physical Society},
	title = {Quasi-Nambu-Goldstone modes in many-body scar models},
	url = {https://link.aps.org/doi/10.1103/PhysRevB.110.245101},
	volume = {110},
	year = {2024}}

@article{Kun-25,
  title = {Systematic construction of asymptotic quantum many-body scar states and their relation to supersymmetric quantum mechanics},
  author = {Kunimi, Masaya and Kato, Yusuke and Katsura, Hosho},
  journal = {Phys. Rev. Res.},
  volume = {7},
  issue = {4},
  pages = {043107},
  numpages = {22},
  year = {2025},
  month = {Oct},
  publisher = {American Physical Society},
  url = {https://link.aps.org/doi/10.1103/zstl-s1y2}
}

@article{Man-45,
  author = {Mandelstam, L. and Tamm, I.},
  title = {The uncertainty relation between energy and time in nonrelativistic quantum mechanics},
  journal = {J. Phys. (USSR)},
  volume = {9},
  pages = {249--254},
  year = {1945}
}

@misc{Gio-26,
      title={Distinct Types of Parent Hamiltonians for Quantum States: Insights from the $W$ State as a Quantum Many-Body Scar}, 
      author={Lei Gioia and Sanjay Moudgalya and Olexei I. Motrunich},
      year={2026},
      eprint={2510.24713},
      archivePrefix={arXiv},
      primaryClass={quant-ph},
      url={https://arxiv.org/abs/2510.24713}, 
}

@article{Bac-12,
	author = {Bachmann, Sven and Nachtergaele, Bruno},
	doi = {10.1103/PhysRevB.86.035149},
	issue = {3},
	journal = {Phys. Rev. B},
	month = {Jul},
	numpages = {6},
	pages = {035149},
	publisher = {American Physical Society},
	title = {Product vacua with boundary states},
	url = {https://link.aps.org/doi/10.1103/PhysRevB.86.035149},
	volume = {86},
	year = {2012}}

@article{Tur-18a,
	author = {Turner, C. J. and Michailidis, A. A. and Abanin, D. A. and Serbyn, M. and Papi{\'c}, Z.},
	da = {2018/07/01},
	doi = {10.1038/s41567-018-0137-5},
	id = {Turner2018},
	isbn = {1745-2481},
	journal = {Nat. Phys.},
	number = {7},
	pages = {745},
	title = {Weak ergodicity breaking from quantum many-body scars},
	ty = {JOUR},
	url = {https://doi.org/10.1038/s41567-018-0137-5},
	volume = {14},
	year = {2018}}

@article{Sch-19,
   title={Weak Ergodicity Breaking and Quantum Many-Body Scars in Spin-1 Magnets},
   volume={123},
   ISSN={1079-7114},
   url={http://dx.doi.org/10.1103/PhysRevLett.123.147201},
   number={14},
   journal={Physical Review Letters},
   publisher={American Physical Society (APS)},
   author={Schecter, Michael and Iadecola, Thomas},
   year={2019},
   month=oct }

@article{Bra-15,
	author = {Sergey Bravyi and David Gosset},
	journal = {Journal of Mathematical Physics},
	pages = {061902},
	title = {Gapped and gapless phases of frustration-free spin-1/2 chains},
	url = {https://api.semanticscholar.org/CorpusID:9474429},
	volume = {56},
	year = {2015}}

@article{Logaric-24,
	author = {Logari{\'c}, Leonard and Dooley, Shane and Pappalardi, Silvia and Goold, John},
	issn = {1079-7114},
	journal = {Physical Review Letters},
	month = jan,
	number = {1},
	publisher = {American Physical Society (APS)},
	title = {Quantum Many-Body Scars in Dual-Unitary Circuits},
	url = {http://dx.doi.org/10.1103/PhysRevLett.132.010401},
	volume = {132},
	year = {2024}}

@article{Shi-17,
	author = {Shiraishi, Naoto and Mori, Takashi},
	issn = {1079-7114},
	journal = {Physical Review Letters},
	month = jul,
	number = {3},
	publisher = {American Physical Society (APS)},
	title = {Systematic Construction of Counterexamples to the Eigenstate Thermalization Hypothesis},
	url = {http://dx.doi.org/10.1103/PhysRevLett.119.030601},
	volume = {119},
	year = {2017}}

@article{Mor-18,
	author = {Takashi Mori and Tatsuhiko N. Ikeda and Eriko Kaminishi and Masahito Ueda},
	doi = {10.1088/1361-6455/aabcdf},
	journal = {J. Phys. B},
	month = {May},
	number = {11},
	pages = {112001},
	publisher = {{IOP} Publishing},
	title = {Thermalization and prethermalization in isolated quantum systems: a theoretical overview},
	volume = {51},
	year = 2018}

@article{Bac-15,
	author = {Bachmann, Sven and Hamza, Eman and Nachtergaele, Bruno and Young, Amanda},
	doi = {10.1007/s10955-015-1260-7},
	issn = {1572-9613},
	journal = {Journal of Statistical Physics},
	month = may,
	number = {3},
	pages = {636--658},
	publisher = {Springer Science and Business Media LLC},
	title = {Product Vacua and Boundary State Models in $d$-Dimensions},
	url = {http://dx.doi.org/10.1007/s10955-015-1260-7},
	volume = {160},
	year = {2015}}

@misc{Moh-25,
	archiveprefix = {arXiv},
	author = {Sashikanta Mohapatra and Sanjay Moudgalya and Ajit C. Balram},
	eprint = {2511.14878},
	primaryclass = {cond-mat.str-el},
	title = {Unraveling additional quantum many-body scars of the spin-$1$ $XY$ model with Fock-space cages and commutant algebras},
	url = {https://arxiv.org/abs/2511.14878},
	year = {2025}}

@article{Fla-11,
	author = {Flammia, Steven T. and Liu, Yi-Kai},
	issn = {1079-7114},
	journal = {Physical Review Letters},
	month = jun,
	number = {23},
	publisher = {American Physical Society (APS)},
	title = {Direct Fidelity Estimation from Few Pauli Measurements},
	url = {http://dx.doi.org/10.1103/PhysRevLett.106.230501},
	volume = {106},
	year = {2011},
    }

@article{Bry-19,
	author = {Brydges, Tiff and Elben, Andreas and Jurcevic, Petar and Vermersch, Beno{\^\i}t and Maier, Christine and Lanyon, Ben P. and Zoller, Peter and Blatt, Rainer and Roos, Christian F.},
	doi = {10.1126/science.aau4963},
	issn = {1095-9203},
	journal = {Science},
	month = apr,
	number = {6437},
	pages = {260--263},
	publisher = {American Association for the Advancement of Science (AAAS)},
	title = {Probing R{\'e}nyi entanglement entropy via randomized measurements},
	url = {http://dx.doi.org/10.1126/science.aau4963},
	volume = {364},
	year = {2019}}

@misc{quantinuum_h1_datasheet_2023,
	author = {Quantinuum},
	howpublished = {\url{https://quantinuum.co.jp/assets/pdf/system_model_h1_product_data_sheet.pdf}},
	month = feb,
	note = {Version 5.20. {\copyright} 2023 by Quantinuum. All rights reserved.},
	title = {Quantinuum System Model H1 Product Data Sheet},
	urldate = {2025-12-31},
	year = {2023}}

@misc{quantinuum_h1_1_access,
	author = {{Quantinuum}},
	howpublished = {\url{https://www.quantinuum.com/}},
	month = oct,
	note = {October 3--8, 2025},
	title = {{Quantinuum H1-1}},
	year = {2025}}

@article{Lar-26a,
	author = {Peter Gr{\ae}ns Larsen and Anne E. B. Nielsen and Andr{\'e} Eckardt and Francesco Petiziol},
	doi = {10.21468/SciPostPhys.20.2.036},
	journal = {SciPost Phys.},
	pages = {036},
	publisher = {SciPost},
	title = {{Experimental protocol for observing single quantum many-body scars with transmon qubits}},
	url = {https://scipost.org/10.21468/SciPostPhys.20.2.036},
	volume = {20},
	year = {2026}}

\clearpage

\setcounter{page}{1}
\onecolumngrid 

{\large \textbf{Supplementary Material for ``Dynamical signatures of conventional and asymptotic quantum many-body scars on a trapped ion simulator''}}
{\large Leonard Logari\'c, John Goold, Shane Dooley}

\section{Proof of Theorem \ref{theorem:AQMBS_main_text}}
\label{app: general_proof_aqmbs}

\subsection{Outline}

Consider a system of $N$ qudits, labelled $n \in \{0,1,\hdots,N-1 \}$, arranged on a $D$-dimensional lattice $\Lambda$, and a set of local projectors $\{ \hat{P}_{X_n} \}_{n=0}^{N-1}$ each of which acts non-trivially only on a finite number of qudits in a local neighbourhood $X_n \subset \Lambda$ of the qudit $n$. With this set of local projectors we define the local reference Hamiltonian: \begin{equation} \hat{\mathbb{H}}_+ = \sum_n \hat{P}_{X_n} , \label{eq:H_+_app} \end{equation} which is non-negative since each term in its sum is non-negative $\hat{P}_{X_n} \geq 0$. With this setup, we will prove below, in Sec. \ref{sec:prop_ham_AQMBS} and Sec. \ref{sec:prop_circuit_AQMBS}, the following two propositions:

\begin{proposition} \label{proposition:hamiltonian_AQMBS}
If $\hat{\mathbb{H}}_+ = \sum_n \hat{P}_{X_n}$ and $\hat{\mathbb{H}}_+ \ket{\mathcal{A}} \stackrel{N \to \infty}{\longrightarrow} 0$ then we have $\hat{\mathbb{H}} \ket{\mathcal{A}} \stackrel{N \to \infty}{\longrightarrow} 0$ for any Hamiltonian of the form $\hat{\mathbb{H}} = \sum_n  \hat{P}_{X_n} \hat{h}_{X_n} \hat{P}_{X_n}$, where $\hat{h}_{X_n}$ are arbitrary Hermitian operators which act non-trivially only on the qudits in $X_n$.
\end{proposition}

\begin{proposition} \label{proposition:circuit_AQMBS}
If $\hat{\mathbb{H}}_+ = \sum_n \hat{P}_{X_n}$ and $\hat{\mathbb{H}}_+ \ket{\mathcal{A}} \stackrel{N \to \infty}{\longrightarrow} 0$ then we have we have $\hat{\mathbb{U}} \ket{\mathcal{A}} \stackrel{N \to \infty}{\longrightarrow} \ket{\mathcal{A}}$ for any finite-depth circuit $\hat{\mathbb{U}}$ composed of local unitary gates of the form $\hat{U}_{X_n} = \exp(i\hat{P}_{X_n} \hat{h}_{X_n} \hat{P}_{X_n})$, where $\hat{h}_{X_n}$ are arbitrary Hermitian operators which act non-trivially only on the qudits in $X_n$.
\end{proposition}

We note that if $\ket{\mathcal{G}}$ is a zero-energy ground state of $\hat{\mathbb{H}}_+$ for finite system size $N$ then Propositions \ref{proposition:hamiltonian_AQMBS} and \ref{proposition:circuit_AQMBS} can be satisfied in a relatively trivial manner by states $\ket{\mathcal{A}}$ that approach $\ket{\mathcal{G}}$ in the thermodynamic limit. In other words, if $\ket{\mathcal{A}} = \sqrt{1 - |\epsilon|^2} \ket{\mathcal{G}} + \epsilon \ket{\psi}$ for some arbitrary normalised state $\ket{\psi}$ and $\epsilon \stackrel{N \to \infty}{\longrightarrow} 0$ then by Theorem \ref{theorem:QMBS} $\ket{\mathcal{A}}$ approaches the exact QMBS $\ket{\mathcal{G}}$ of $\hat{\mathbb{H}}$ (or the circuit $\hat{\mathbb{U}}$) in the thermodynamic limit. However, Propositions \ref{proposition:hamiltonian_AQMBS} and \ref{proposition:circuit_AQMBS} can also be satisfied in a nontrivial way by states $\ket{\mathcal{A}}$ that are positive energy eigenstates of $\hat{\mathbb{H}}_+$ for any finite $N$, but only become zero-energy ground states of $\hat{\mathbb{H}}_+$ in the thermodynamic limit (e.g., gapless excitations of $\hat{\mathbb{H}}_+$, or if $\hat{\mathbb{H}}_+$ is asymptotically frustration-free). In this case, $\ket{\mathcal{A}}$ is orthogonal to any exact QMBS $\ket{\mathcal{G}}$, and Propositions \ref{proposition:hamiltonian_AQMBS} and \ref{proposition:circuit_AQMBS} lead to Theorem \ref{theorem:AQMBS_main_text} in the main text:

\aqmbsthm*

But before proving Propositions \ref{proposition:hamiltonian_AQMBS} and \ref{proposition:circuit_AQMBS} (and hence Theorem \ref{theorem:AQMBS_main_text}), it is useful to first prove the following: 

\begin{lemma} \label{lemma:vanishing_quantities}
If $\hat{\mathbb{H}}_+ = \sum_n \hat{P}_{X_n}$ and $\hat{\mathbb{H}}_+ \ket{\mathcal{A}} \stackrel{N \to \infty}{\longrightarrow} 0$ then:
\begin{enumerate}
\item $\bra{\mathcal{A}} \hat{\mathbb{H}}'_+ \ket{\mathcal{A}} \stackrel{N \to \infty}{\longrightarrow} 0$ where $\hat{\mathbb{H}}'_+ \equiv \sum_{n} \lambda_n \hat{P}_{X_n}$ for any non-negative but finite couplings $0 \leq \lambda_n < \infty$.
\item $\bra{\mathcal{A}} \hat{\mathbb{J}}'_+ \ket{\mathcal{A}} \stackrel{N \to \infty}{\longrightarrow} 0$ where $\hat{\mathbb{J}}'_+ \equiv  \sum \limits_{\substack{n \neq n' \\  X_n \cap X_{n'} = \emptyset}} \lambda_{n,n'} \hat{P}_{X_n} \hat{P}_{X_{n'}}$ for any non-negative but finite couplings $0 \leq \lambda_{n,n'} < \infty$ and the sum is over all non-overlapping pairs of regions $X_n \cap X_{n'} = \emptyset$. 
\item $\bra{\mathcal{A}} \hat{\mathbb{F}}_+ \ket{\mathcal{A}} \stackrel{N \to \infty}{\longrightarrow} 0$ where $\hat{\mathbb{F}}_+ \equiv \sum \limits_{\substack{n < n' \\  X_n \cap X_{n'} \neq \emptyset}} \sum_{\alpha, \alpha'} \Big( |\pi_{X_{n}}^{\alpha} \rangle\langle \pi_{X_{n}}^{\alpha} | \otimes \hat{\mathbb{I}}_{\bar{X}_{n}} + | \pi_{X_{n'}}^{\alpha'} \rangle \langle \pi_{X_{n'}}^{\alpha'} | \otimes \hat{\mathbb{I}}_{\bar{X}_{n'}} \Big)^2$ is a non-negative Hermitian operator and the sum is over all overlapping but non-equal pairs of regions. Here, $\hat{P}_{X_n} = \sum_{\alpha = 0}^{r_n - 1} |\pi_{X_{n}}^{\alpha} \rangle\langle \pi_{X_{n}}^{\alpha} | \otimes \hat{\mathbb{I}}_{\bar{X}_{n}}$ is a spectral decomposition of the local projector in the subspace of the qubits in $X_n$, $1 \leq r_n < \infty$ is the rank of the projector, and $\hat{\mathbb{I}}_{\bar{X}_{n}}$ is the identity operator on all qudits in the complement $\bar{X}_n$ of $X_n$.
\end{enumerate}
\end{lemma}

\emph{Proof of 1.} Since the operators $\hat{P}_{X_n}$ and the numbers $\lambda_n$ are non-negative, we have: \begin{equation} \bra{\mathcal{A}} \hat{\mathbb{H}}'_+ \ket{\mathcal{A}} = \sum_n \lambda_n \bra{\mathcal{A}}\hat{P}_{X_n} \ket{\mathcal{A}} \leq \lambda_{\rm max} \sum_n \bra{\mathcal{A}}\hat{P}_{X_n} \ket{\mathcal{A}} = \lambda_{\rm max} \bra{\mathcal{A}}\hat{\mathbb{H}}_+ \ket{\mathcal{A}} , \label{eq:H_prime_+} \end{equation} where $\lambda_{\rm max} = \max_n \{ \lambda_n \} < \infty$. But since $\hat{\mathbb{H}}_+ \ket{\mathcal{A}} \stackrel{N \to \infty}{\longrightarrow} 0$ we have $\bra{\mathcal{A}}\hat{\mathbb{H}}_+ \ket{\mathcal{A}} \stackrel{N \to \infty}{\longrightarrow} 0$, and so the right hand side of Eq. \ref{eq:H_prime_+} vanishes in the thermodynamic limit, which implies that $\bra{\mathcal{A}} \hat{\mathbb{H}}'_+ \ket{\mathcal{A}} \stackrel{N \to \infty}{\longrightarrow} 0$. \\

\emph{Proof of 2 and 3.} Since $\hat{\mathbb{H}}_+ \ket{\mathcal{A}} \stackrel{N \to \infty}{\longrightarrow} 0$ its norm vanishes in the thermodynamic limit, i.e.: \begin{equation} \| \hat{\mathbb{H}}_+ \ket{\mathcal{A}} \|^2 \equiv \bra{\mathcal{A}} \hat{\mathbb{H}}_+ \hat{\mathbb{H}}_+ \ket{\mathcal{A}} = \sum_{n,n'} \bra{\mathcal{A}} \hat{P}_{X_n} \hat{P}_{X_{n'}} \ket{\mathcal{A}} \stackrel{N \to \infty}{\longrightarrow} 0 . \end{equation}
Let us break up the sum $\sum_{n,n'}$ into three groups of terms: those for which $n=n'$, those for which $n \neq n'$ and $X_n$ does not overlap with $X_{n'}$, and those for which $n \neq n'$ and $X_n$ overlaps with $X_{n'}$. This gives:
\begin{eqnarray} \| \hat{\mathbb{H}}_+ \ket{\mathcal{A}} \|^2 &=& \sum_{n, n'} \bra{\mathcal{A}} \hat{P}_{X_n} \hat{P}_{X_{n'}} \ket{\mathcal{A}} \\ &=& \sum_{n} \bra{\mathcal{A}} \hat{P}_{X_n} \ket{\mathcal{A}} + \sum_{\substack{n \neq n' \\  X_n \cap X_{n'} = \emptyset}} \bra{\mathcal{A}} \hat{P}_{X_n} \hat{P}_{X_{n'}} \ket{\mathcal{A}} + \sum_{\substack{n \neq n' \\  X_n \cap X_{n'} \neq \emptyset}} \bra{\mathcal{A}} \hat{P}_{X_n} \hat{P}_{X_{n'}} \ket{\mathcal{A}} \\ &=& \bra{\mathcal{A}} \hat{\mathbb{H}}_{+} \ket{\mathcal{A}}  + \bra{\mathcal{A}} \hat{\mathbb{J}}_+ \ket{\mathcal{A}} + \sum_{\substack{n < n' \\  X_n \cap X_{n'} \neq \emptyset}} \bra{\mathcal{A}} (\hat{P}_{X_n} \hat{P}_{X_{n'}} + \hat{P}_{X_{n'}} \hat{P}_{X_{n}}) \ket{\mathcal{A}}
\\ &\stackrel{N \to \infty}{\longrightarrow}& 0 , \end{eqnarray} where in the first group of terms we have used $\sum_{n} \bra{\mathcal{A}} \hat{P}_{X_n} \hat{P}_{X_n} \ket{\mathcal{A}} = \sum_{n} \bra{\mathcal{A}} \hat{P}_{X_n} \ket{\mathcal{A}} = \bra{\mathcal{A}} \hat{\mathbb{H}}_{+} \ket{\mathcal{A}}$, and in the second group we have defined: \begin{equation} \hat{\mathbb{J}}_+ \equiv  \sum \limits_{\substack{n \neq n' \\  X_n \cap X_{n'} = \emptyset}} \hat{P}_{X_n}  \hat{P}_{X_{n'}}\end{equation} which is similar to the definition of $\hat{\mathbb{J}}'_+$ given in the statement of Lemma \ref{lemma:vanishing_quantities} but with the couplings set to $\lambda_{n,n'} = 1$. Now, in the last group of terms, we write the projectors in their partial spectral decomposition $\hat{P}_{X_n} = \sum_{\alpha = 0}^{r_n - 1} \ket{\pi_{X_n}^{(\alpha)}} \bra{\pi_{X_n}^{(\alpha)}} \otimes \hat{\mathbb{I}}_{\bar{X}_{n}}$, which gives:
\begin{eqnarray} \| \hat{\mathbb{H}}_+ \ket{\mathcal{A}} \|^2 &=& \bra{\mathcal{A}} \hat{\mathbb{H}}_{+} \ket{\mathcal{A}} + \bra{\mathcal{A}} \hat{\mathbb{J}}_+ \ket{\mathcal{A}} \\ && \qquad + \sum_{\substack{n < n' \\  X_n \cap X_{n'} \neq \emptyset}} \sum_{\alpha, \alpha'} \bra{\mathcal{A}} \Big( |\pi_{X_{n}}^{\alpha} \rangle\langle \pi_{X_{n}}^{\alpha} | \cdot | \pi_{X_{n'}}^{\alpha'} \rangle \langle \pi_{X_{n'}}^{\alpha'} | + |\pi_{X_{n'}}^{\alpha'} \rangle\langle \pi_{X_{n'}}^{\alpha'} | \cdot | \pi_{X_{n}}^{\alpha} \rangle \langle \pi_{X_{n}}^{\alpha} | \Big) \ket{\mathcal{A}} \stackrel{N \to \infty}{\longrightarrow} 0 \label{eq:H_+_norm_almost_there} \end{eqnarray}
Rewrite the term in round brackets as:
\begin{eqnarray} |\pi_{X_{n}}^{\alpha} \rangle\langle \pi_{X_{n}}^{\alpha} | \cdot | \pi_{X_{n'}}^{\alpha'} \rangle \langle \pi_{X_{n'}}^{\alpha'} | &+& |\pi_{X_{n'}}^{\alpha'} \rangle\langle \pi_{X_{n'}}^{\alpha'} | \cdot | \pi_{X_{n}}^{\alpha} \rangle \langle \pi_{X_{n}}^{\alpha} | \nonumber \\ &=& \Big( |\pi_{X_{n}}^{\alpha} \rangle\langle \pi_{X_{n}}^{\alpha} | + | \pi_{X_{n'}}^{\alpha'} \rangle \langle \pi_{X_{n'}}^{\alpha'} | \Big)^2 - |\pi_{X_{n'}}^{\alpha'} \rangle\langle \pi_{X_{n'}}^{\alpha'} | - | \pi_{X_{n}}^{\alpha} \rangle \langle \pi_{X_{n}}^{\alpha} | . \label{eq:cross_terms} \end{eqnarray}
Substituting back into Eq. \ref{eq:H_+_norm_almost_there} gives:
\begin{eqnarray} \| \hat{\mathbb{H}}_+ \ket{\mathcal{A}} \|^2 = \bra{\mathcal{A}} \hat{\mathbb{H}}_{+} \ket{\mathcal{A}} + \bra{\mathcal{A}} \hat{\mathbb{J}}_+ \ket{\mathcal{A}} + \bra{\mathcal{A}} \hat{\mathbb{F}}_+ \ket{\mathcal{A}} -  \bra{\mathcal{A}} \hat{\mathbb{H}}'_{+} \ket{\mathcal{A}}  \stackrel{N \to \infty}{\longrightarrow} 0 , \label{eq:H_+_norm_final_form} \end{eqnarray} where we have used the definition of $\hat{\mathbb{F}}_+$ given in the statement of the lemma, and also: \begin{equation} \hat{\mathbb{H}}'_+ = \sum_{\substack{n < n' \\  X_n \cap X_{n'} \neq \emptyset}} \sum_{\alpha, \alpha'} ( |\pi_{X_{n'}}^{\alpha'} \rangle\langle \pi_{X_{n'}}^{\alpha'} | + | \pi_{X_{n}}^{\alpha} \rangle \langle \pi_{X_{n}}^{\alpha} | ) = \sum_{\substack{n \neq n' \\  X_n \cap X_{n'} \neq \emptyset}} \sum_{\alpha, \alpha'} | \pi_{X_{n}}^{\alpha} \rangle \langle \pi_{X_{n}}^{\alpha} |  = \sum_{n} r_n c_n \hat{P}_{X_n} , \end{equation} which is in the form of $\hat{\mathbb{H}}'_+$ in the first part of the lemma. Here, $0 \leq c_n < \infty$ is the number of regions $X_{n' \neq n}$ that overlap with $X_{n}$.

Now, since we know that $\bra{\mathcal{A}} \hat{\mathbb{H}}_{+} \ket{\mathcal{A}} \stackrel{N \to \infty}{\longrightarrow} 0$ by assumption, $\bra{\mathcal{A}} \hat{\mathbb{H}}'_{+} \ket{\mathcal{A}} \stackrel{N \to \infty}{\longrightarrow} 0$ (by the first part of the lemma) and also that $\hat{\mathbb{F}}_+$ and $\hat{\mathbb{J}}_+$ are both non-negative Hermitian operators, the only way to have $\| \hat{\mathbb{H}}_+ \ket{\mathcal{A}} \|^2 \stackrel{N \to \infty}{\longrightarrow} 0$ in Eq. \ref{eq:H_+_norm_final_form} is if we have: \begin{equation} \bra{\mathcal{A}} \hat{\mathbb{F}}_{+} \ket{\mathcal{A}} \stackrel{N \to \infty}{\longrightarrow} 0 , \label{eq:F_lim} \end{equation} and: \begin{equation} \bra{\mathcal{A}} \hat{\mathbb{J}}_{+} \ket{\mathcal{A}} \stackrel{N \to \infty}{\longrightarrow} 0 . \label{eq:J_lim} \end{equation} Eq. \ref{eq:F_lim} completes part 3 of the Lemma. 

Finally, since the operators $\hat{P}_{X_n} \hat{P}_{X_{n'}}$ are non-negative when $X_n$ and $X_{n'}$ do not intersect, and the numbers $\lambda_{n,n'}$ are also non-negative, we have: \begin{equation} \bra{\mathcal{A}} \hat{\mathbb{J}}'_+ \ket{\mathcal{A}} = \sum \limits_{\substack{n \neq n' \\  X_n \cap X_{n'} = \emptyset}} \lambda_{n,n'} \bra{\mathcal{A}} \hat{P}_{X_n} \hat{P}_{X_{n'}} \ket{\mathcal{A}} \leq \lambda_{\rm max} \sum \limits_{\substack{n \neq n' \\  X_n \cap X_{n'} = \emptyset}} \bra{\mathcal{A}} \hat{P}_{X_n} \hat{P}_{X_{n'}} \ket{\mathcal{A}}  = \lambda_{\rm max} \bra{\mathcal{A}}\hat{\mathbb{J}}_+ \ket{\mathcal{A}} , \label{eq:H_prime_+2} \end{equation} where $\lambda_{\rm max} = \max_{n,n'} \{ \lambda_n \}$. But since $\lambda_{\rm max}$ is finite and $\bra{\mathcal{A}}\hat{\mathbb{J}}_+ \ket{\mathcal{A}} \stackrel{N \to \infty}{\longrightarrow} 0$ (Eq. \ref{eq:J_lim}), the right hand side vanishes in the thermodynamic limit, which implies that $\bra{\mathcal{A}} \hat{\mathbb{J}}'_+ \ket{\mathcal{A}} \stackrel{N \to \infty}{\longrightarrow} 0$. This completes part 2 of the Lemma. \unskip\nobreak\hfill $\square$

\subsection{Proof of Proposition \ref{proposition:hamiltonian_AQMBS}} \label{sec:prop_ham_AQMBS}

To prove Proposition \ref{proposition:hamiltonian_AQMBS} we will show that the norm of the vector $\hat{\mathbb{H}} \ket{\mathcal{A}}$ vanishes in the thermodynamics limit. To do this, let us first write the squared norm as: \begin{equation} \| \hat{\mathbb{H}} \ket{\mathcal{A}} \|^2 = \bra{\mathcal{A}} \hat{\mathbb{H}} \hat{\mathbb{H}} \ket{\mathcal{A}} = \sum_{n,n'} \bra{\mathcal{A}} \hat{P}_{X_n} \hat{h}_{X_n} \hat{P}_{X_n}  \hat{P}_{X_{n'}} \hat{h}_{X_{n'}} \hat{P}_{X_{n'}} \ket{\mathcal{A}} . \label{eq:H_norm} \end{equation} Now, since $\hat{P}_{X_n} \hat{h}_{X_n} \hat{P}_{X_n}$ and $\hat{P}_{X_n}$ commute with each other they have a common eigenbasis. Previously we wrote the spectral decomposition of the projector as $\hat{P}_{X_n} = \sum_{\alpha = 0}^{r_n - 1} \ket{\pi_{X_{n}}^{\alpha}} \bra{\pi_{X_{n}}^{\alpha}} \otimes \hat{\mathbb{I}}_{\bar{X}_{n}}$, so let us write the spectral decomposition of the commuting projected term as $\hat{P}_{X_n} \hat{h}_{X_n} \hat{P}_{X_n} = \sum_{\alpha = 0}^{r_n - 1} \eta_{X_n}^\alpha \ket{\pi_{X_{n}}^{\alpha}} \bra{\pi_{X_{n}}^{\alpha}} \otimes \hat{\mathbb{I}}_{\bar{X}_{n}}$, where $\eta_{X_n}^\alpha \in \mathbb{R}$. Substituting into Eq. \ref{eq:H_norm} gives: \begin{equation} \| \hat{\mathbb{H}} \ket{\mathcal{A}} \|^2 = \sum_{n,n'} \sum_{\alpha, \alpha'} \eta_{X_n}^{\alpha} \eta_{X_{n'}}^{\alpha'} \bra{\mathcal{A}} (\ket{\pi_{X_{n}}^{\alpha}} \bra{\pi_{X_{n}}^{\alpha}} \otimes \hat{\mathbb{I}}_{\bar{X}_{n}}) (\ket{\pi_{X_{n'}}^{\alpha'}} \bra{\pi_{X_{n'}}^{\alpha'}} \otimes \hat{\mathbb{I}}_{\bar{X}_{n'}}) \ket{\mathcal{A}} . \end{equation}
Now we break up the sum $\sum_{n,n'}$ into three groups of terms: those for which $n=n'$, those for which $n \neq n'$ and $X_n$ does not overlap with $X_{n'}$, and those for which $n \neq n'$ and $X_n$ overlaps with $X_{n'}$. This gives:

\begin{eqnarray} \| \hat{\mathbb{H}} \ket{\mathcal{A}} \|^2 &=& \sum_{n} \sum_{\alpha} (\eta_{X_n}^\alpha)^2 \bra{\mathcal{A}} (\ket{\pi_{X_{n}}^{\alpha}} \bra{\pi_{X_{n}}^{\alpha}} \otimes \hat{\mathbb{I}}_{\bar{X}_{n}}) \ket{\mathcal{A}} \\ && + \sum_{\substack{n \neq n' \\  X_n \cap X_{n'} = \emptyset}} \sum_{\alpha, \alpha'} \eta_{X_n}^\alpha \eta_{X_{n'}}^{\alpha'} \bra{\mathcal{A}} \left( \ket{\pi_{X_{n}}^{\alpha}} \bra{\pi_{X_{n}}^{\alpha}} \cdot \ket{\pi_{X_{n'}}^{\alpha'}} \bra{\pi_{X_{n'}}^{\alpha'}} \right) \ket{\mathcal{A}} \\ && + \sum_{\substack{n < n' \\  X_n \cap X_{n'} \neq \emptyset}} \sum_{\alpha, \alpha'} \eta_{X_n}^\alpha \eta_{X_{n'}}^{\alpha'} \bra{\mathcal{A}} \left( \ket{\pi_{X_{n}}^{\alpha}} \bra{\pi_{X_{n}}^{\alpha}} \cdot \ket{\pi_{X_{n'}}^{\alpha'}} \bra{\pi_{X_{n'}}^{\alpha'}} + \ket{\pi_{X_{n'}}^{\alpha'}} \bra{\pi_{X_{n'}}^{\alpha'}} \cdot \ket{\pi_{X_{n}}^{\alpha}} \bra{\pi_{X_{n}}^{\alpha}}  \right) \ket{\mathcal{A}} . \end{eqnarray} 
Next, we rewrite the term in round brackets in the last line using Eq. \ref{eq:cross_terms}, to obtain:

\begin{eqnarray} \| \hat{\mathbb{H}} \ket{\mathcal{A}} \|^2 &=& \underbrace{\sum_{n} \sum_{\alpha} (\eta_{X_n}^\alpha)^2 \bra{\mathcal{A}} \left( \ket{\pi_{X_{n}}^{\alpha}} \bra{\pi_{X_{n}}^{\alpha}} \otimes \hat{\mathbb{I}}_{\bar{X}_{n}} \right) \ket{\mathcal{A}}}_{\equiv A} \\ && + \underbrace{\sum_{\substack{n \neq n' \\  X_n \cap X_{n'} = \emptyset}} \sum_{\alpha, \alpha'} \eta_{X_n}^\alpha \eta_{X_{n'}}^{\alpha'} \bra{\mathcal{A}} \left( \ket{\pi_{X_{n}}^{\alpha}} \bra{\pi_{X_{n}}^{\alpha}} \cdot \ket{\pi_{X_{n'}}^{\alpha'}} \bra{\pi_{X_{n'}}^{\alpha'}} \right) \ket{\mathcal{A}} }_{\equiv B} \\ && + \underbrace{\sum_{\substack{n < n' \\  X_n \cap X_{n'} \neq \emptyset}} \sum_{\alpha, \alpha'} \eta_{X_n}^\alpha \eta_{X_{n'}}^{\alpha'} \bra{\mathcal{A}} \left( \ket{\pi_{X_{n}}^{\alpha}} \bra{\pi_{X_{n}}^{\alpha}} + \ket{\pi_{X_{n'}}^{\alpha'}} \bra{\pi_{X_{n'}}^{\alpha'}} \right)^2 \ket{\mathcal{A}} }_{\equiv C} \\ && - \underbrace{\sum_{\substack{n \neq n' \\  X_n \cap X_{n'} \neq \emptyset}} \sum_{\alpha, \alpha'} \eta_{X_n}^\alpha \eta_{X_{n'}}^{\alpha'} \bra{\mathcal{A}} \Big( \ket{\pi_{X_{n}}^{\alpha}} \bra{\pi_{X_{n}}^{\alpha}} \otimes \hat{\mathbb{I}}_{\bar{X}_n} \Big) \ket{\mathcal{A}} }_{\equiv D} , \end{eqnarray} where we have grouped terms together into $A$, $B$, $C$, $D$ as indicated above. By the triangle inequality, we must have: \begin{equation} \| \hat{\mathbb{H}} \ket{\mathcal{A}} \|^2 \leq |A| + |B| + |C| + |D| . \end{equation}
To show that $\| \hat{\mathbb{H}} \ket{\mathcal{A}} \|^2 \stackrel{N \to \infty}{\longrightarrow} 0$ we will now show that each term on the right hand side vanishes in the thermodynamic limit.

The absolute value of the first term is: \begin{eqnarray} |A| = \Big| \sum_{n, \alpha} (\eta_{X_n}^\alpha)^2 \bra{\mathcal{A}} (\ket{\pi_{X_{n}}^{\alpha}} \bra{\pi_{X_{n}}^{\alpha}} \otimes \hat{\mathbb{I}}_{\bar{X}_{n}}) \ket{\mathcal{A}} \Big| &\leq& |\eta|_{\rm max}^2 \sum_{n, \alpha} \bra{\mathcal{A}} (\ket{\pi_{X_{n}}^{\alpha}} \bra{\pi_{X_{n}}^{\alpha}} \otimes \hat{\mathbb{I}}_{\bar{X}_{n}}) \ket{\mathcal{A}} = |\eta|_{\rm max}^2 \bra{\mathcal{A}} \hat{\mathbb{H}}_+ \ket{\mathcal{A}} \label{eq:A_norm} \end{eqnarray} where $|\eta|_{\rm max} \equiv \max_{n,\alpha} \{ |\eta_n^\alpha| \}$. Since we already know that $\bra{\mathcal{A}} \hat{\mathbb{H}}_+ \ket{\mathcal{A}} \stackrel{N \to \infty}{\longrightarrow} 0$ (by assumption), the right hand side of Eq. \ref{eq:A_norm} vanishes and so we have $|A| \stackrel{N \to \infty}{\longrightarrow} 0$ .

The absolute value of the second term is, using the triangle inequality, bounded by: \begin{eqnarray} |B| &=& \Big| \sum_{\substack{n \neq n' \\  X_n \cap X_{n'} = \emptyset}} \sum_{\alpha, \alpha'} \eta_{X_n}^\alpha \eta_{X_{n'}}^{\alpha'} \bra{\mathcal{A}} \left( \ket{\pi_{X_{n}}^{\alpha}} \bra{\pi_{X_{n}}^{\alpha}} \cdot \ket{\pi_{X_{n'}}^{\alpha'}} \bra{\pi_{X_{n'}}^{\alpha'}} \right) \ket{\mathcal{A}} \Big| \\ &\leq& \sum_{\substack{n \neq n' \\  X_n \cap X_{n'} = \emptyset}} \sum_{\alpha, \alpha'} |\eta_{X_n}^\alpha| |\eta_{X_{n'}}^{\alpha'}| \bra{\mathcal{A}} \left( \ket{\pi_{X_{n}}^{\alpha}} \bra{\pi_{X_{n}}^{\alpha}} \cdot \ket{\pi_{X_{n'}}^{\alpha'}} \bra{\pi_{X_{n'}}^{\alpha'}} \right) \ket{\mathcal{A}} \\ &\leq& |\eta|_{\rm max}^2 \sum_{\substack{n \neq n' \\  X_n \cap X_{n'} = \emptyset}} \sum_{\alpha, \alpha'} \bra{\mathcal{A}} \left( \ket{\pi_{X_{n}}^{\alpha}} \bra{\pi_{X_{n}}^{\alpha}} \cdot \ket{\pi_{X_{n'}}^{\alpha'}} \bra{\pi_{X_{n'}}^{\alpha'}} \right) \ket{\mathcal{A}} \\ &=& |\eta|_{\rm max}^2 \sum_{\substack{n \neq n' \\  X_n \cap X_{n'} = \emptyset}} \bra{\mathcal{A}} \hat{P}_{X_n} \hat{P}_{X_{n'}} \ket{\mathcal{A}} \\ &=& |\eta|_{\rm max}^2 \bra{\mathcal{A}} \hat{\mathbb{J}}_+ \ket{\mathcal{A}} .\end{eqnarray} Since we know that $\bra{\mathcal{A}} \hat{\mathbb{J}}_+ \ket{\mathcal{A}} \stackrel{N \to \infty}{\longrightarrow} 0$ (by Lemma \ref{lemma:vanishing_quantities}), the last line vanishes and so we have $|B| \stackrel{N \to \infty}{\longrightarrow} 0$.

The absolute value of the third term is, using the triangle inequality, bounded by: \begin{eqnarray} |C| &=& \Big| \sum_{\substack{n < n' \\  X_n \cap X_{n'} \neq \emptyset}} \sum_{\alpha, \alpha'} \eta_{X_n}^\alpha \eta_{X_{n'}}^{\alpha'} \bra{\mathcal{A}} \left( \ket{\pi_{X_{n}}^{\alpha}} \bra{\pi_{X_{n}}^{\alpha}} + \ket{\pi_{X_{n'}}^{\alpha'}} \bra{\pi_{X_{n'}}^{\alpha'}} \right)^2 \ket{\mathcal{A}} \Big| \\ &\leq& \sum_{\substack{n < n' \\  X_n \cap X_{n'} \neq \emptyset}} \sum_{\alpha, \alpha'} |\eta_{X_n}^\alpha| |\eta_{X_{n'}}^{\alpha'}| \bra{\mathcal{A}} \left( \ket{\pi_{X_{n}}^{\alpha}} \bra{\pi_{X_{n}}^{\alpha}} + \ket{\pi_{X_{n'}}^{\alpha'}} \bra{\pi_{X_{n'}}^{\alpha'}} \right)^2 \ket{\mathcal{A}} \\ &\leq& |\eta|_{\rm max}^2 \sum_{\substack{n < n' \\  X_n \cap X_{n'} \neq \emptyset}} \sum_{\alpha, \alpha'} \bra{\mathcal{A}} \left( \ket{\pi_{X_{n}}^{\alpha}} \bra{\pi_{X_{n}}^{\alpha}} + \ket{\pi_{X_{n'}}^{\alpha'}} \bra{\pi_{X_{n'}}^{\alpha'}} \right)^2 \ket{\mathcal{A}} \\ &=& |\eta|_{\rm max}^2 \bra{\mathcal{A}} \hat{\mathbb{F}}_+ \ket{\mathcal{A}} . \end{eqnarray} Since we know that $\bra{\mathcal{A}} \hat{\mathbb{F}}_+ \ket{\mathcal{A}} \stackrel{N \to \infty}{\longrightarrow} 0$ (by Lemma \ref{lemma:vanishing_quantities}), the last line vanishes and so we have $|C| \stackrel{N \to \infty}{\longrightarrow} 0$.

Finally, the absolute value of the last term is, using the triangle inequality, bounded by: \begin{eqnarray} |D| &=& \Big| \sum_{n} \sum_{\alpha} \eta_{X_n}^\alpha \Big( \sum_{\alpha'} \sum_{\substack{n' \\ n' \neq n,  X_n \cap X_{n'} \neq \emptyset}} \eta_{X_{n'}}^{\alpha'}  \Big) \bra{\mathcal{A}} \left( \ket{\pi_{X_{n}}^{\alpha}} \bra{\pi_{X_{n}}^{\alpha}} \otimes \hat{\mathbb{I}}_{\bar{X}_{n}} \right) \ket{\mathcal{A}} \Big| \\ &\leq& \sum_{n} \sum_{\alpha} |\eta_{X_n}^\alpha| \Big| \Big( \sum_{\alpha'} \sum_{\substack{n' \\ n' \neq n,  X_n \cap X_{n'} \neq \emptyset}} \eta_{X_{n'}}^{\alpha'}  \Big) \Big| \bra{\mathcal{A}} \left( \ket{\pi_{X_{n}}^{\alpha}} \bra{\pi_{X_{n}}^{\alpha}} \otimes \hat{\mathbb{I}}_{\bar{X}_{n}} \right) \ket{\mathcal{A}} \\ &\leq& |\eta|_{\rm max}^2 \sum_n \sum_\alpha r_{\rm max} c_n \bra{\mathcal{A}} (\ket{\pi_{X_{n}}^{\alpha}} \bra{\pi_{X_{n}}^{\alpha}} \otimes \hat{\mathbb{I}}_{\bar{X}_{n}}) \ket{\mathcal{A}} \\ &=& |\eta|_{\rm max}^2 \bra{\mathcal{A}} \hat{\mathbb{H}}'_+ \ket{\mathcal{A}} . \end{eqnarray} Since we know that $\bra{\mathcal{A}} \hat{\mathbb{H}}'_+ \ket{\mathcal{A}} \stackrel{N \to \infty}{\longrightarrow} 0$ (by Lemma \ref{lemma:vanishing_quantities}), the last line vanishes and we have $|D| \stackrel{N \to \infty}{\longrightarrow} 0$.

Putting everything together, we have $\| \hat{\mathbb{H}} \ket{\mathcal{A}} \|^2 \leq |A| + |B| + |C| + |D| \stackrel{N \to \infty}{\longrightarrow} 0$. \unskip\nobreak\hfill $\square$

\subsection{Proof of Proposition \ref{proposition:circuit_AQMBS}} \label{sec:prop_circuit_AQMBS}

Let $\hat{\mathbb{U}}$ be a circuit of depth $d$, which means that it can be expressed as a product of $d$ layers: \begin{equation} \hat{\mathbb{U}} = \hat{\mathbb{U}}^{(d)} \hat{\mathbb{U}}^{(d-1)} \hdots \hat{\mathbb{U}}^{(2)} \hat{\mathbb{U}}^{(1)} , \end{equation} where the $\ell$'th layer is a product of \emph{commuting} local unitary gates: \begin{eqnarray} \hat{\mathbb{U}}^{(\ell)} &=& \prod_{n \in L^{(\ell)}} \exp (i\hat{P}_{X_n} \hat{h}_{X_n} \hat{P}_{X_n}) \\ &=& \exp \left(i \sum_{n \in L^{(\ell)}} \hat{P}_{X_n} \hat{h}_{X_n} \hat{P}_{X_n} \right) \\ &=& \exp \left( i \hat{\mathbb{H}}^{(\ell)} \right) . \end{eqnarray} Here we have defined the Hamiltonian $\hat{\mathbb{H}}^{(\ell)} \equiv \sum_{n \in L^{(\ell)}} \hat{P}_{X_n} \hat{h}_{X_n} \hat{P}_{X_n}$ which generates the unitary $\hat{\mathbb{U}}^{(\ell)}$, and $L^{(\ell)}$ is the subset of indices for the commuting gates in the $\ell$'th layer (i.e., $X_n \cap X_{n'} = \emptyset$ for any $n,n' \in L^{(\ell)}$).

Now consider the non-negative Hamiltonian: \begin{equation} \hat{\mathbb{H}}_+^{(\ell)} \equiv \sum_{n \in L^{(\ell)}} \hat{P}_{X_n} . \end{equation} Since we have assumed that $\hat{\mathbb{H}}_+ \ket{\mathcal{A}} \stackrel{N \to \infty}{\longrightarrow} 0$, we can use Lemma \ref{lemma:vanishing_quantities} to see that:

\begin{eqnarray} \bra{\mathcal{A}} \hat{\mathbb{H}}_+^{(\ell)} \hat{\mathbb{H}}_+^{(\ell)} \ket{\mathcal{A}} &=& \bra{\mathcal{A}} \left( \sum_{n\in L^{(\ell)}} \hat{P}_{X_n} \right)  \ket{\mathcal{A}} + \bra{\mathcal{A}} \left( \sum_{n \neq n' \in L^{(\ell)}} \hat{P}_{X_n} \hat{P}_{X_{n'}} \right)  \ket{\mathcal{A}} \\ & \stackrel{N \to \infty}{\longrightarrow} & 0 , \end{eqnarray} where the first term vanishes due to the part 1 of Lemma \ref{lemma:vanishing_quantities} and the second term vanishes due to part 2 of Lemma \ref{lemma:vanishing_quantities}. Since $\hat{\mathbb{H}}_+^{(\ell)} \ket{\mathcal{A}} \stackrel{N \to \infty}{\longrightarrow} 0$, by Proposition \ref{proposition:hamiltonian_AQMBS} we also have $\hat{\mathbb{H}}^{(\ell)} \ket{\mathcal{A}} \stackrel{N \to \infty}{\longrightarrow} 0$. This means that: \begin{equation} e^{i\hat{\mathbb{H}}^{(\ell)}} \ket{\mathcal{A}} = \sqrt{1 - |\epsilon_N^{(\ell)}|^2} \ket{\mathcal{A}} + \epsilon_N^{(\ell)} \ket{\psi_N^{(\ell)}} , \label{eq:U^l|A>} \end{equation} where $\epsilon_N^{(\ell)} \stackrel{N \to \infty}{\longrightarrow} 0$. So we have:

\begin{eqnarray} \hat{\mathbb{U}} \ket{\mathcal{A}} &=& e^{i\hat{\mathbb{H}}^{(d)}} e^{i\hat{\mathbb{H}}^{(d-1)}} \hdots e^{i\hat{\mathbb{H}}^{(1)}} \ket{\mathcal{A}} \label{eq:U|A>} \\ &=& \sqrt{1 - |\epsilon_N^{(d)}|^2} \sqrt{1 - |\epsilon_N^{(d-1)}|^2} \hdots \sqrt{1 - |\epsilon_N^{(1)}|^2} \ket{\mathcal{A}} \nonumber \\ && + \sum_{\ell = 1}^d \epsilon_N^{(\ell)} \sqrt{1 - |\epsilon_N^{(\ell-1)}|^2} \sqrt{1 - |\epsilon_N^{(\ell-2)}|^2} \hdots \sqrt{1 - |\epsilon_N^{(1)}|^2} e^{i\hat{\mathbb{H}}^{(d)}} e^{i\hat{\mathbb{H}}^{(d)-1}} \hdots e^{i\hat{\mathbb{H}}^{(\ell+1)}} \ket{\psi_N^{(\ell)}} , \label{eq:U|A>_expanded} \end{eqnarray} where to get from Eq. \ref{eq:U|A>} to Eq. \ref{eq:U|A>_expanded} we repeatedly use Eq. \ref{eq:U^l|A>}. Now, if the depth $d$ of the circuit is finite, only the first term of Eq. \ref{eq:U|A>_expanded} remains in the thermodynamic limit, and we have: \begin{equation} \hat{\mathbb{U}} \ket{\mathcal{A}} \stackrel{N \to \infty}{\longrightarrow} \ket{\mathcal{A}} , \end{equation} as we wished to show. \unskip\nobreak\hfill $\square$

\section{Boundary QMBS} \label{app: edge_localised_state}

\subsection{The boundary state is a ground state of the PVBS reference Hamiltonian}

In this section, we verify that the boundary QMBS from Eq. \ref{eq:BQMBS}:
\begin{equation} \ket{\mathcal{B}^{(g)}} = \frac{1}{\mathcal{N}^{(g)}} \sum_{n=0}^{N-1} g^{-n} \ket{n} ,  \end{equation} 
is annihilated by the local projectors defined in Eq. \ref{eq:two_qubit_projector}: 
\begin{equation} 
\hat{P}^{(g)} = \ket{11}\bra{11} + \ket{\psi^{(g)}} \bra{\psi^{(g)}}, \quad \ket{\psi^{(g)} } = \frac{1}{\sqrt{1 + |g|^2}} (g^*\ket{01} - \ket{10}), \quad \ket{\psi^{(g)}_\perp} = \frac{1}{\sqrt{1 + |g|^2}} (\ket{01} + g\ket{10})
\end{equation}
and is therefore a zero-energy ground state of the PVBS reference Hamiltonian $\hat{\mathbb{H}}_+^{(g)} = \sum_{n} \hat{P}_{n,n+1}^{(g)}$. Recall that, by construction, $\hat{P}_{n,n+1}^{(g)}\ket{\psi^{(g)}_\perp} = 0$. Using this:
\begin{eqnarray}
    \hat{P}^{(g)}_{n,n+1} \ket{\mathcal{B}^{(g)}} &=& \hat{P}^{(g)}_{n,n+1} \frac{1}{\mathcal{N}^{(g)}} \sum_{n'=0}^{N-1} g^{-n'} \ket{n'} \notag = \frac{1}{\mathcal{N}^{(g)}}\hat{P}^{(g)}_{n,n+1} (g^{-n}\ket{n} + g^{-n-1}\ket{n+1}) \\ &=& \frac{1}{\mathcal{N}^{(g)}} \frac{\sqrt{1+|g|^2}}{g^{-n-1}} \hat{P}^{(g)}_{n,n+1} (\ket{0}^{\otimes n-1} \otimes\ket{\psi^{(g)}_\perp}_{n,n+1} \otimes \ket{0}^{\otimes N-n-1}) = 0.
\end{eqnarray}
In the second equality we have used the fact that $\hat{P}^{(g)}_{n,n+1} \ket{n'} = 0$ if $n' \ne n, n+1$ and in the last equality we used $\ket{0}^{\otimes n-1}\otimes\ket{\psi^{(g)}_\perp}_{n,n+1} \otimes \ket{0}^{\otimes N-n-1} = \frac{g^{-n-1}}{\sqrt{1+|g|^2}} [\ket{n+1} + g\ket{n}]$. Therefore, the state $\ket{\mathcal{B}^{(g)}}$ is annihilated by $\hat{P}_{n, n+1}^{(g)}$ for any $n$. As such, it is a zero energy eigenstate of any Hamiltonian of the form $\hat{\mathbb{H}} = \sum_{n = 0}^{N-2} \hat{H}_{n, n+1}^{(g)}$, where $\hat{H}_{n, n+1}^{(g)} = \hat{P}^{(g)}_{n,n+1}\hat{h}_{n,n+1}\hat{P}^{(g)}_{n,n+1}$, and of any circuit where the local unitaries have the form $\hat{U}_{n,n+1} = \text{exp}[i\hat{H}^{(g)}_{n,n+1}]$.

\subsection{Critical properties of the boundary QMBS}

To quantify the degree and position of the localisation, we introduce an ``imbalance'' operator, defined as: \begin{equation} \hat{\mathcal{I}} = \sum_{i = 0}^{N/2-1} \hat{n}_i - \sum_{i = N/2}^{N-1} \hat{n}_i, \qquad \hat{n}_i = \frac{\hat{I}_i - \hat{Z}_i}{2}, \end{equation} where $\hat{n}_i$ is the occupation operator at site $i$. The expectation value of the imbalance with respect to the state $\ket{\mathcal{B} ^{(g)}} $ can be computed exactly in terms of $g$ and $N$: $\langle  \mathcal{B} ^{(g)}  |\; \hat{\mathcal{I}} \; |\mathcal{B} ^{(g)} \rangle(N)  = \frac{1 - |g^{-1}|^N}{1 + |g^{-1}|^N}$. Let us introduce $q = g^{-1}$ in order to simplify the calculations:
\begin{eqnarray}
    \notag \bra{\mathcal{B}^{(g)}} \; \hat{\mathcal{I}} \; \ket{\mathcal{B}^{(g)}} &=& \frac{1}{\mathcal{N}^2} (\sum_{n = 0}^{N/2 - 1}|q|^{2n} - \sum_{n = N/2}^{N - 1}|q|^{2n}) 
    = \frac{1}{\mathcal{N}^2} (|q|^2 \frac{1-|q|^N}{1-|q|^{2}} - |q|^{N+2} \frac{1-|q|^N}{1-|q|^{2}}) \\ \notag
    &=& \frac{1}{|q|^2} \frac{1- |q|^2}{1 - |q|^{2N}}\frac{1-|q|^N}{1 - |q|^2} (|q|^2 - |q|^{N+2})  = \frac{(1-|q|^N)^2}{1-|q|^{2N}} = \frac{1 - |g^{-1}|^N}{1 + |g^{-1}|^N}
\end{eqnarray}

\begin{figure}[ht]
    \centering
    \begin{minipage}{0.48\textwidth}
        \centering
        \includegraphics[width=\linewidth]{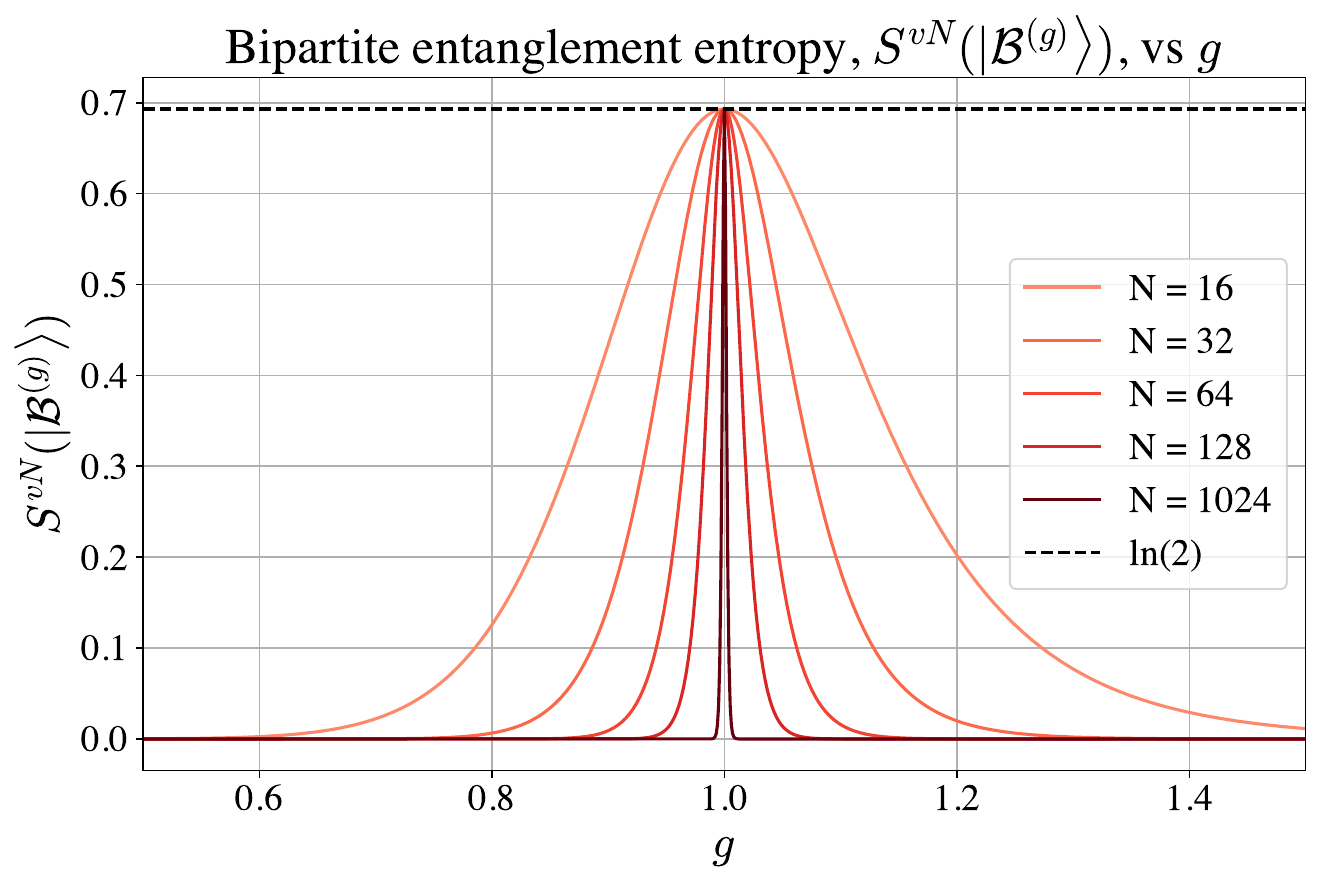}
    \end{minipage}
    \hfill
    \begin{minipage}{0.48\textwidth}
        \centering
        \includegraphics[width=\linewidth]{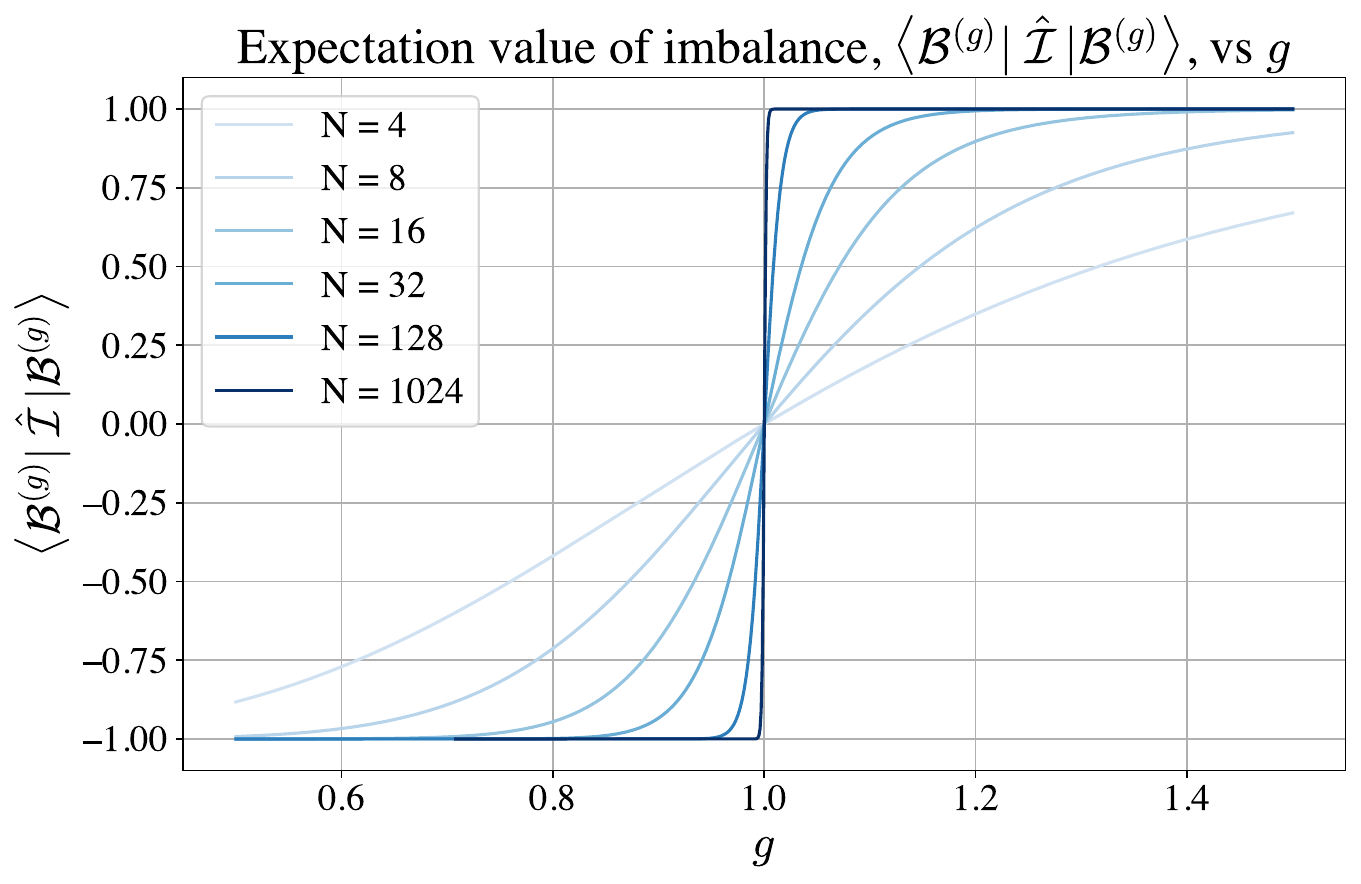}
    \end{minipage}
    \caption{LEFT: Bipartite entanglement entropy, $S^{vN}(\ket{\mathcal{B}^{(g)}})$, as a function of $g$ for different system sizes $N$. Again we observe a sharper and sharper peak as $N$ is increased, displaying a drastic change as $g$ is tuned across the critical point $g_c = 1$. RIGHT: Expectation value of imbalance, $\langle  \mathcal{B} ^{(g)}  |\; \hat{\mathcal{I}} \; |\mathcal{B} ^{(g)} \rangle(N)$, as a function of $g$, plotted for various system sizes, $N = 4, \; 8, \; 16, \; 32, \; 128, \; 1024$. As can be predicted analytically, the form of the curve is approaching the a step function as $N$ is increased. This change in the imbalance directly shows the localisation properties of $|\mathcal{B} ^{(g)} \rangle $.}
    \label{fig: edge_scar_properties}
\end{figure}

This expectation value is plotted in Fig. \ref{fig: edge_scar_properties}, showing a sharper transition at $g =1$, as the system size $N$ is increased. In the limit $N \rightarrow \infty$, this becomes a step function $\lim_{N \rightarrow \infty }\langle  \mathcal{B} ^{(g)}  |\; \hat{\mathcal{I}} \; |\mathcal{B} ^{(g)} \rangle (N) =  2\;\Theta(g) - 1$, where $\Theta (g)$ is the Heaviside step function. This showcases critical behaviour in the properties of the embedded QMBS state. For very large values of $N$, $\ket{\mathcal{B}^{(g)}}$ will be exponentially localised at either boundary, except at $g = 1$, where it becomes an equal superposition of all single excitation states, i.e. the $N$-qubit $W$ state, delocalised over the whole chain. 

We can also characterize this transition via the localization length $\xi _{\text{loc}} (g)$, and its divergence around the critical point. First, we note that due to the particularly simple structure of $| \mathcal{B} ^{(g)} \rangle$, constrained in the single excitation subspace, the localization length scale is the same as the correlation length. Using the exact form of the ground state: $|\mathcal{B}^{(g)}\rangle = \frac{1}{\mathcal{N}} \sum_{n=0}^{N-1} g^{-n} |0...0 1_n 0...0\rangle$, and rewriting $g^{-n} = e^{-\ln(g)n}$, we can immediately identify the localisation length as $e^{-n/\xi_{\text{loc}}} = e^{-\ln(g)n} \rightarrow \xi_{\text{loc}} = \frac{1}{|\ln (g)|}$. Evidently the localization length diverges at $g_c = 1$. By expanding around the latter, $\ln(g) \approx g - 1 = g - g_c \implies \xi_{\text{loc}} = \frac{1}{|g - g_c|}$, we see that the critical exponent associated with the localisation (and consequently correlation) length is $\nu = 1$.

Finally, we note that this critical behaviour can also be seen in the bipartite entanglement entropy of the boundary QMBS. We will again use $q = g^{-1}$ for brevity. The $\ket{\mathcal{B} ^{(g)}}$ state can be split across the bipartition as:
\begin{eqnarray}
    \ket{\mathcal{B}^{(g)}} &=& \frac{1}{\mathcal{N}} \sum_{n = 0}^{N-1}q^n \hat{X}_n\ket{0}^{\otimes N} = \frac{1}{\mathcal{N}}[(\sum_{n = 0}^{N-1} q^n \hat{X}_n\ket{0}^{\otimes N/2}) \otimes \ket{0}^{\otimes N/2} + \ket{0}^{\otimes N/2} \otimes (\sum_{n = 0}^{N-1} q^{N/2+n} \hat{X}_n\ket{0}^{\otimes N/2}) ].
\end{eqnarray}
We now introduce the two states:
\begin{equation}
    \ket{\chi^{(q)}} = \sum_{n = 0}^{N-1} q^n \hat{X}_n\ket{0}^{\otimes N/2}, \quad \ket{\varnothing} = \ket{0}^{\otimes [N/2]}.
\end{equation}
Note that $\ket{\chi^{(q)}}$ is not normalized, $\langle \chi^{(q)}\ket{\chi^{(q)}} = |q|^2 \frac{1 - |q|^N}{1 - |q|^2}$. Let $\ket{\tilde{\chi} ^{(q)}} = \alpha \ket{\chi ^{(q)}} \, (\text{with }\alpha = \frac{1}{|q|} \sqrt{\frac{1- |q|^2}{1-|q|^N}}),$ be the corresponding normalized state. We refer to the first $N/2$ sites as subsystem $A$, and the rest of the chain as subsystem $B$. Now:
\begin{eqnarray}
    \ket{\mathcal{B}^{(g)}} &=& \frac{1}{\mathcal{N}}[\ket{\chi ^{(q)}}_A \otimes \ket{\varnothing}_B + q^{N/2} \ket{\varnothing}_A \otimes \ket{\chi ^{(q)}}_B] \implies \notag \\ \hat{\rho}^{(g)} &=& \ket{\mathcal{B}^{(g)}} \bra{\mathcal{B}^{(g)}} = \frac{1}{\mathcal{N}^2} [\ket{\chi ^{(q)}}\bra{\chi ^{(q)}}_A\otimes \ket{\varnothing}\bra{\varnothing}_B + |q| ^{N}\ket{\varnothing}\bra{\varnothing}_A \otimes \ket{\chi ^{(q)}}\bra{\chi ^{(q)}}_B \notag \\&&+ (q^*)^{N/2} \ket{\chi ^{(g)}}\bra{\varnothing}_A \otimes \ket{\varnothing}\bra{\chi ^{(q)}}_B + q^{N/2} \ket{\varnothing}\bra{\chi ^{(q)}}_A \otimes  \ket{\chi ^{(q)}}\bra{\varnothing}_B ]. \label{eq: density_mat_psi_g}
\end{eqnarray}
Since $\bra{\chi ^{(q)}} \varnothing \rangle = 0$, it is evident that, when tracing over subsystem $B$, the last two terms in Eq.~\eqref{eq: density_mat_psi_g} will vanish. Inserting the normalized state $\ket{\tilde{\chi}^{(q)}}$ into the obtained expression, we have:

\begin{eqnarray}
    \hat{\rho}^{(g)} &=& \frac{1}{\mathcal{N}^2}[\frac{1}{\alpha^2} \ket{\tilde{\chi}^{(q)}} \bra{\tilde{\chi}^{(q)}}_A \otimes \ket{\varnothing} \bra{\varnothing}_B +  \frac{|q|^N}{\alpha^2} \ket{\varnothing} \bra{\varnothing}_A \otimes \ket{\tilde{\chi}^{(q)}} \bra{\tilde{\chi}^{(q)}}_B + ...] \implies \notag \\
    \hat{\rho}_A^{(g)}&=& \text{Tr}_B[\hat{\rho}^{(g)}] =  \frac{1}{\mathcal{N}^2 \alpha^2}[  \ket{\tilde{\chi}^{(q)}} \bra{\tilde{\chi}^{(q)}}_A + |q|^N \ket{\varnothing} \bra{\varnothing}_A ] = p_1 \ket{\tilde{\chi}^{(q)}} \bra{\tilde{\chi}^{(q)}}_A + p_2 \ket{\varnothing} \bra{\varnothing}_A \implies \notag \\
     S^{vN}(\ket{\mathcal{B}^{(g)}}) &=& -p_1 \ln p_1 -p_2 \ln p_2, \; \text{where}: \quad p_1 = \frac{1}{1+|g^{-1}|^N}, \quad p_2 = \frac{|g^{-1}|^N}{1+|g^{-1}|^N}.
\end{eqnarray}

From this general expression we can also read of $S^{vN}(\ket{\mathcal{B}^{(|g| = 1)}}) = \ln 2$. It is also evident that in the limit $N \rightarrow \infty$, we have either $p_1 \rightarrow 1, \; p_2 \rightarrow 0$ (in the case $|g| < 1$) or $p_1 \rightarrow 0, \; p_2 \rightarrow 1$ (for $|g| > 1$). In both cases we have $S^{vN}(\ket{\Psi ^{(|g| \neq 1)}}) \rightarrow 0$, which can also be seen clearly in Fig. \ref{fig: edge_scar_properties}.

\section{Gapless excitations of the PVBS reference Hamiltonian}
\label{app: aqmbs_concrete_model}

In this section, we derive the form of the AQMBS presented in Eqs. \ref{eq: aqmbs_amplitudes} and \ref{eq: aqmbs_energies}. First we recall the form of the PVBS reference Hamiltonian:
\begin{equation} 
\hat{\mathbb{H}}^{(g)}_+ = \sum_{n=0}^{N-2} \hat{P}^{(g)}_{n,n+1} . 
\end{equation} 
This Hamiltonian is positive semi-definite and also frustration-free since the product state $\ket{\mathcal{V}}$ and the boundary state $\ket{\mathcal{B}^{(g)}}$ are zero-energy eigenstates of all the individual projectors. It was shown in Ref. \cite{Bra-15} that it is also gapless at $|g|= 1$, hosting a set of low-lying energy states $\ket{\mathcal{A}_k}, \, k \in \mathbb{Z}_+$, satisfying $\hat{\mathbb{H}}_+ \ket{\mathcal{A}_k} \stackrel{N\to\infty}{\longrightarrow} 0$. Therefore, as shown in \ref{app: general_proof_aqmbs}, these states will form asymptotic QMBS for $\hat{\mathbb{H}}$.


In order to obtain the exact form of the $\ket{\mathcal{A}_k}$ states, we note that $\hat{\mathbb{H}}_+$ commutes with the total magnetisation: 

\begin{equation}
    \big [ \hat{\mathbb{H}}_+, \,  \sum_{n = 0}^{N-1} \hat{Z}_n \big ],
\end{equation}
since each individual projector $\hat{P}^{(g)}_{n, \, n+1}$ cannot change the total magnetisation. Therefore $\hat{\mathbb{H}}_+$ will be block diagonal in the computational basis. In the sub-block in the single excitation subspace, we obtain a tight-binding model with energy off-sets at the boundary sites. Note that we will write the latter in a different basis, where $\ket{n}$ corresponds to the state $\hat{X}_n \ket{0}^{\otimes N}$ of the original computational basis, which we used for the full Hamiltonian. In the single excitation subspace, we then have:

\begin{eqnarray}
    \hat{\mathbb{H}}_+^{\text{SE}} = \epsilon_0 \sum_{n = 1}^{N-2}\ket{n}\bra{n} +\sum_{n = 0}^{N-2}(t\ket{n}\bra{n+1} + t^*\ket{n + 1}\bra{n}) \notag + \Delta_0 \ket{0}\bra{0} + \Delta_{N-1} \ket{N-1}\bra{N-1},
\end{eqnarray}
where $\epsilon_0 = 1$ is the on-site energy, $t = -\frac{g}{1+ |g|^2}$ is the hopping term, and $\Delta_0 = \frac{1}{1+ |g|^2}, \Delta_{N-1} = \frac{|g|^2}{1+ |g|^2}$ are the energy off-sets at the left/right hand-side boundaries. We note that the ground state energy gap, $\Delta \equiv \varepsilon_1 - \varepsilon_0$, for $\hat{\mathbb{H}}_+$ is closing at $|g| = 1$, as shown in Fig. \ref{fig: gap_v_g}.

\begin{figure}
    \centering
    \includegraphics[width=0.5 \linewidth]{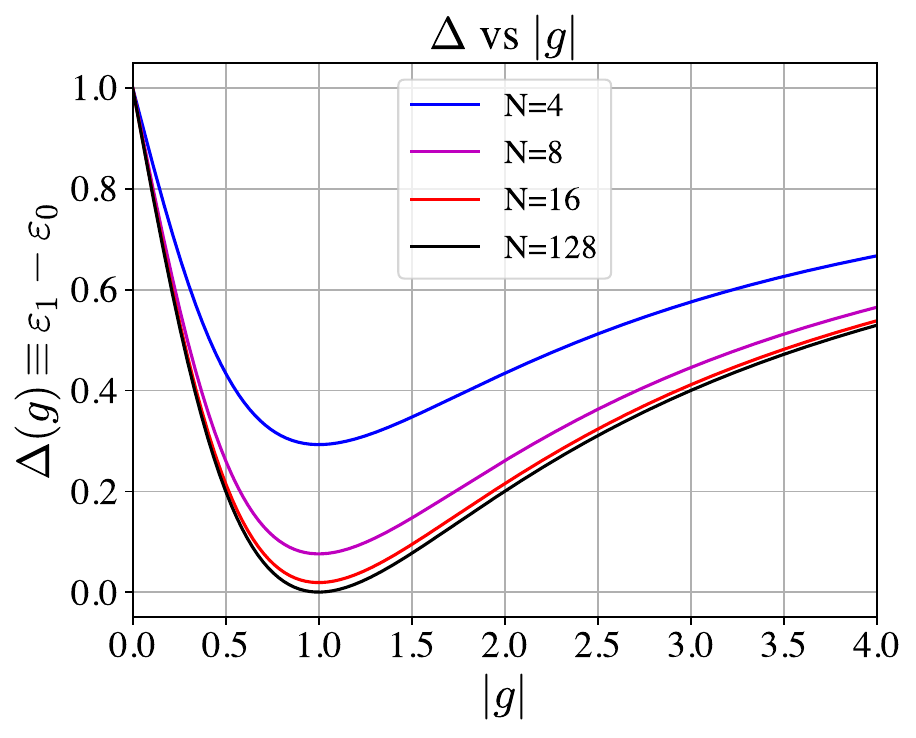}
    \caption{Gap between the first excited state and the ground state of $\hat{\mathbb{H}}_{+}^{(g)}$ as a function of $g$ in the range $g \in [0, 4]$, for system sizes $N = 4, \; 8, \; 16, \; 256$. The gap is closing at $|g| = 1.0$, due to the gapless excitation $\ket{\mathcal{A}_1}$.}
    \label{fig: gap_v_g}
\end{figure}

At $g = e^{i\phi}$, we can also find the eigenstates and corresponding eigenvalues of $\hat{\mathbb{H}}_+^{\text{SE}}$ in a closed-form expression. First we note that for $g = 1$, the eigenstates and eigenvalues of $\hat{\mathbb{H}}_+^{\text{SE}, \,(g= 1)}$ have the form \cite{Doo-20}:  $\ket{\mathcal{A}_k'} = \frac{1}{\mathcal{N}_k} \sum_{n=0}^{N-1} \cos\left[\left(N-n-\frac{1}{2}\right)\frac{k \pi}{N} \right] \ket{n}, \quad\varepsilon_k' = 1 - \cos\frac{k\pi}{N}$. Now, for $g = e^{i\phi}$, we notice that the resulting Hamiltonian satisfies $\hat{U} \,\hat{\mathbb{H}}_+^{\text{SE}, \,(g= \text{exp}[i\phi])} \hat{U}^{\dagger} = \hat{\mathbb{H}}_+^{\text{SE}, \,(g= 1)}$ with $\hat{U} = \text{exp}[i\phi \sum_{n = 0}^{N-1}n\ket{n}\bra{n}]$. Hence, $\hat{\mathbb{H}}_+^{\text{SE}, \,(g= \text{exp}[i\phi])}$ and $\hat{\mathbb{H}}_+^{\text{SE}, \,(g= 1)}$ will have the same eigenvalues, and the eigenstates of the former will be $\ket{\mathcal{A}_k} = \hat{U} \ket{\mathcal{A}_k'}$. We can thus write the eigenstates and eigenvalues of $\hat{\mathbb{H}}_+^{\text{SE}, \,(g= \text{exp}[i\phi])}$ as:
\begin{eqnarray}
    \ket{\mathcal{A}_k} &=& \frac{1}{\mathcal{N}_k} \sum_{n=0}^{N-1} e^{in\phi} \cos\left[\left(N-n-\frac{1}{2}\right)\frac{k \pi}{N} \right] \ket{n}
    \\
    \varepsilon_k &=& 1 - \cos\frac{k\pi}{N},
\end{eqnarray}
where $k \in \{ 0, 1,..., N - 1\}$ and $\mathcal{N}_k$ is a normalisation factor. The lowest energy eigenstate ($k=0$) in this subspace corresponds to the boundary state $\ket{\mathcal{A}_0} =  \ket{\mathcal{B}^{(|g|= 1)}}$, which is a ground state of $\hat{\mathbb{H}}_+$ with $\varepsilon_0 = 0$. Importantly, if $k > 0$ is held fixed as $N$ increases, the energy gap $\varepsilon_k$ closes, corresponding to gapless excitations of $\hat{\mathbb{H}}_+^{(|g|=1)}$ in the thermodynamic limit. Hence, $\ket{\mathcal{A}_k}$ with fixed $k > 0$, will form AQMBS of $\hat{\mathbb{H}}$.

\section{State Preparation Details}
\label{app: state_prep}

We presented the AQMBS state preparation circuit structure for $N = 8$ in the main text, Fig. \ref{fig:circuits_and_initial_states}(c). For $N = 2^k, \, k \in \mathbb{Z}_+, \, k > 3$, we can recursively construct the state preparation circuits from the $N = 8$ case. For each subsequent $k$, we perform $X$ on the first qubit, and append the entangling gates from the $k-1$ state preparation construction to the first $2^{k-1}$ qubits, as well as the last $2^{k-1}$ qubits separately, adjusting the rotation angles appropriately. Hence, doubling the system size requires only one additional layer of entangling gates, leading to logarithmic scaling of the state preparation process. This can already be seen in the $N = 8$ case: The $U^{\text{sp}}_{0,2}, U^{\text{sp}}_{0,1}, U^{\text{sp}}_{2,3}$ gates are precisely the entangling operations required for $N = 4$, plus the $X$ gate at the start.  
In the case where $N \ne 2^k$, we can use an adapted version of the state preparation circuit for $N' = 2^{\lceil \log _2 N\rceil}$. This will result in a circuit depth $\lceil \log _2 N\rceil$, using $N-1$ entangling gates. 

We now turn to the concrete protocols used in the performed experiments. 
For clarity and convenience, we will write the $2$-local gates used for the state preparation circuits in terms of the rotation angle within the $\{\ket{01}, \, \ket{10} \}$ subspace:

\begin{equation}
    {U}_{\text{q1, q2}}^{\text{sp}} (\gamma_i) = \begin{bmatrix}
1 & 0 & 0 & 0 \\
0 & \cos (\gamma_i) & -\sin (\gamma_i) & 0 \\
0 & \sin (\gamma_i) & \cos (\gamma_i) & 0 \\
0 & 0 & 0 & 1
\end{bmatrix},
\end{equation}
where the q1, q2 indices denote the labels of the two qubits in the register acted on by $U^{\text{sp}}_{\text{q1, q2}} (\gamma_i)$. Note that $\gamma_i$ is distinct from, but related to the parameters depicted in Fig. \ref{fig:circuits_and_initial_states} ($\{\eta_j, \vec{\theta}_k  \}$).

Using this convention, the state preparation protocols in our work can be written out concisely, displayed below. Note that we separate the gates from different layers using ``;'' for clarity. In practice, the $X$ gate at the start of the protocol and the $Z_i$ gates at the end can be absorbed into the decomposition of $\hat{U}^{\text{sp}}_{\text{q1, q2}}$ shown in Fig. \ref{fig:circuits_and_initial_states}

\begin{itemize}
    \item $N = 8$: $X_0, \, U^{\text{sp}}_{0, 4}(\frac{\pi}{2}); \quad U^{\text{sp}}_{0, 2} (0.429), \, U^{\text{sp}}_{4,6}(1.141); \quad U^{\text{sp}}_{0,1}(0.703), \, U^{\text{sp}}_{2,3}(0.337), \, U^{\text{sp}}_{4,5}(1.233), \, U^{\text{sp}}_{6,7} (0.867)$, $Z_i$ for $i \in \mathbb{Z}, \, 4 \le i \le 7$
    \item $N = 12$: $X_0, \, U^{\text{sp}}_{0, 6}(\frac{\pi}{2}); \quad U^{\text{sp}}_{0, 3} (0.436), \, U^{\text{sp}}_{6,9}(1.135); \quad U^{\text{sp}}_{0,2}(0.523), \, U^{\text{sp}}_{3,5}(0.180), \, U^{\text{sp}}_{6,8}(0.985), \, U^{\text{sp}}_{9,11} (0.683);$ $U^{\text{sp}}_{0,1}(0.750), \, U^{\text{sp}}_{3,4}(0.561), \, U^{\text{sp}}_{6,7}(1.242), \, U^{\text{sp}}_{9,10} (0.861)$, $Z_i$ for $i \in \mathbb{Z}, \, 6 \le i \le 11$
    \item $N = 16$: $X_0, \, U^{\text{sp}}_{0, 8}(\frac{\pi}{2}); \quad U^{\text{sp}}_{0, 4} (0.438), \, U^{\text{sp}}_{8,12}(1.133); \quad U^{\text{sp}}_{0,2}(0.704), \, U^{\text{sp}}_{4,6}(0.370), \, U^{\text{sp}}_{8,10}(1.201), \, U^{\text{sp}}_{12,14} (0.867);$ $ U^{\text{sp}}_{0,1}(0.766),$ $ \, U^{\text{sp}}_{2,3}(0.720), \, U^{\text{sp}}_{4,5}(0.639), \, U^{\text{sp}}_{6,7} (0.326),  U^{\text{sp}}_{8,9}(1.245), \, U^{\text{sp}}_{10,11}(0.931), \, U^{\text{sp}}_{12,13} (0.851), \, U^{\text{sp}}_{14,15}(0.805)$, $Z_i$ for $i \in \mathbb{Z}, \, 8 \le i \le 15$
    \item $N = 20$: $X_0, \, U^{\text{sp}}_{0, 10}(\frac{\pi}{2}); \quad U^{\text{sp}}_{0, 5} (0.439), \, U^{\text{sp}}_{10, 15}(1.132); \quad U^{\text{sp}}_{0,3}(0.600), \, U^{\text{sp}}_{5,8}(0.262), \, U^{\text{sp}}_{10,13}(1.071), \, U^{\text{sp}}_{15,18} (0.759); $ $ U^{\text{sp}}_{0,2}(0.586), \, U^{\text{sp}}_{3,4}(0.728), \, U^{\text{sp}}_{5,7}(0.430), \, U^{\text{sp}}_{8,9} (0.324),  U^{\text{sp}}_{10,12}(0.999), \, U^{\text{sp}}_{13,14}(0.893), \, U^{\text{sp}}_{15,17} (0.680), \, U^{\text{sp}}_{18,19}(0.798); $ $\, U^{\text{sp}}_{0,1}(0.773), \, U^{\text{sp}}_{5,6}(0.677), \, U^{\text{sp}}_{10,11} (1.247), \, U^{\text{sp}}_{15,16}(0.843)$, $Z_i$ for $i \in \mathbb{Z}, \, 10 \le i \le 19$
\end{itemize}

\section{Symmetry of the Model at $g = \pm 1$}
\label{app:additional_symmetry_g1}

One caveat to note is that at $g = \pm 1$ our models derived from the PVBS reference Hamiltonian have an additional symmetry described by the operator $\hat{R}$ (for $g=-1$) or $\hat{\Pi} \hat{R}$ (for $g=1$), with:
\begin{eqnarray}
    \hat{R} : |i_0, \; i_1, ..., i_{N-2}, i_{N-1} \rangle \rightarrow |i_{N-1}, \; i_{N-2}, ..., i_{1}, i_{0}  \rangle, \quad
    \hat{\Pi} = e^{i\pi \sum_j \ket{1} \bra{1}_j}, \notag
\end{eqnarray}
where $i_m \in \{0, 1 \}$ denotes the computational basis state at site $m$. Here, $\hat{R}$ is the standard reflection operator, flipping the chain through its midpoint. The operator $\hat{\Pi}$ acts exactly as the identity for computational basis states with an even number of qubits in the $\ket{1}$ state, while an odd number of qubits in $\ket{1}$ leads to an extra $-1$ factor. 

To show these symmetries at $g = \pm 1$, recall that the Hamiltonian derived from the PVBS reference Hamiltonian has the form $\hat{\mathbb{H}}^{(g)} = \sum_{n=0}^{N-2} \hat{H}_{n,n+1}^{(g)} (h_{n+1})$ with local interactions $\hat{H}_{n,n+1}^{(g)} = \hat{P}^{(g)}_{n, n+1} \hat{h}_{n, n+1}\hat{P}^{(g)}_{n, n+1}$. We can write $\hat{H}_{n,n+1}^{(g)}$ as:
\[
\hat{H}_{n,n+1}^{(g)} = a\ket{\psi ^{(g)}}\bra{\psi ^{(g)}} + b\ket{11}\bra{11} + c\ket{\psi ^{(g)}}\bra{11} + c^*\ket{11}\bra{\psi ^{(g)}}, \quad |\psi^{(g)}\rangle = \frac{1}{\sqrt{1+|g|^2}} \big(g^*|01\rangle - |10\rangle \big),
\]
where we have omitted the possible site dependencies on $a, \, b$ and $c$ for brevity.

We note that the reflection operator $\hat{R}$ swaps the qubits at sites $j$ and $N - j-1$. Writing out the full Hamiltonian:
\begin{equation}
    \hat{\mathbb{H}} = \sum_{j = 0}^{N-2} \{ a|\psi^{(g)}\rangle \bra{\psi^{(g)}} + b \ket{11}\bra{11} \nonumber + c \ket{\psi^{(g)}} \bra{11} +  c^* \ket{11} \bra{\psi^{(g)}} \}_{j, j+1}, 
\end{equation}
and acting on it with the reflection operator we obtain:
\begin{eqnarray}
    \hat{R} \, \hat{\mathbb{H}} \, \hat{R}^{\dagger} &=& \sum_{j = 0}^{N-2} \{a |\psi^{(g)}\rangle \bra{\psi^{(g)}}  + b \ket{11}\bra{11} \nonumber + c \ket{\psi^{(g)}} \bra{11} +  c^* \ket{11} \bra{\psi^{(g)}} \}_{N-j-1,\, N-j-2} \\ &=&
    \sum_{j = 0}^{N-2} \{ a  \hat{S} |\psi^{(g)}\rangle \bra{\psi^{(g)}}\hat{S}^{\dagger}  + b \hat{S}\ket{11}\bra{11} \hat{S}^{\dagger} \nonumber + c \hat{S}\ket{\psi^{(g)}} \bra{11}\hat{S}^{\dagger} +  c^* \hat{S}\ket{11} \bra{\psi^{(g)}} \hat{S}^{\dagger} \}_{N-j-2,\, N-j-1},
\end{eqnarray}
where $\hat{S}$ denotes the $2$-qubit swap operator. We also note that:
\begin{eqnarray}
    \hat{S}\ket{\psi ^{(g)}} = \frac{1}{\sqrt{1 + |g|^2}} (g^*\ket{10} - \ket{01}) =
    -g^*\frac{1}{\sqrt{1 + |g|^2}} ((g^*)^{-1}\ket{01} - \ket{10}) = -g^*\sqrt{\frac{1 + |g^{-1}|^2}{1 + |g|^2}} \ket{\psi ^{(g^{-1})}}, 
\end{eqnarray}
and that $\hat{S}\ket{00} = \ket{00}$, $\hat{S}\ket{11} = \ket{11}$. Using these and relabelling $N - j - 2 = j'$:
\begin{eqnarray}
    \hat{R} \, \hat{\mathbb{H}} \, \hat{R}^{\dagger} &=& \sum_{j' = 0}^{N-2} \{ a  |\psi^{(g^{-1})}\rangle \bra{\psi^{(g^{-1})}} + b \ket{11}\bra{11} \nonumber \\ &&- g^*c \sqrt{\frac{1 + |g^{-1}|^2}{1 + |g|^2}} \ket{\psi^{(g^{-1})}} \bra{11} -  gc^* \sqrt{\frac{1 + |g^{-1}|^2}{1 + |g|^2}}\ket{11} \bra{\psi^{(g^{-1})}}  \}_{j',\, j' + 1}.
\end{eqnarray}
In the case that $g$ has unit modulus ($g = e^{i\phi}$), it further reduces to:
\begin{eqnarray}
    \hat{R} \, \hat{\mathbb{H}} \, \hat{R}^{\dagger} = \sum_{j' = 0}^{N-2} \{a  |\psi^{(g^{-1})}\rangle \bra{\psi^{(g^{-1})}} + b \ket{11}\bra{11} \nonumber  - e^{-i\phi}c \ket{\psi^{(g^{-1})}} \bra{11} - e^{i\phi} c^* \ket{11} \bra{\psi^{(g^{-1})}}  \}_{j',\, j' + 1}.
\end{eqnarray}
For $\phi = \pi$ we have $g = g^{-1} = -1$ and thus obtain:
\begin{eqnarray}
    \hat{R} \, \hat{\mathbb{H}} \, \hat{R}^{\dagger} = \sum_{j' = 0}^{N-2} \{a  |\psi^{(g^{-1})}\rangle \bra{\psi^{(g^{-1})}} + b \ket{11}\bra{11} \nonumber  + c  \ket{\psi^{(g)}} \bra{11} + c^* \ket{11} \bra{\psi^{(g)}}  \}_{j',\, j' + 1} = \hat{\mathbb{H}},
\end{eqnarray}
showing that $\hat{\mathbb{H}}$ is invariant under a reflection about the midpoint for $g = -1$. However, in the case that $\phi = 0$, we have $g = g^{-1} = 1$, leading to:
\begin{eqnarray}
    \hat{R} \, \hat{\mathbb{H}} \, \hat{R}^{\dagger} = \sum_{j' = 0}^{N-2} \{a  |\psi^{(g^{-1})}\rangle \bra{\psi^{(g^{-1})}} + b \ket{11}\bra{11} \nonumber  - c  \ket{\psi^{(g)}} \bra{11} - c^* \ket{11} \bra{\psi^{(g)}}  \}_{j',\, j' + 1}.
\end{eqnarray}
At $g = 1$ the system is therefore not invariant under the pure reflection. We can introduce a parity transformation $\hat{\Pi} = e^{i\pi \sum_j \ket{1} \bra{1}_j}$, which acts on the relevant states as:
\begin{eqnarray}
    \hat{\Pi} \ket{11} \bra{11}_{j, j+1} \hat{\Pi}^{\dagger} = \ket{11} \bra{11}_{j, j+1} ,&& \quad
    \hat{\Pi} \ket{\psi^{(g)}} \bra{\psi^{(g)}}_{j, j+1} \hat{\Pi}^{\dagger} = \ket{\psi^{(g)}} \bra{\psi^{(g)}}_{j, j+1} \nonumber \\
    \hat{\Pi} \ket{\psi^{(g)}} \bra{11}_{j, j+1} \hat{\Pi}^{\dagger} = -\ket{\psi^{(g)}} \bra{11}_{j, j+1} ,&& \quad
    \hat{\Pi} \ket{11} \bra{\psi^{(g)}}_{j, j+1} \hat{\Pi}^{\dagger} = -\ket{11} \bra{\psi^{(g)}}_{j, j+1}.
\end{eqnarray}
Now it is evident that at $g = 1$ we have:
\begin{equation} \label{eq: modified_reflection_symmetry}
    (\hat{\Pi} \hat{R}) \, \hat{\mathbb{H}} \,(\hat{\Pi} \hat{R})^{\dagger} = \hat{\mathbb{H}},
\end{equation}
showing that at $g = 1$, the Hamiltonian is invariant under the combined action of the parity and reflection transformations. Using the same arguments, one can show that the corresponding circuit model $\hat{\mathbb{U}}$ has the same symmetries at $g = \pm 1$, as it consists of the exponentiated local generators. 

\section{Level Spacing Statistics and Bipartite entanglement of eigenstates} \label{app: level_spacing}

As shown in SM, Sec. \ref{app:additional_symmetry_g1}, the Hamiltonian models ($\hat{\mathbb{H}}$) and circuit models ($\hat{\mathbb{U}}$) deformed from the PVBS reference Hamiltonian have an additional symmetry at $g = \pm 1$. More specifically, at $g = 1$, the eigenstates of $\hat{\mathbb{U}}$ (and $\hat{\mathbb{H}}$) will be either even, $\hat{\Pi}\hat{R}\ket{\varphi_\alpha} = \ket{\varphi_\alpha}$, or odd, $\hat{\Pi}\hat{R}\ket{\varphi_\alpha} = -\ket{\varphi_\alpha}$. Therefore, when performing level spacing statistics analysis, we need to resolve this symmetry. The level spacing statistics for two cases of local generators is displayed in Fig. \ref{fig: level_spacing_parity_symmetry}. From now on, we will write the local generators in terms of Pauli matrices acting in the $\{\ket{\psi ^{(g)}}, \, \ket{11} \}$ subspace: $\tilde{\sigma}_x = \ket{\psi^{(g)}}\bra{11} + \ket{11}\bra{\psi^{(g)}}, \tilde{\sigma}_y=i\ket{\psi^{(g)}}\bra{11} -i \ket{11}\bra{\psi^{(g)}}, \tilde{\sigma}_z= \ket{\psi^{(g)}}\bra{\psi^{(g)}} - \ket{11}\bra{11}$, for brevity. Both symmetry subsectors exhibit level spacing statistics in accordance with the Gaussian unitary ensemble, for both sets of probed generators, showing that the model is not integrable. In addition to the level spacing analysis, we also compute the half-system bipartite entanglement entropy for each eigenstate $\ket{\varphi_\alpha}$ of the Floquet unitary $\hat{\mathbb{U}}$. The results are shown in Fig. \ref{fig:bipartite_entanglement}.

\begin{figure}[h!]
    \centering
    \includegraphics[width=0.32\linewidth]{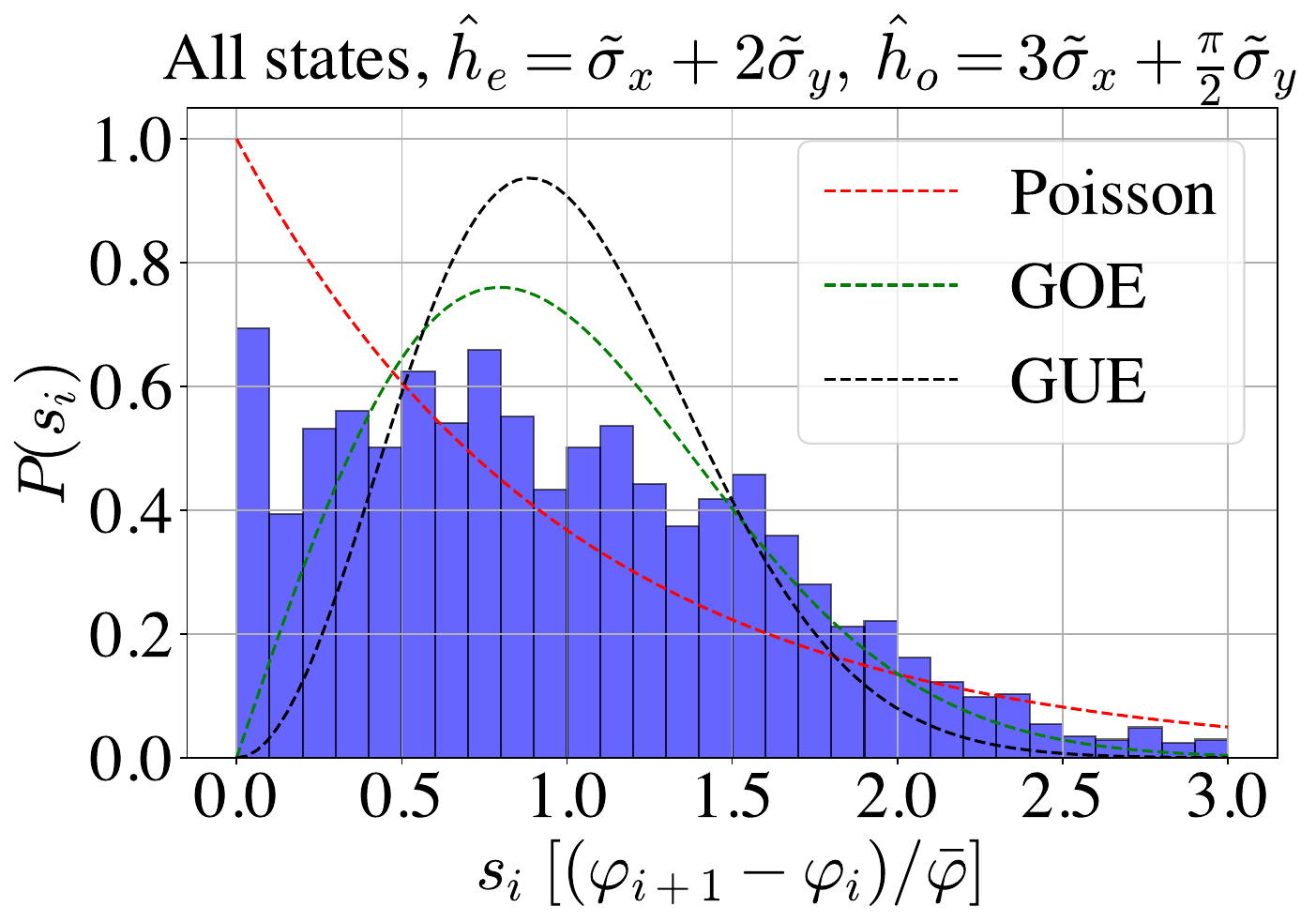}
    \includegraphics[width=0.32\linewidth]{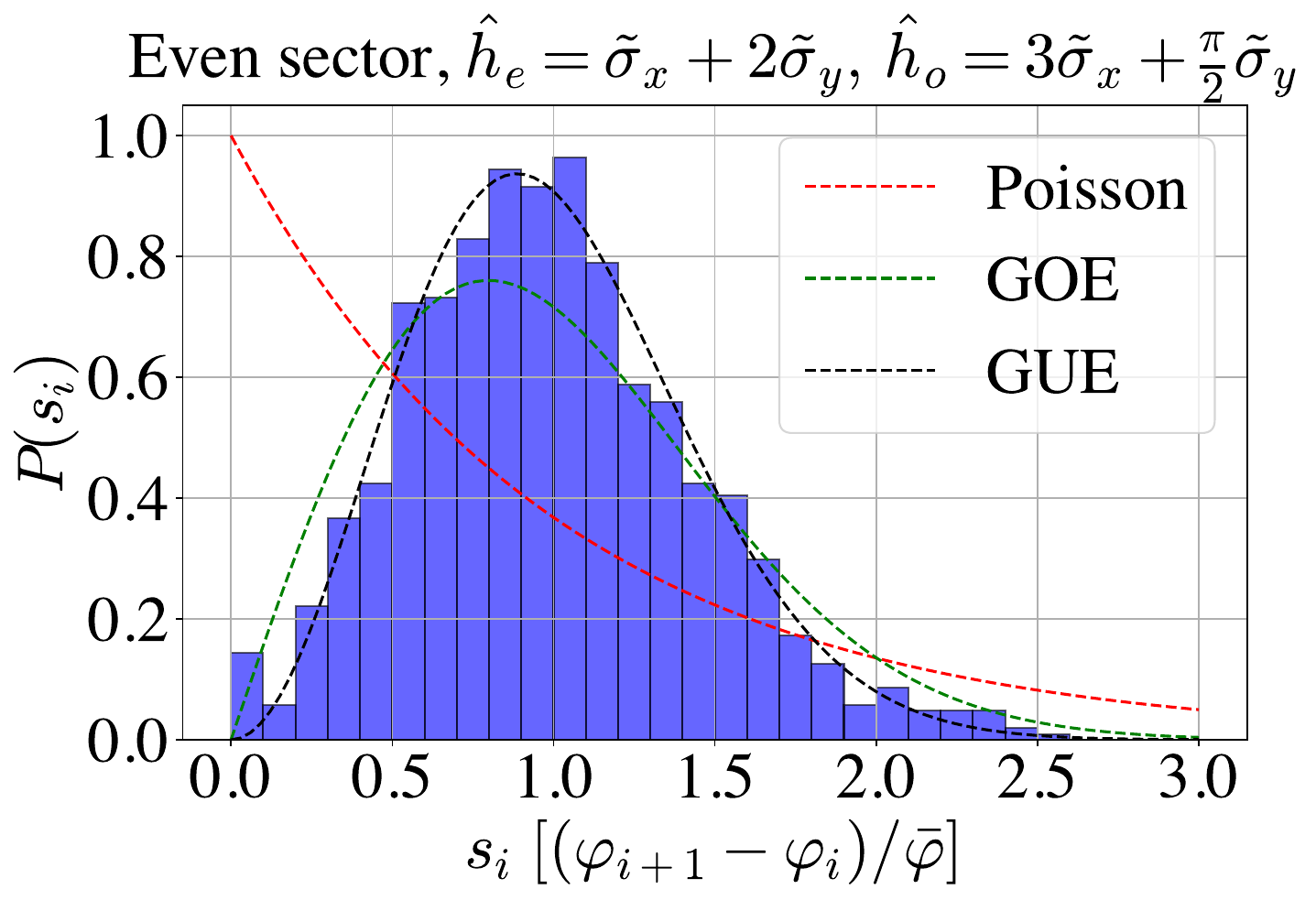}
    \includegraphics[width=0.32\linewidth]{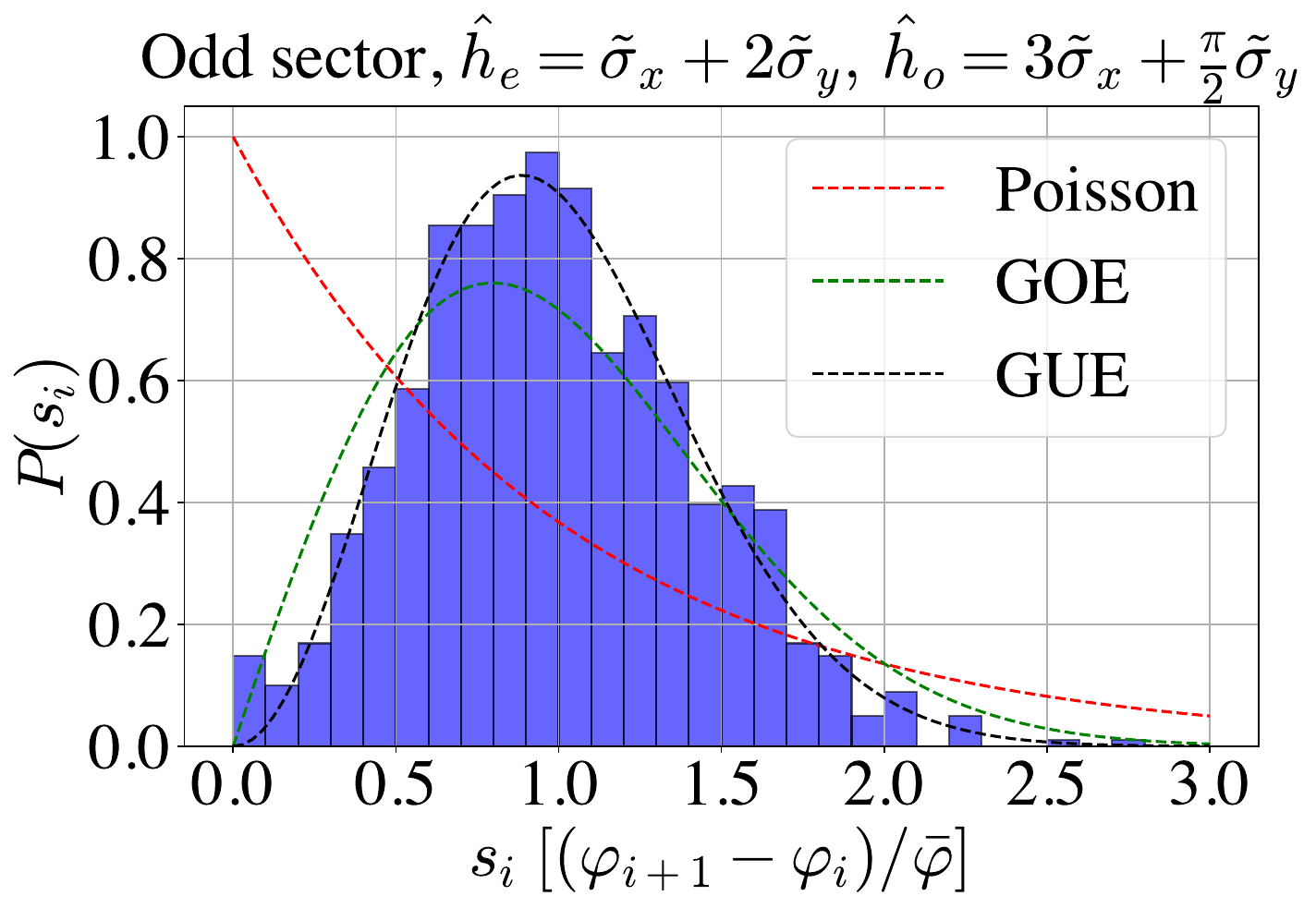}
    \includegraphics[width=0.32\linewidth]{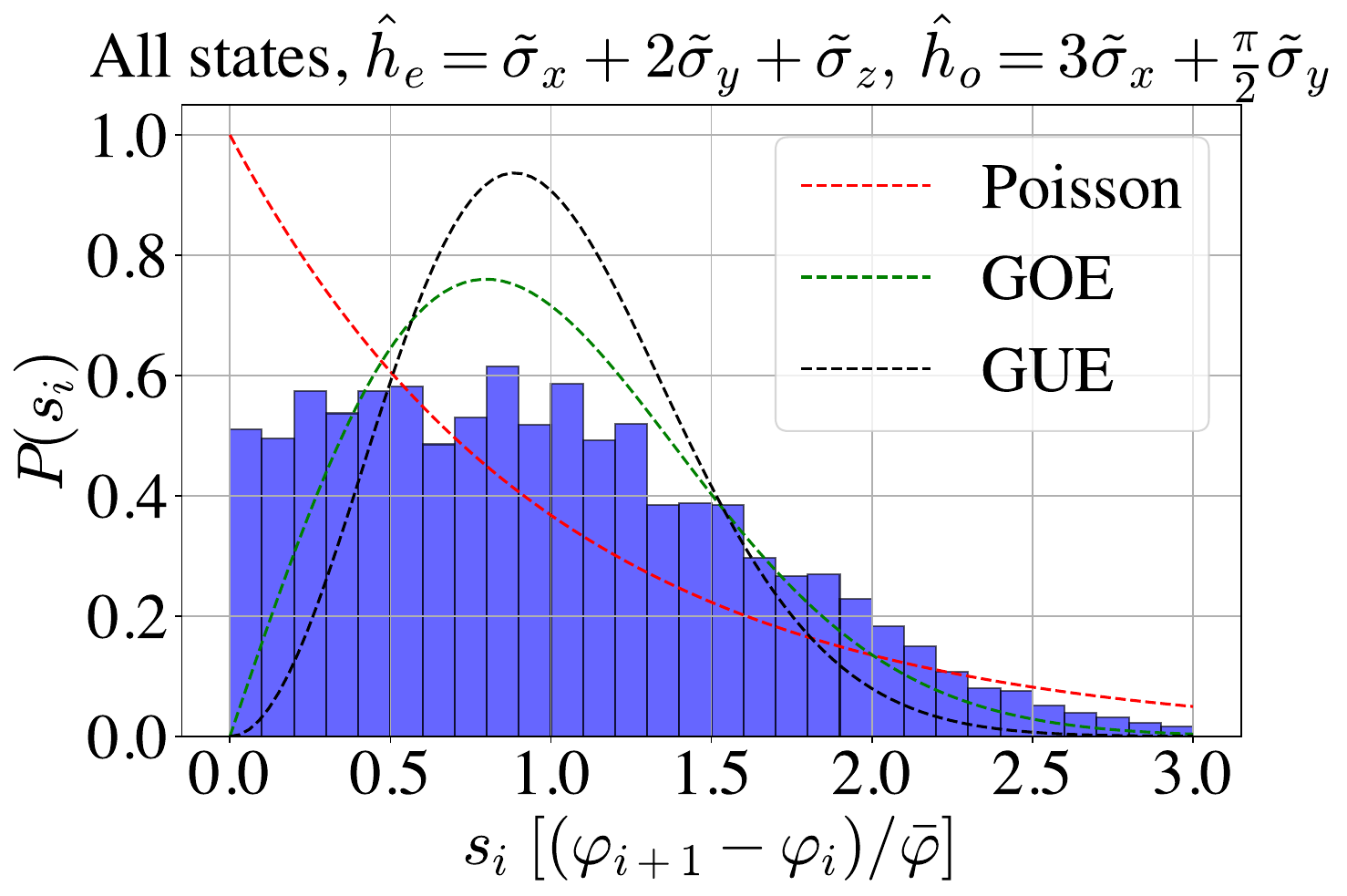}
    \includegraphics[width=0.32\linewidth]{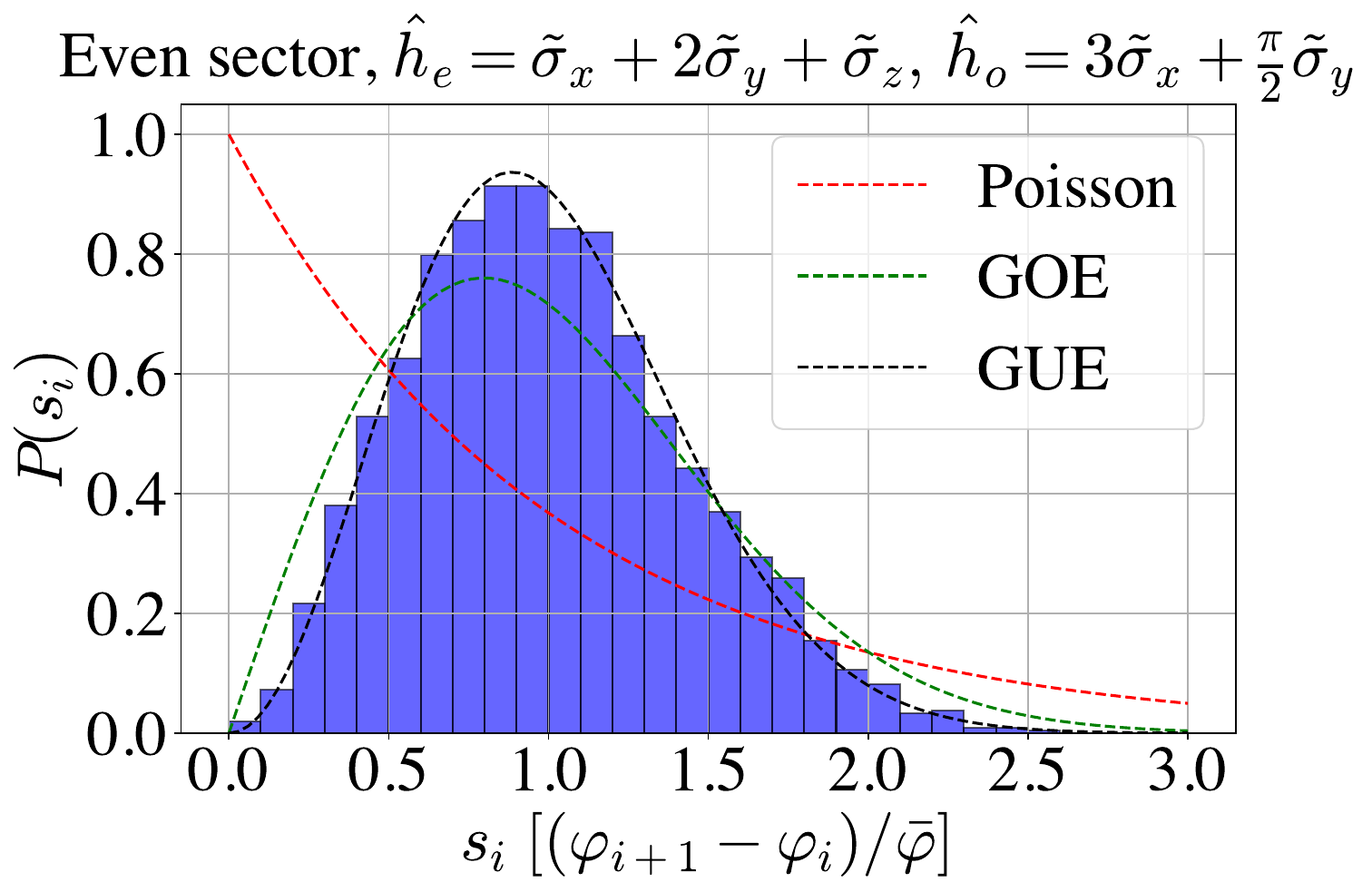}
    \includegraphics[width=0.32\linewidth]{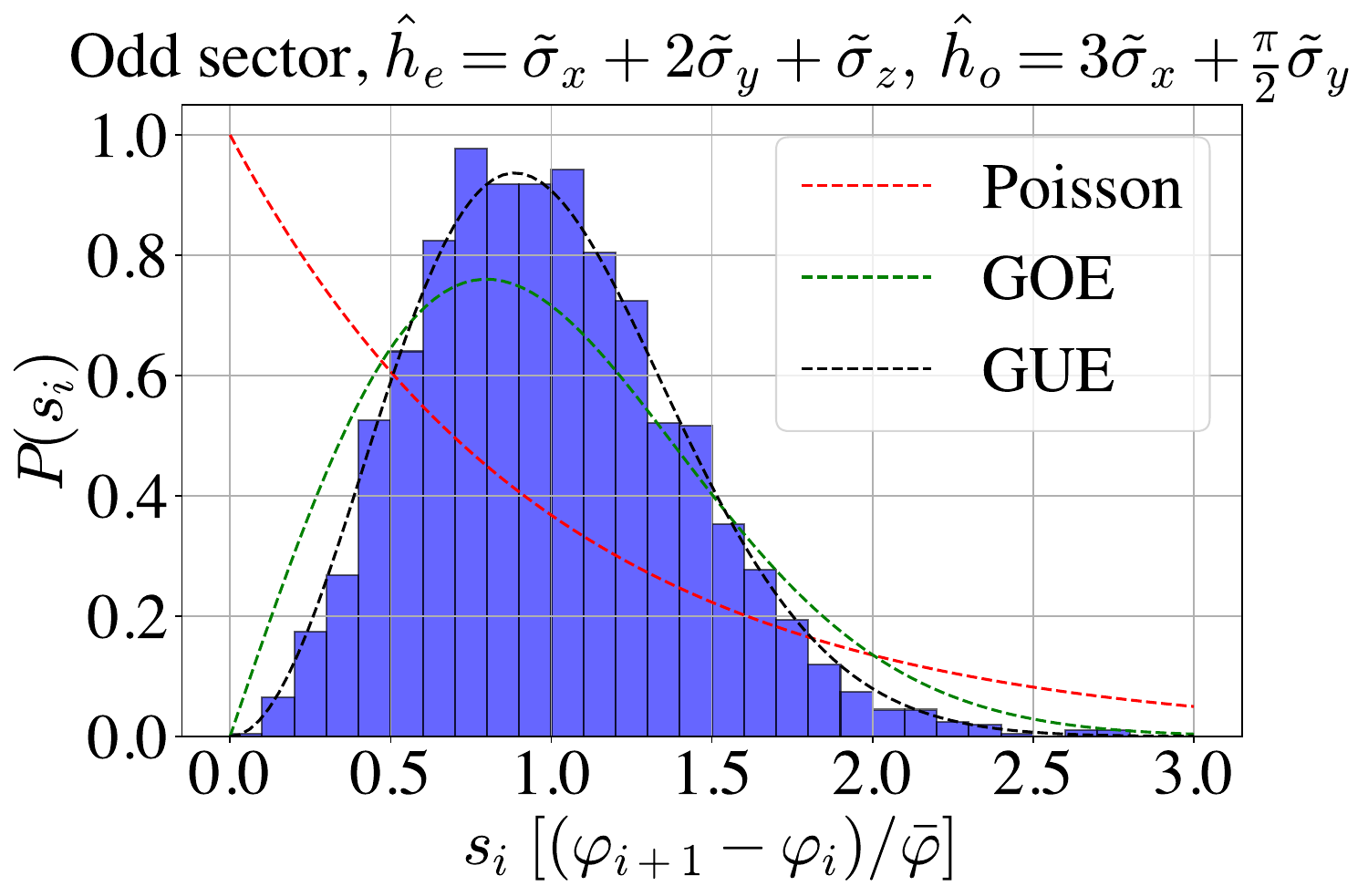}
    \caption{Level spacing statistics for the eigenstates of $\hat{\mathbb{U}}$, at parameter value $g = 1$. TOP: Local generators set to $\hat{h}_e = \tilde{\sigma}_x + 2\tilde{\sigma}_y$, $\hat{h}_o = 3\tilde{\sigma}_x +  \frac{\pi}{2}\tilde{\sigma}_y$, with the three panels displaying statistics over the full spectrum (left), as well as within the even (middle) and odd (right) parity sectors. We note that there are minor kinks close to $P(s_i) = 0$, due to the exponential degeneracy at $\varphi_i = 0$ BOTTOM: Local generators set to $\hat{h}_e = \tilde{\sigma}_x + 2\tilde{\sigma}_y + \tilde{\sigma}_z$, $\hat{h}_o = 3\tilde{\sigma}_y +  \frac{\pi}{2}\tilde{\sigma}_y$, such that there is no degeneracy in the middle of the spectrum.}
    \label{fig: level_spacing_parity_symmetry}
\end{figure}

\begin{figure}
    \centering
    \includegraphics[width=0.48\linewidth]{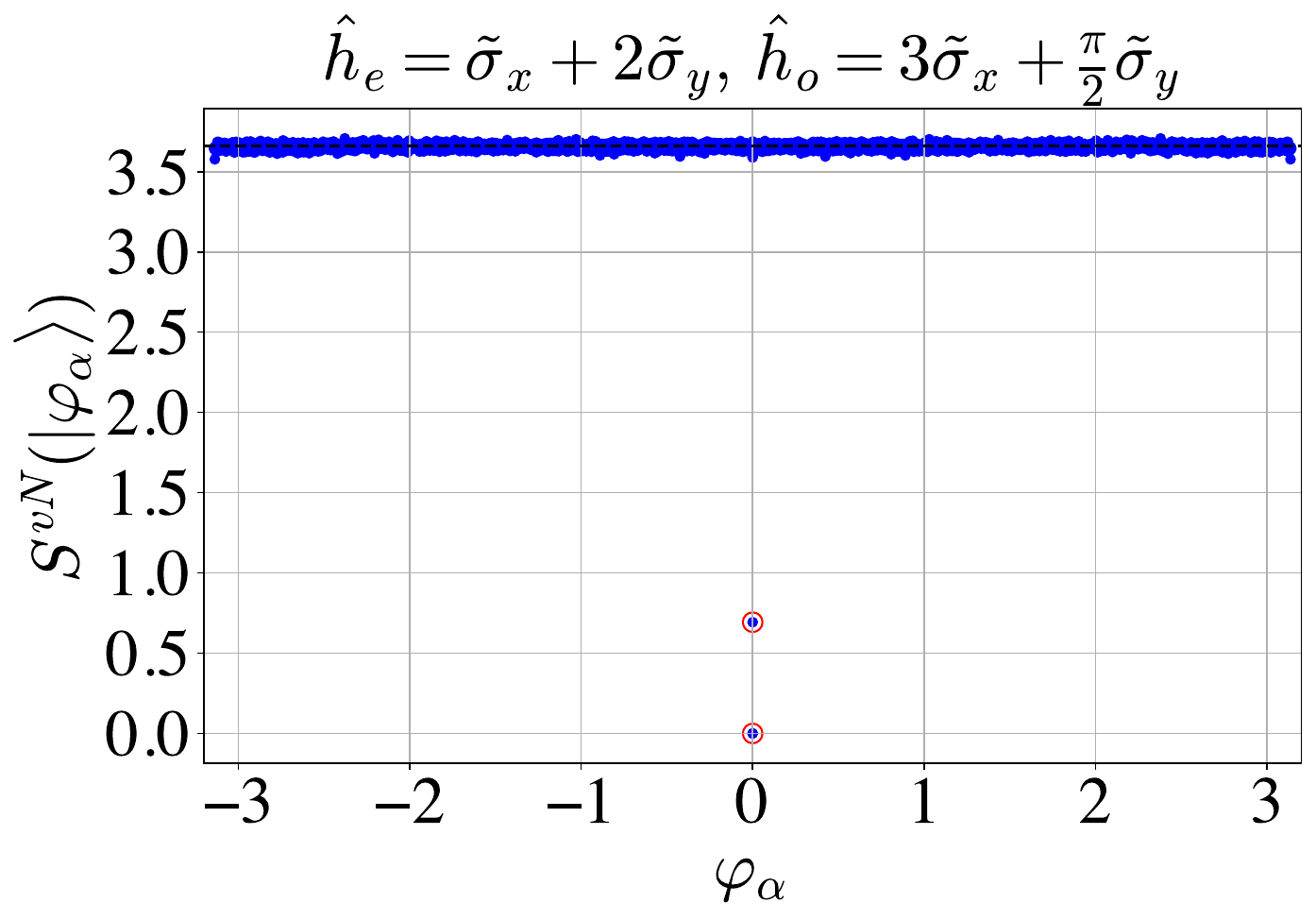}
    \includegraphics[width=0.48\linewidth]{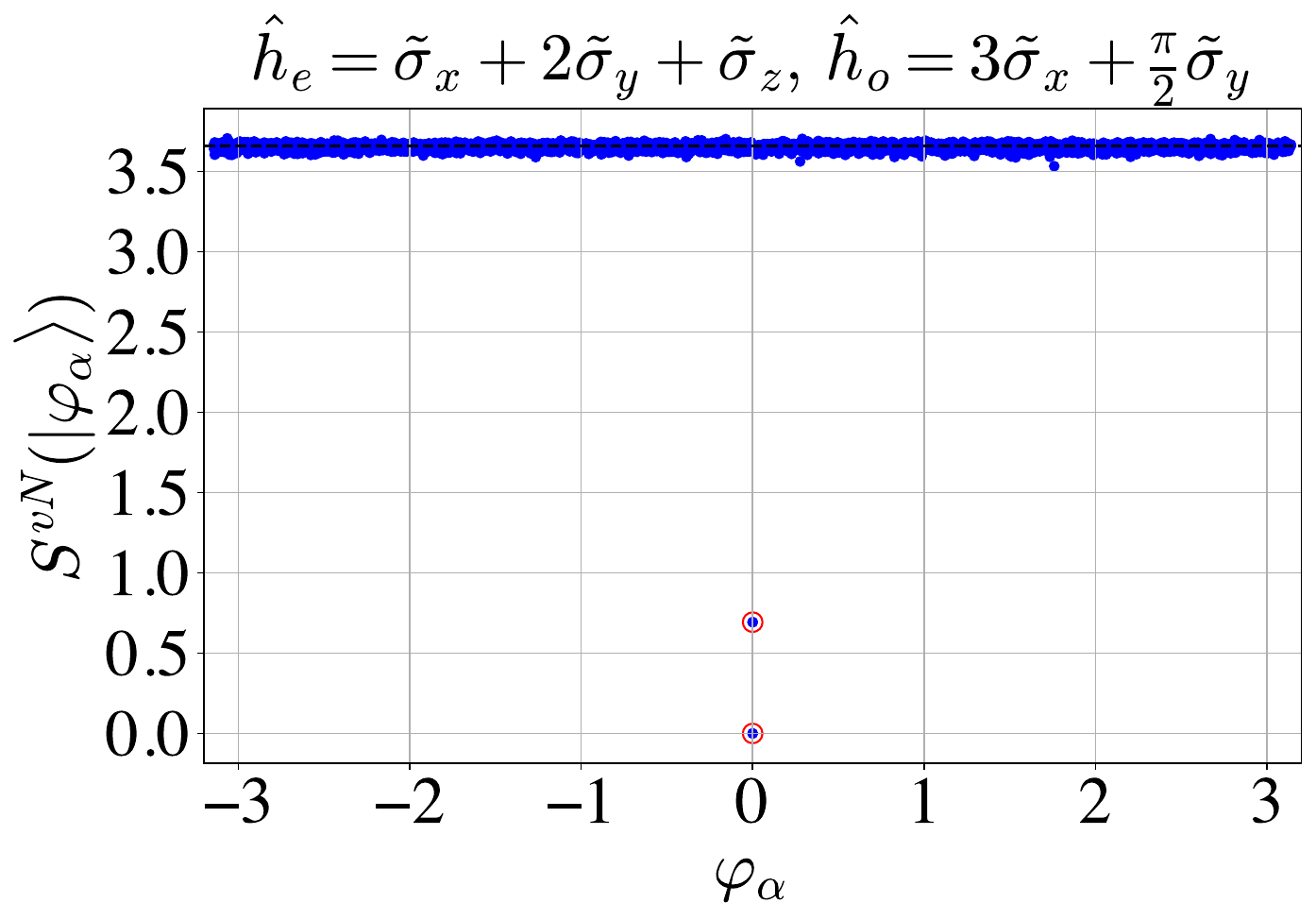}
    \caption{Half-system bipartite entanglement entropy of the eigenstates $\ket{\varphi_\alpha}$ of the Floquet unitary $\hat{\mathbb{U}}$. The entanglement entropy is computed as $\hat{S}(\ket{\varphi_\alpha}) = -{\rm Tr}(\hat{\rho}_\alpha \log \hat{\rho}_\alpha)$, where $\hat{\rho_\alpha} = {\rm Tr}_{1,\hdots,N/2}(\ket{\varphi_\alpha}\bra{\varphi_\alpha})$ is the half-chain reduced density matrix of the Floquet eigenstate $\ket{\varphi_\alpha}$. We probe two sets of generators, $\{ \hat{h}_e = \tilde{\sigma}_x + 2\tilde{\sigma}_y, \, \hat{h}_o = 3\tilde{\sigma}_x +  \frac{\pi}{2}\tilde{\sigma}_y \}$ (left) and $\{ \hat{h}_e = \tilde{\sigma}_x + 2\tilde{\sigma}_y + \tilde{\sigma}_z, \, \hat{h}_o = 3\tilde{\sigma}_x +  \frac{\pi}{2}\tilde{\sigma}_y \}$ (right), which were also considered for the level spacing analysis. We set the system size to $N = 12$ and $g = 1.0$. For both instances we can clearly distinguish the conventional QMBS states $\ket{\mathcal{V}}$ (at $\varphi_\alpha = 0, \, S^{vN} = 0$) and $\ket{\mathcal{B}^{(g)}}$ (at $\varphi_\alpha = 0, \, S^{vN} = \ln (2)$) which are separated from the bulk of the spectrum which is saturated at $S_\text{Page} = [N\ln (d) - 1]/2$}
    \label{fig:bipartite_entanglement}
\end{figure}

\section{Transition to strong ergodicity breaking}
\label{sec:ergodicity_breaking_transition}

We note that for $g = 0$ and $g \rightarrow \infty$, the considered Hamiltonian and circuit models reduce to the kinetically constrained quantum East/West models. Recalling the generic form of the used local interactions $\hat{H}_{n,n+1}^{(g)} = \hat{P}^{(g)}_{n, n+1} \hat{h}_{n, n+1}\hat{P}^{(g)}_{n, n+1}$, we can write $\hat{H}_{n,n+1}^{(g)}$ as:
\[
\hat{H}_{n,n+1}^{(g)} = a\ket{\psi ^{(g)}}\bra{\psi ^{(g)}} + b\ket{11}\bra{11} + c\ket{\psi ^{(g)}}\bra{11} + c^*\ket{11}\bra{\psi ^{(g)}}, \quad |\psi^{(g)}\rangle = \frac{1}{\sqrt{1+|g|^2}} \big(g^*|01\rangle - |10\rangle \big).
\]
For $g = 0$, we obtain $|\psi^{(g=0)}\rangle = - \ket{10}$, such that the local interactions reduce to: $\hat{H}_{n,n+1}^{(g=0)} = a\ket{10}\bra{10} + b\ket{11}\bra{11} - c\ket{10}\bra{11} - c^*\ket{11}\bra{10} = \ket{1}\bra{1} \otimes \hat{h}_{\text{East}}$, with the single qubit operator $\hat{h}_{\text{East}}$ determined by $a, \, b, \, c$. The model composed solely of such interactions will be an instance of so-called quantum East models, possessing kinetic constraints which lead to anomalous transport and ergodicity behaviour \cite{Hor-15, Pan-20}. 

Similarly, for $ g \rightarrow \infty$, we obtain $|\psi^{(g\rightarrow \infty)}\rangle = \ket{01}$, leading to $\hat{H}_{n,n+1}^{(g \rightarrow \infty)} = a\ket{01}\bra{01} + b\ket{11}\bra{11} + c\ket{01}\bra{11} + c^*\ket{01}\bra{10} = \hat{h}_{\text{West}} \otimes \ket{1}\bra{1}$. A model composed only of these types of interactions will be an instance of quantum West models, also dictated by strong kinetic constraints \cite{Bri-24}.

\section{Models with exponential mid-spectrum degeneracy at $g=1$}
\label{app: midspectrum degeneracy}

Note that, for particular choices of $\{ \hat{h}_e,\; \hat{h}_o \}$, the Hamiltonian or corresponding circuit model will have an exponential degeneracy at $E_\alpha (\varphi_\alpha) = 0$ at  $g = 1$. The number of states at $E_\alpha = 0$ is $2^{N/2}$, which was observed for any instance where the local generators have the form $\hat{h} = \alpha_x \tilde{\sigma}_x + \alpha_y \tilde{\sigma}_y$. One particular example of such a model is with $\hat{h}_e = \tilde{\sigma}_x + 2\tilde{\sigma}_y, \; \hat{h}_o = 2\tilde{\sigma}_x + \frac{\pi}{2}\tilde{\sigma}_y$. The plots of the bipartite entanglement entropy versus the energy of the eigenstates for different system sizes are displayed in Fig. \ref{fig:exponential_zero_mode_degeneracy}. Note that we display the results for the Hamiltonian model primarily for clarity, but the same degeneracy can also be observed in eigenstates of the Floquet circuit model.

\begin{figure}[h!]
    \centering
    \includegraphics[width=0.325\linewidth]{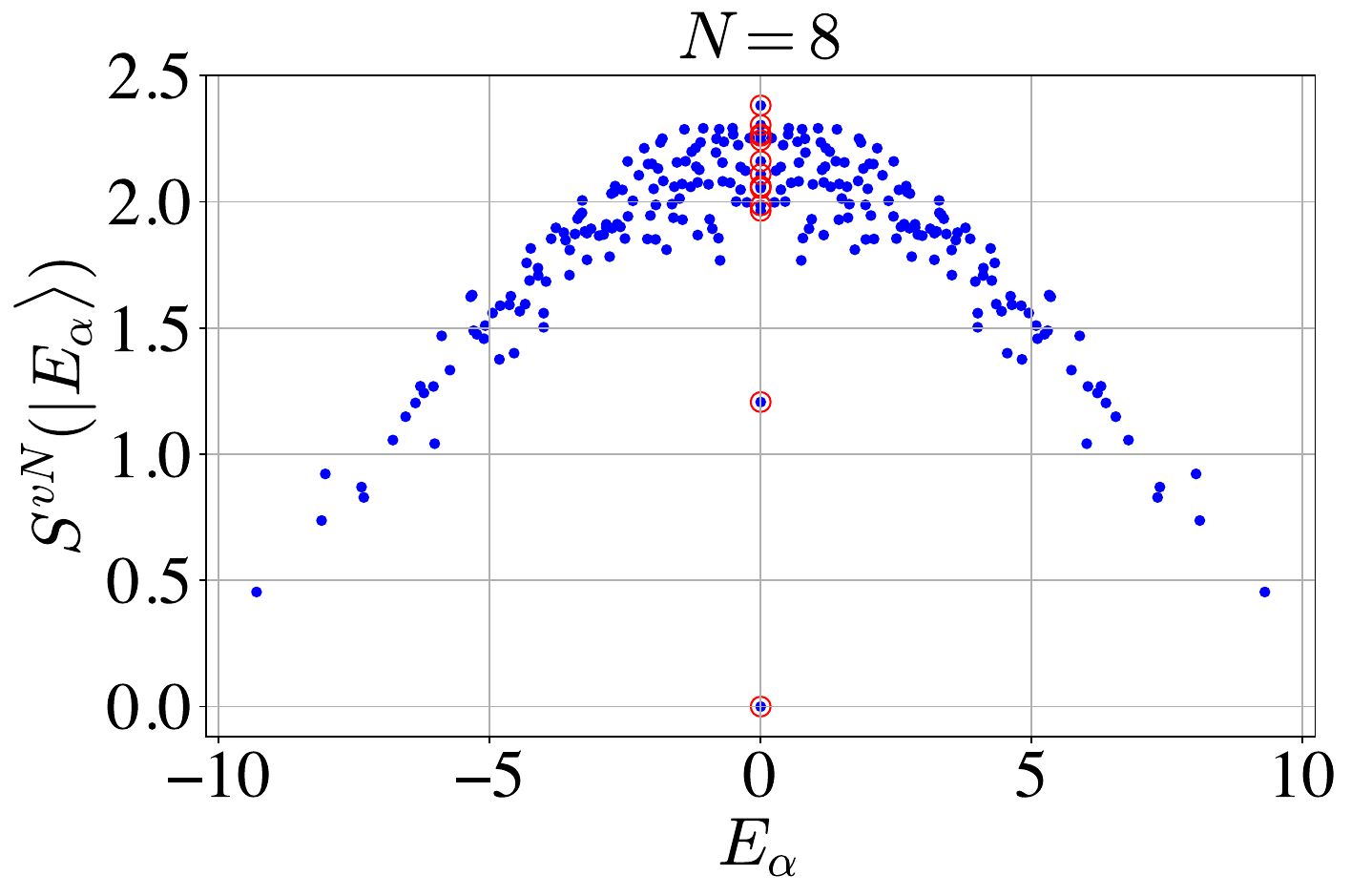}
    \includegraphics[width=0.325\linewidth]{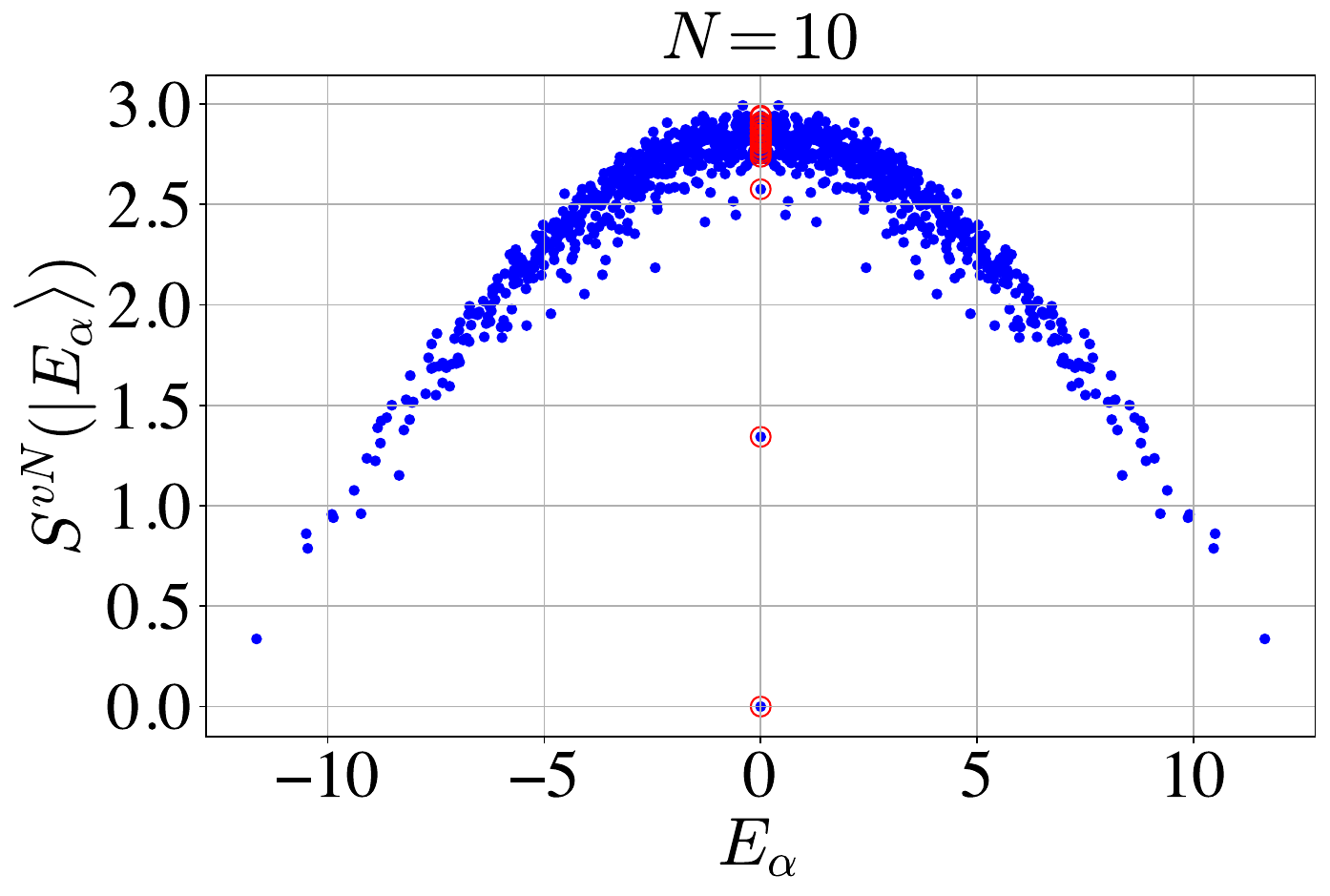}
    \includegraphics[width=0.325\linewidth]{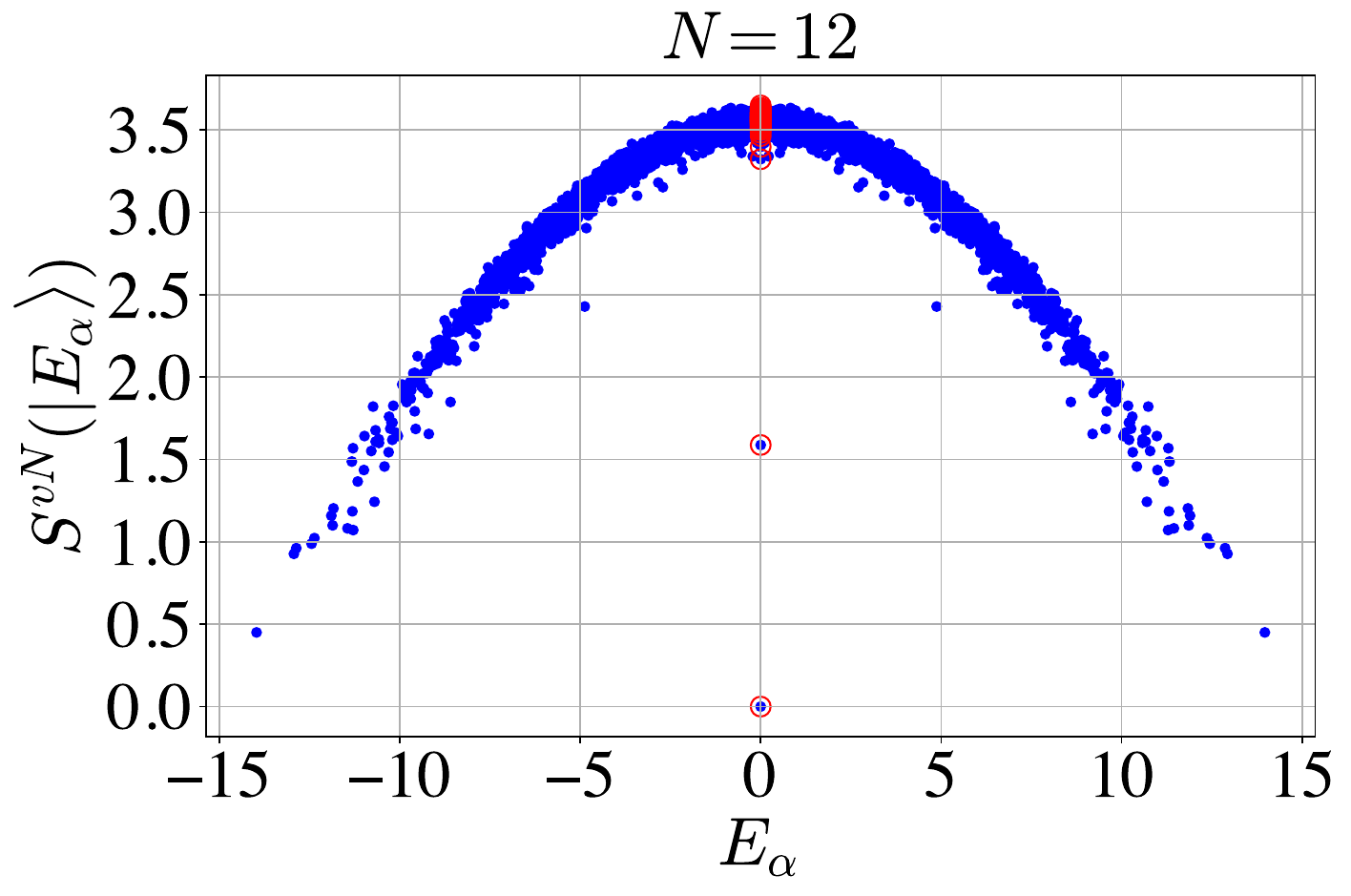}
    \caption{Bipartite entanglement entropy, $S^{vN}$, versus the energy $E_\alpha$, for the energy eigenstates $\ket{E _\alpha}$. We are working with the Hamiltonian model with alternating chains, where $\hat{h}_e = \tilde{\sigma}_x + 2\tilde{\sigma}_y, \; \hat{h}_o = 2\tilde{\sigma}_x + \frac{\pi}{2}\tilde{\sigma}_y$. The zero energy modes are marked with red circles. We use exact diagonalization, for system sizes $N = 8, \; 10, \; 12$.}
    \label{fig:exponential_zero_mode_degeneracy}
\end{figure}

\section{Additional Classical Simulations}
\label{app:additional_classical_results}

In this section we present classical numerical results demonstrating the presence of AQMBS states in the Hamiltonian and Floquet models, whose existence has been demonstrated in SM \ref{app: aqmbs_concrete_model}. We present the dynamics of the fidelity and total magnetisation, starting from three AQMBS states: $\{\ket{\mathcal{A}_1}, \, \ket{\mathcal{A}_2} \}$. We probe two instances of the Hamiltonian, with and without the exponential degeneracy at $E_\alpha = 0$, and display the corresponding results in Fig. \ref{fig: aqmbs_dynamics_hamiltonian}. Similarly, we also probe the dynamics for the Floquet circuit model, with and without the exponential degeneracy in the middle of the spectrum, in Fig. \ref{fig: aqmbs_dynamics_floquet}.

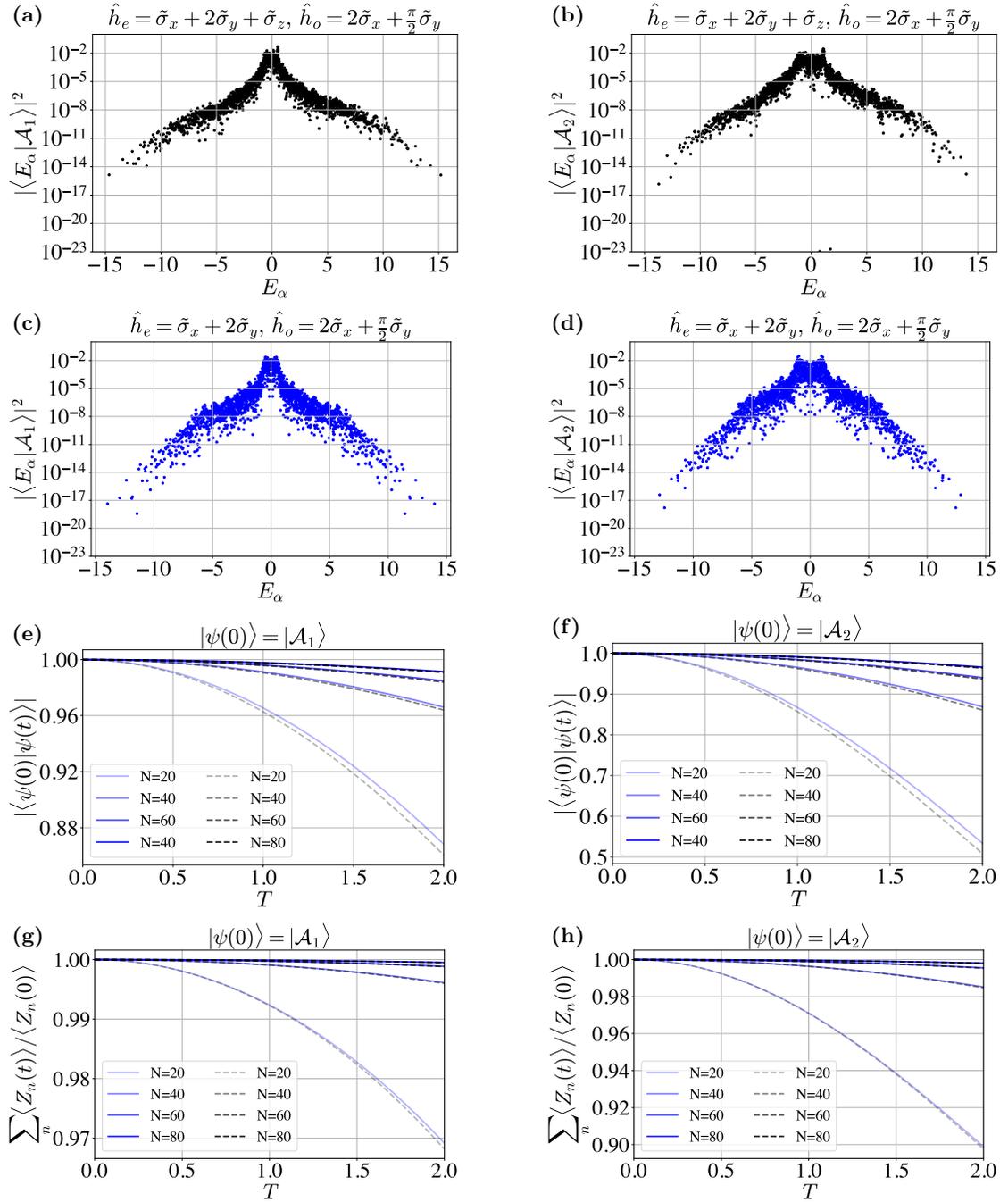
\begin{figure*}[h!]
    \centering

    \begin{tikzpicture}
        \node[inner sep=0] (a) {\includegraphics[width=0.38\linewidth]{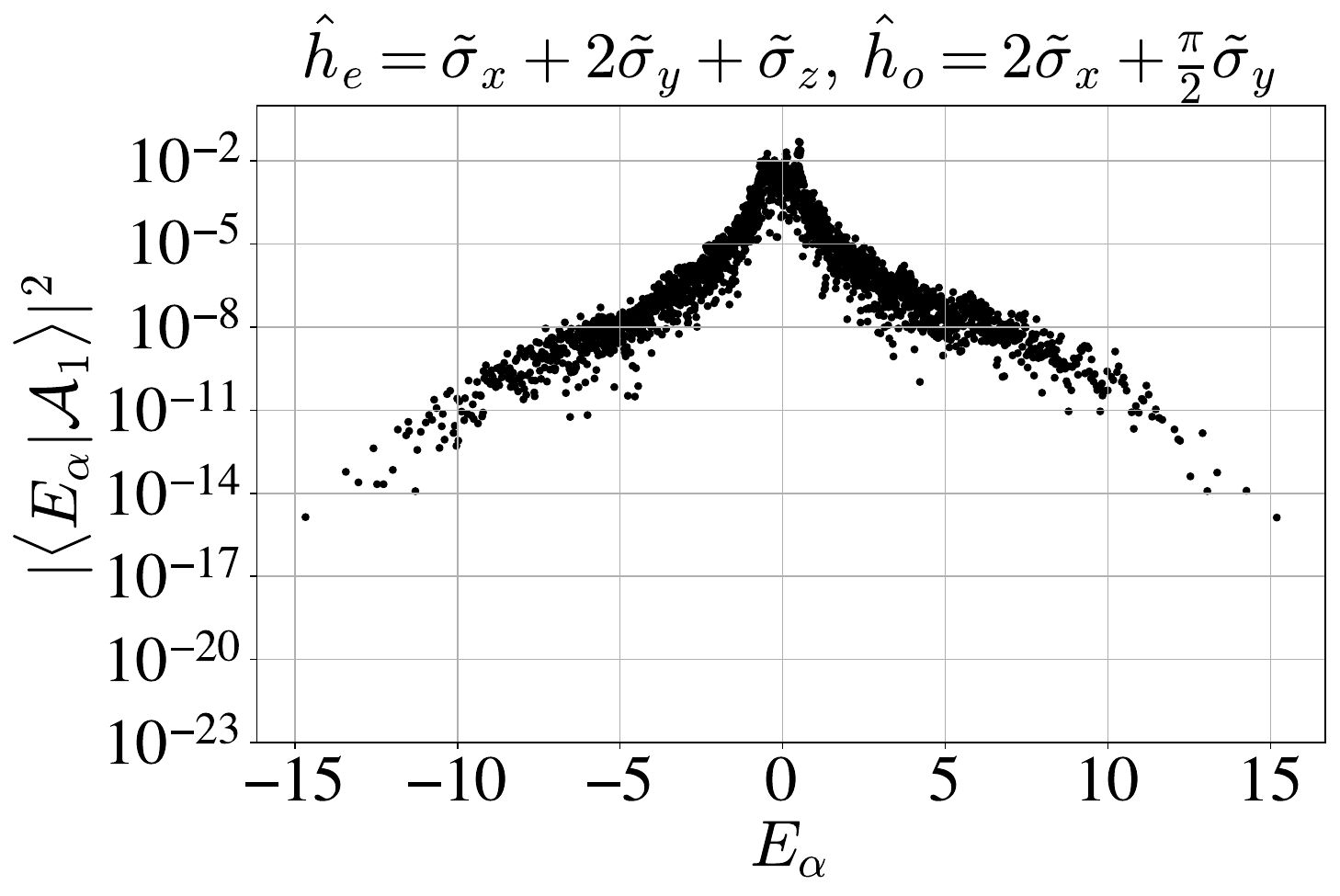}};
        \node[anchor=north west,font=\bfseries] at (a.north west) {(a)};
    \end{tikzpicture}
    \hspace{1cm}
    \begin{tikzpicture}
        \node[inner sep=0] (b) {\includegraphics[width=0.38\linewidth]{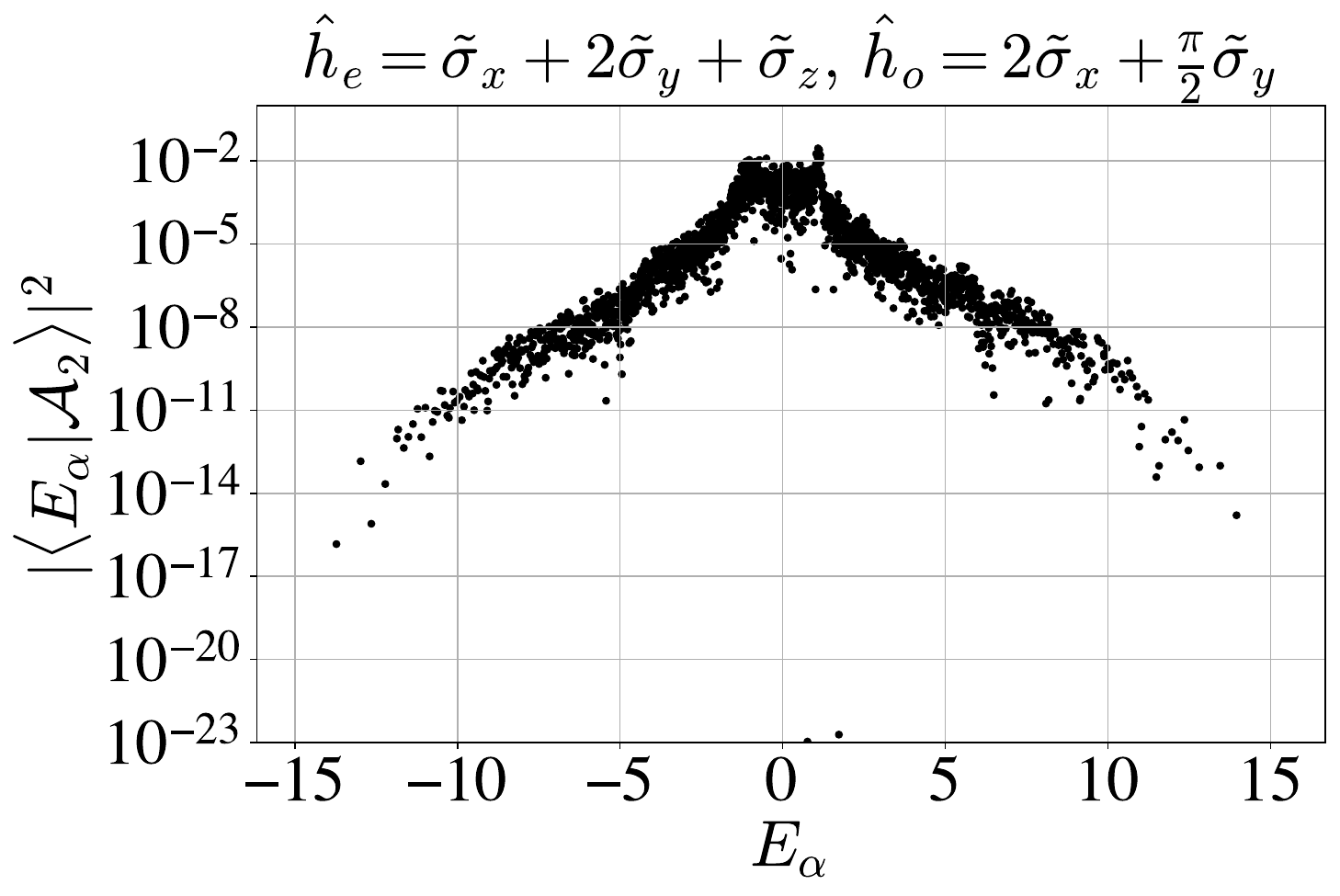}};
        \node[anchor=north west,font=\bfseries] at (b.north west) {(b)};
    \end{tikzpicture}


    \begin{tikzpicture}
        \node[inner sep=0] (c) {\includegraphics[width=0.38\linewidth]{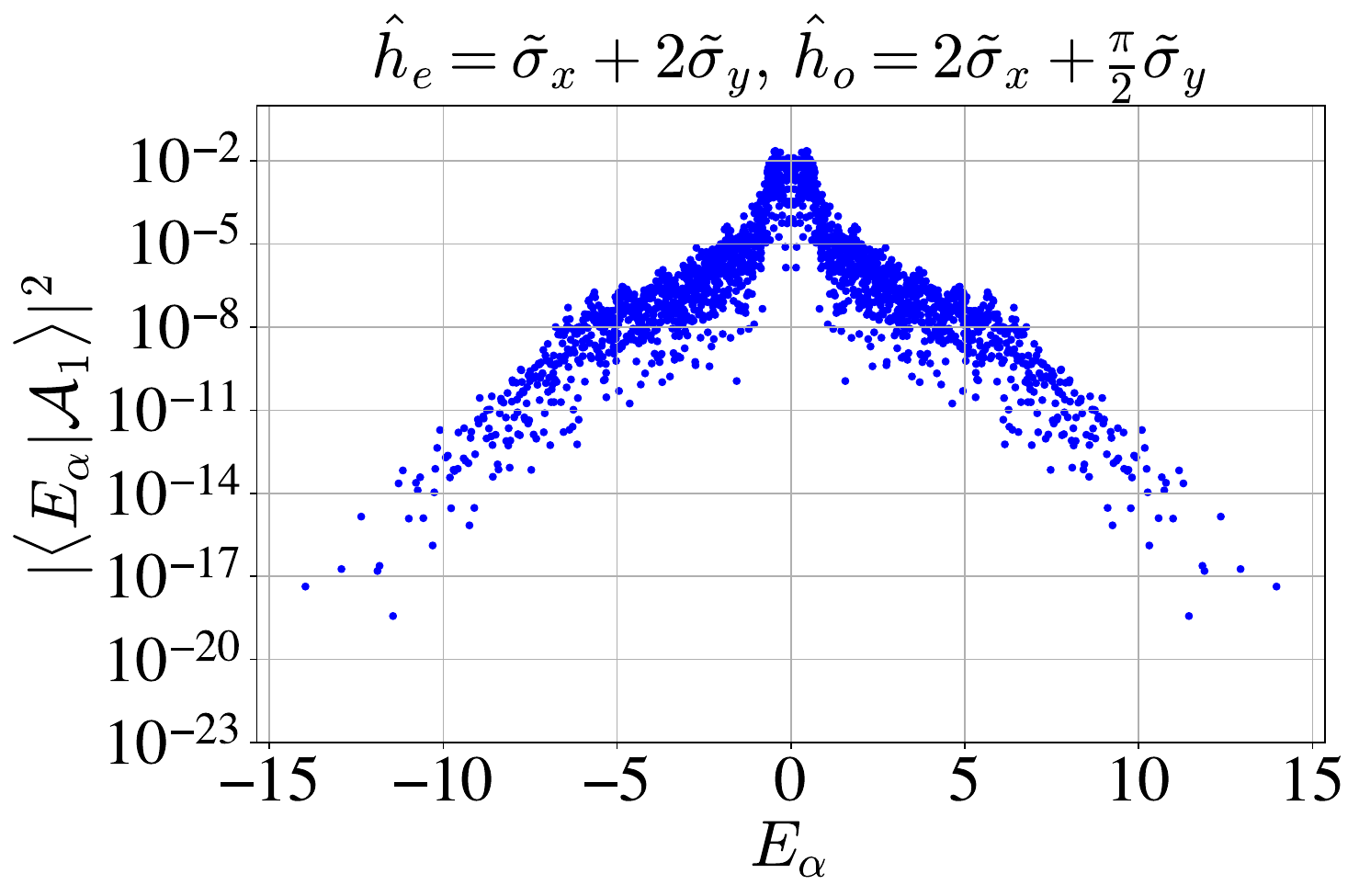}};
        \node[anchor=north west,font=\bfseries] at (c.north west) {(c)};
    \end{tikzpicture}
    \hspace{1cm}
    \begin{tikzpicture}
        \node[inner sep=0] (d) {\includegraphics[width=0.38\linewidth]{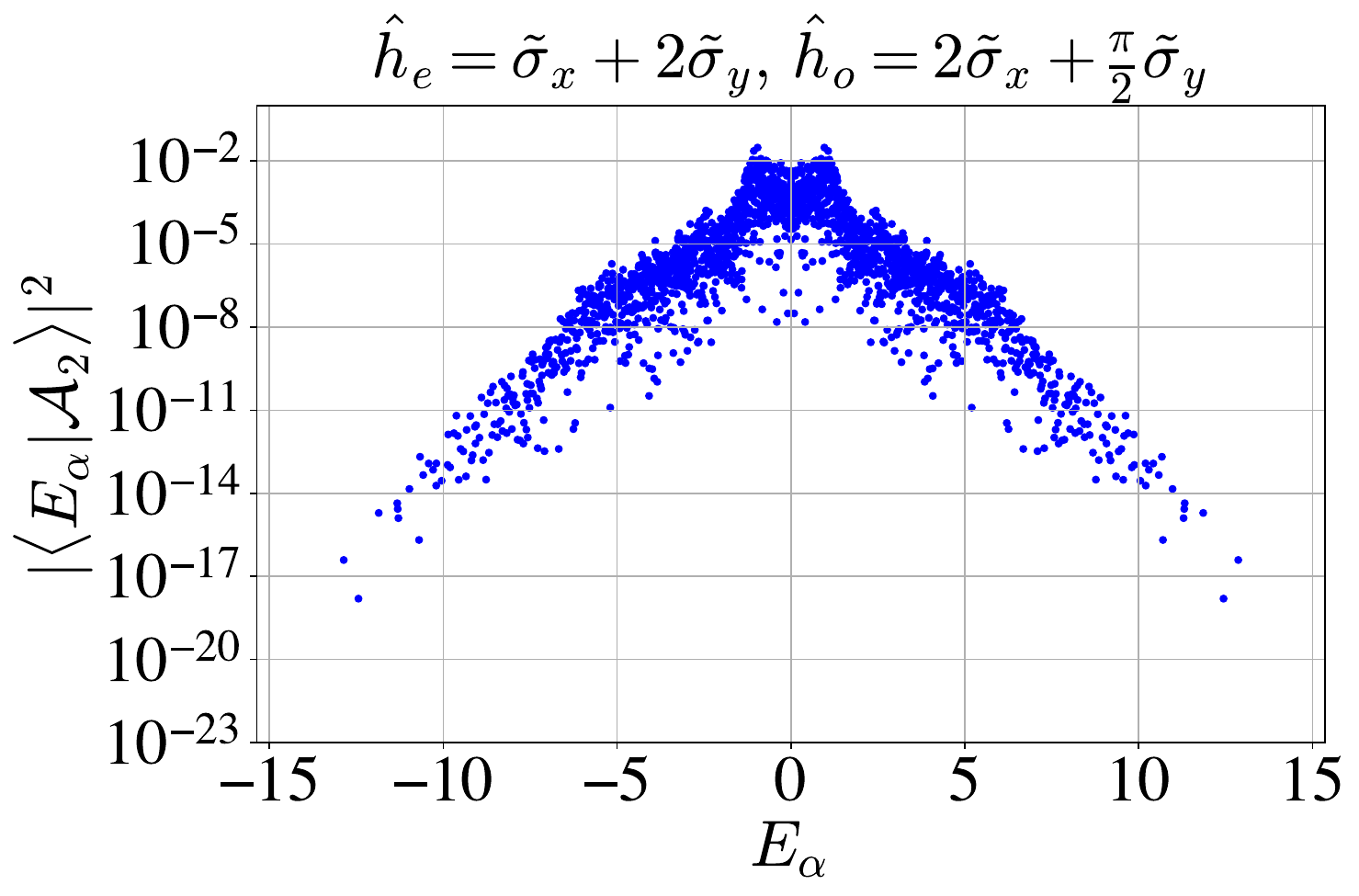}};
        \node[anchor=north west,font=\bfseries] at (d.north west) {(d)};
    \end{tikzpicture}

    \begin{tikzpicture}
        \node[inner sep=0] (e) {\includegraphics[width=0.38\linewidth]{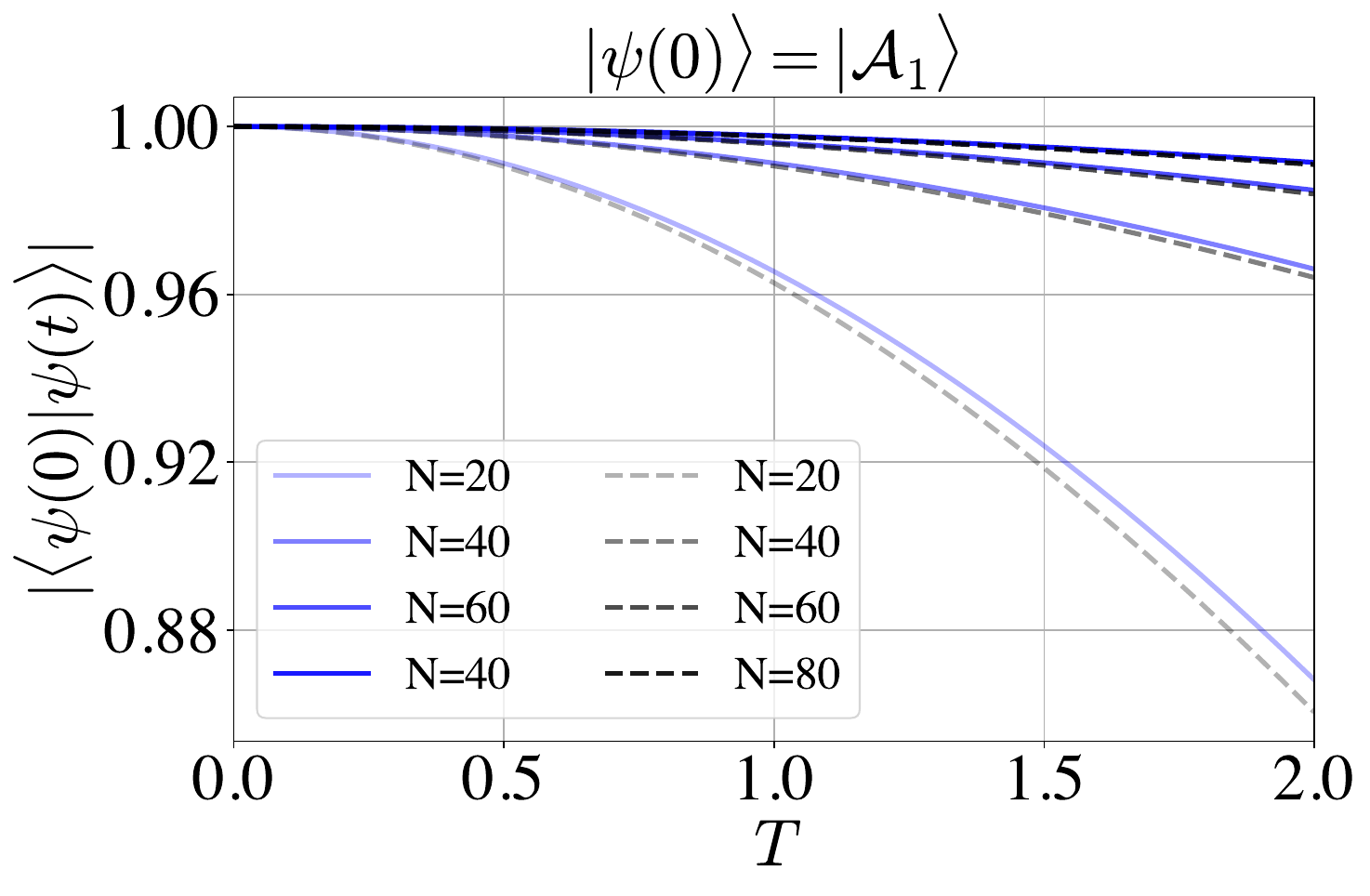}};
        \node[anchor=north west,font=\bfseries] at (e.north west) {(e)};
    \end{tikzpicture}
    \hspace{1cm}
    \begin{tikzpicture}
        \node[inner sep=0] (f) {\includegraphics[width=0.38\linewidth]{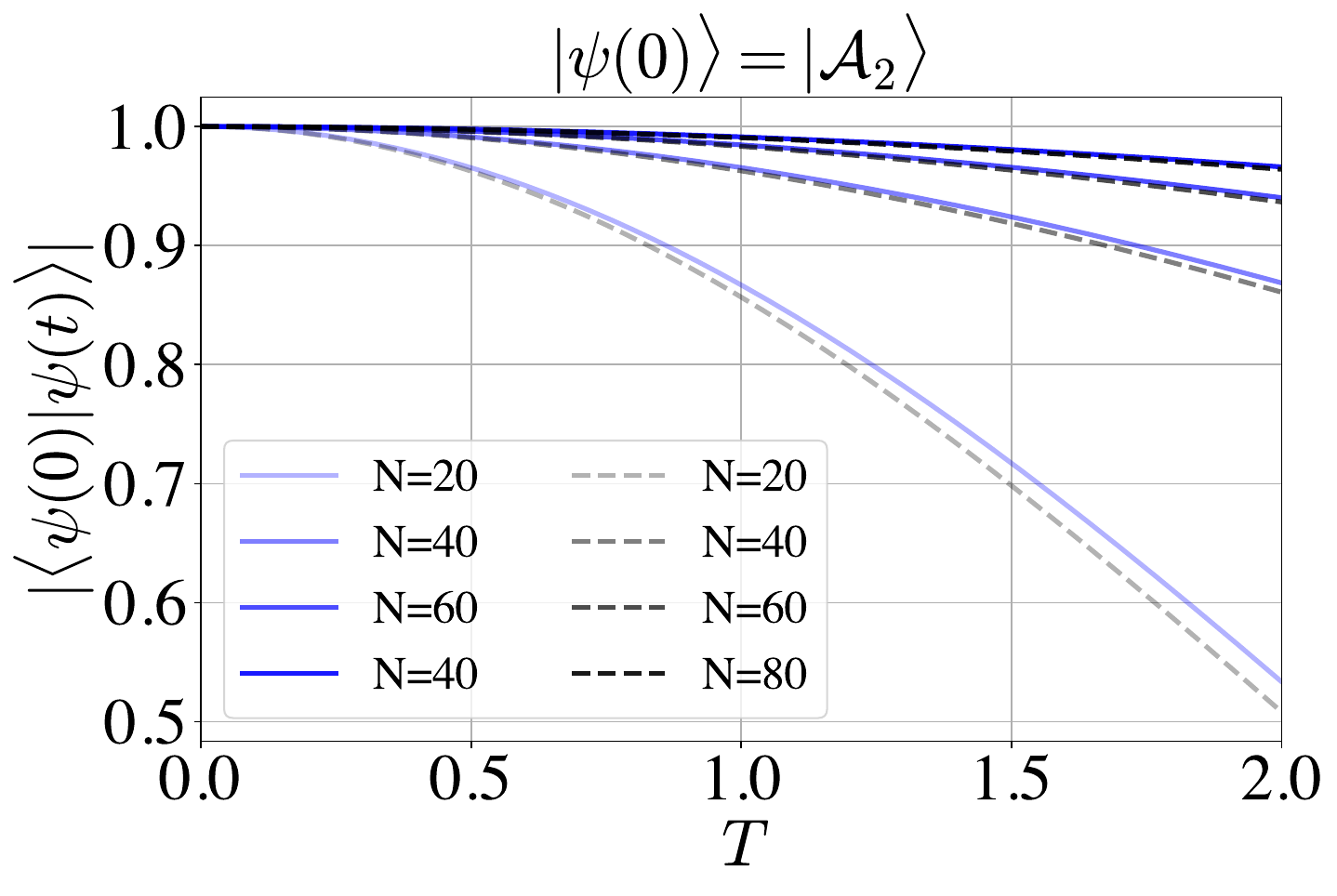}};
        \node[anchor=north west,font=\bfseries] at (f.north west) {(f)};
    \end{tikzpicture}

    \begin{tikzpicture}
        \node[inner sep=0] (g) {\includegraphics[width=0.38\linewidth]{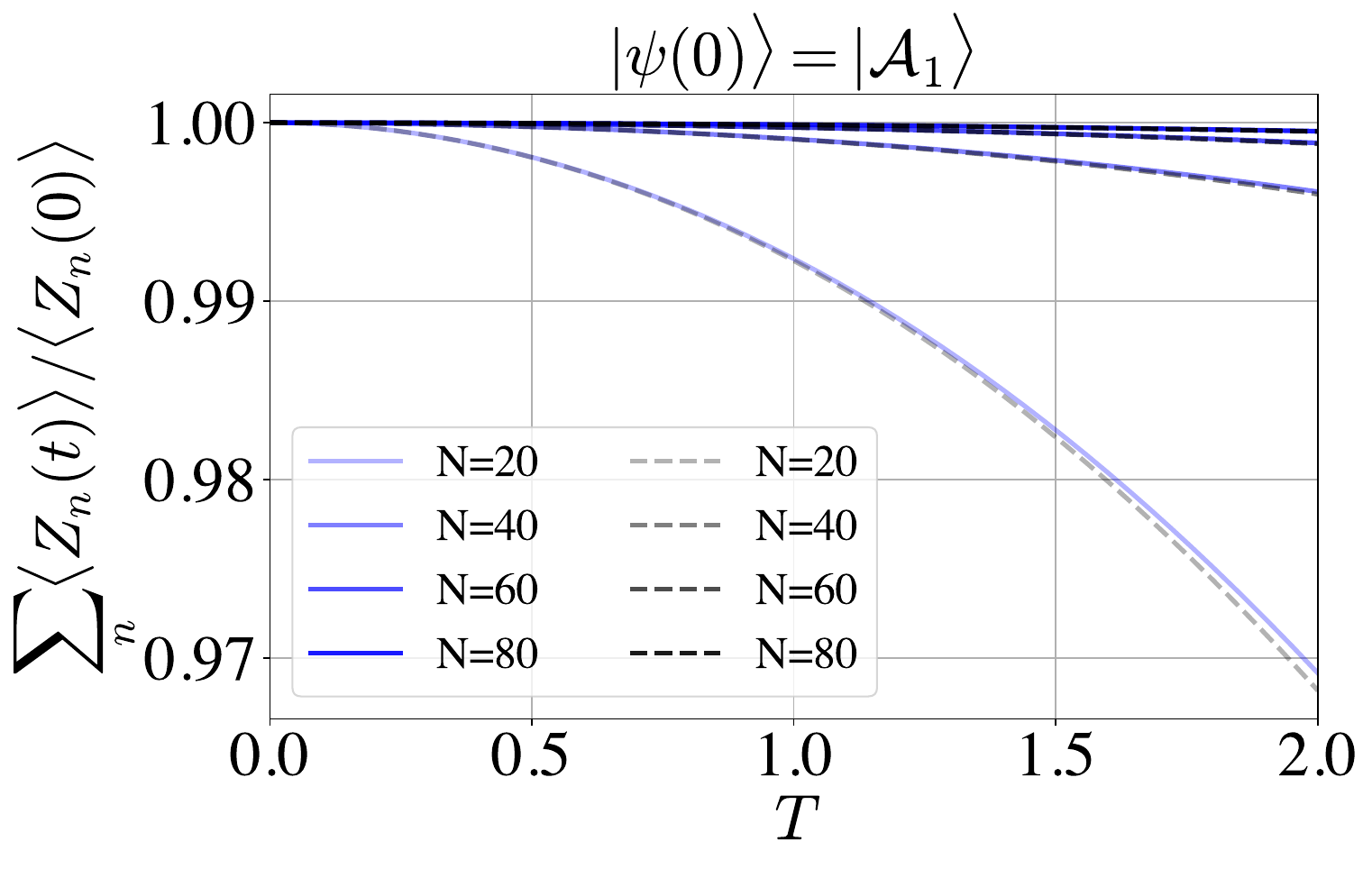}};
        \node[anchor=north west,font=\bfseries] at (g.north west) {(g)};
    \end{tikzpicture}
    \hspace{1cm}
    \begin{tikzpicture}
        \node[inner sep=0] (h) {\includegraphics[width=0.38\linewidth]{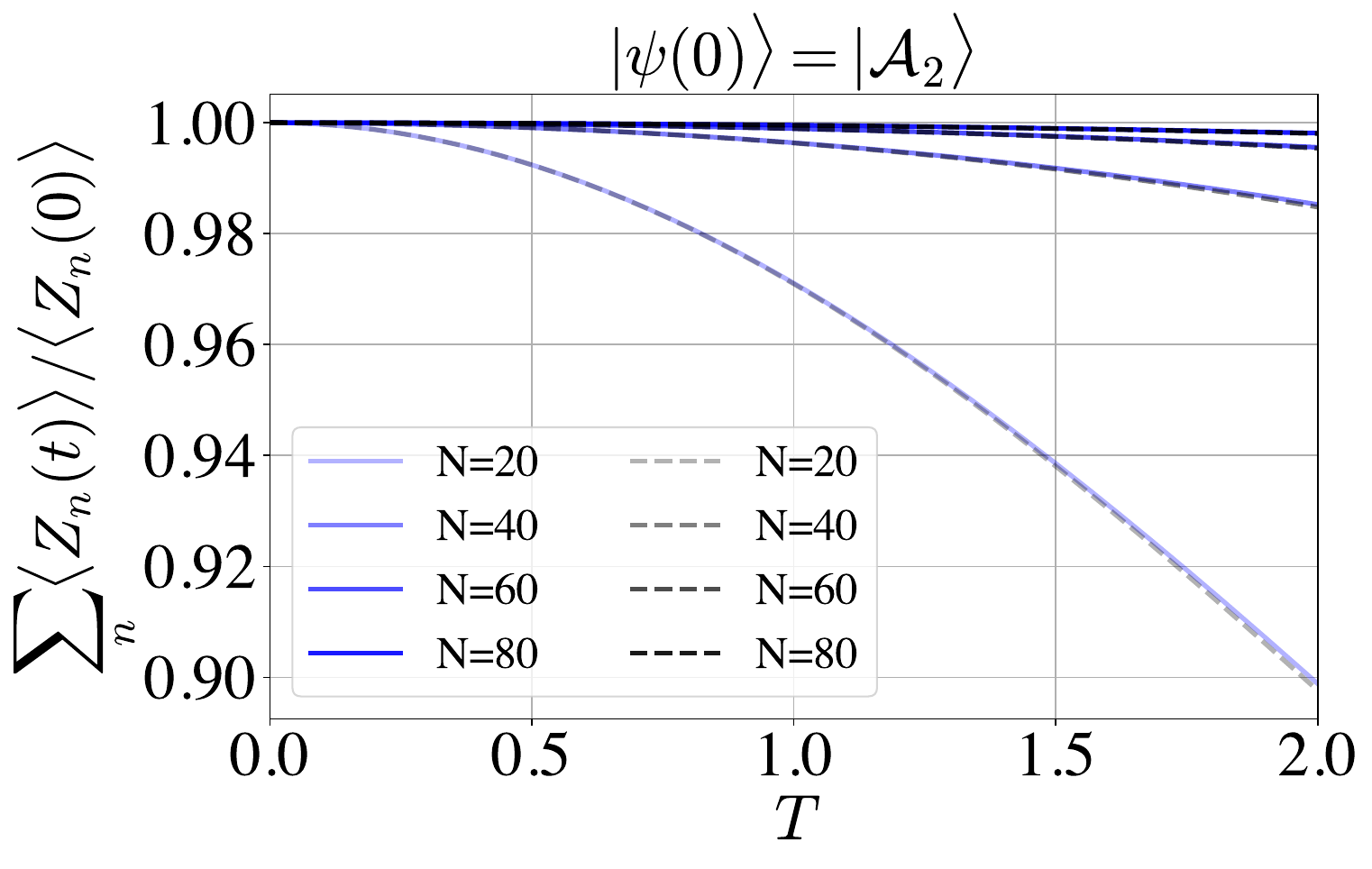}};
        \node[anchor=north west,font=\bfseries] at (h.north west) {(h)};
    \end{tikzpicture}

    \caption{
    Results for the Hamiltonian with and without exponential degeneracy at $E_\alpha = 0$. Panels \textbf{(a), (b)} show overlaps of $\{\ket{\mathcal{A}_1}, \ket{\mathcal{A}_2}\}$ with energy eigenstates for $\hat{h}_e = \tilde{\sigma}_x + 2\tilde{\sigma}_y + \tilde{\sigma}_z$, $\hat{h}_o = 2\tilde{\sigma}_x + \frac{\pi}{2}\tilde{\sigma}_y$, respectively. Panels \textbf{(c), (d)} show the same overlaps for $\hat{h}_e = \tilde{\sigma}_x + 2\tilde{\sigma}_y$, $\hat{h}_o = 2\tilde{\sigma}_y + \frac{\pi}{2}\tilde{\sigma}_y$, where an exponential mid-spectrum degeneracy is present, as discussed in Supplementary material \ref{app: midspectrum degeneracy}. Parameters used to obtain the overlaps: $g=1$, $N=12$. Black (blue) denotes results without (with) degeneracy.
    The dynamical results are obtained via TEBD (with bond dimension $\chi=512$ and time-step $\delta t=0.01$), using the Julia ITensor package \cite{itensor}. The fidelity with respect to the initial state is shown in \textbf{(e), (f)} and total magnetisation in \textbf{(g), (h)} for $N \in \{20,40,60,80\}$. Increasing $N$ slows the decay of both quantities, a clear indication of AQMBS. For any of the probed $N$, decay is faster at larger $k$, consistent with the rate of the gap closing in the reference Hamiltonian.}
    \label{fig: aqmbs_dynamics_hamiltonian}
\end{figure*}

\begin{figure*}
    \centering
    \begin{tikzpicture}
        \node[inner sep=0] (a) {\includegraphics[width=0.38\linewidth]{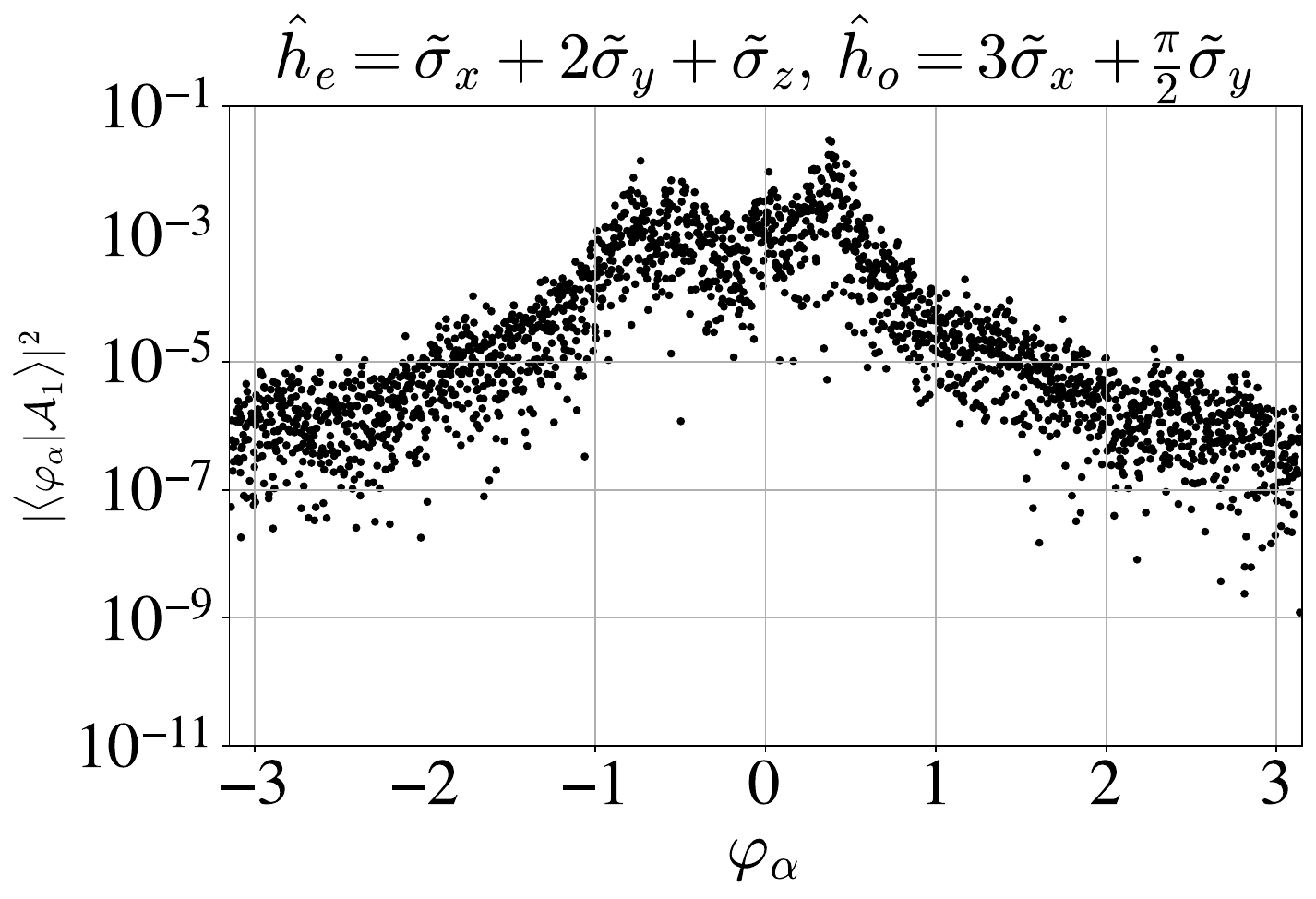}};
        \node[anchor=north west,font=\bfseries] at (a.north west) {(a)};
    \end{tikzpicture}
    \hspace{1cm}
    \begin{tikzpicture}
        \node[inner sep=0] (b) {\includegraphics[width=0.38\linewidth]{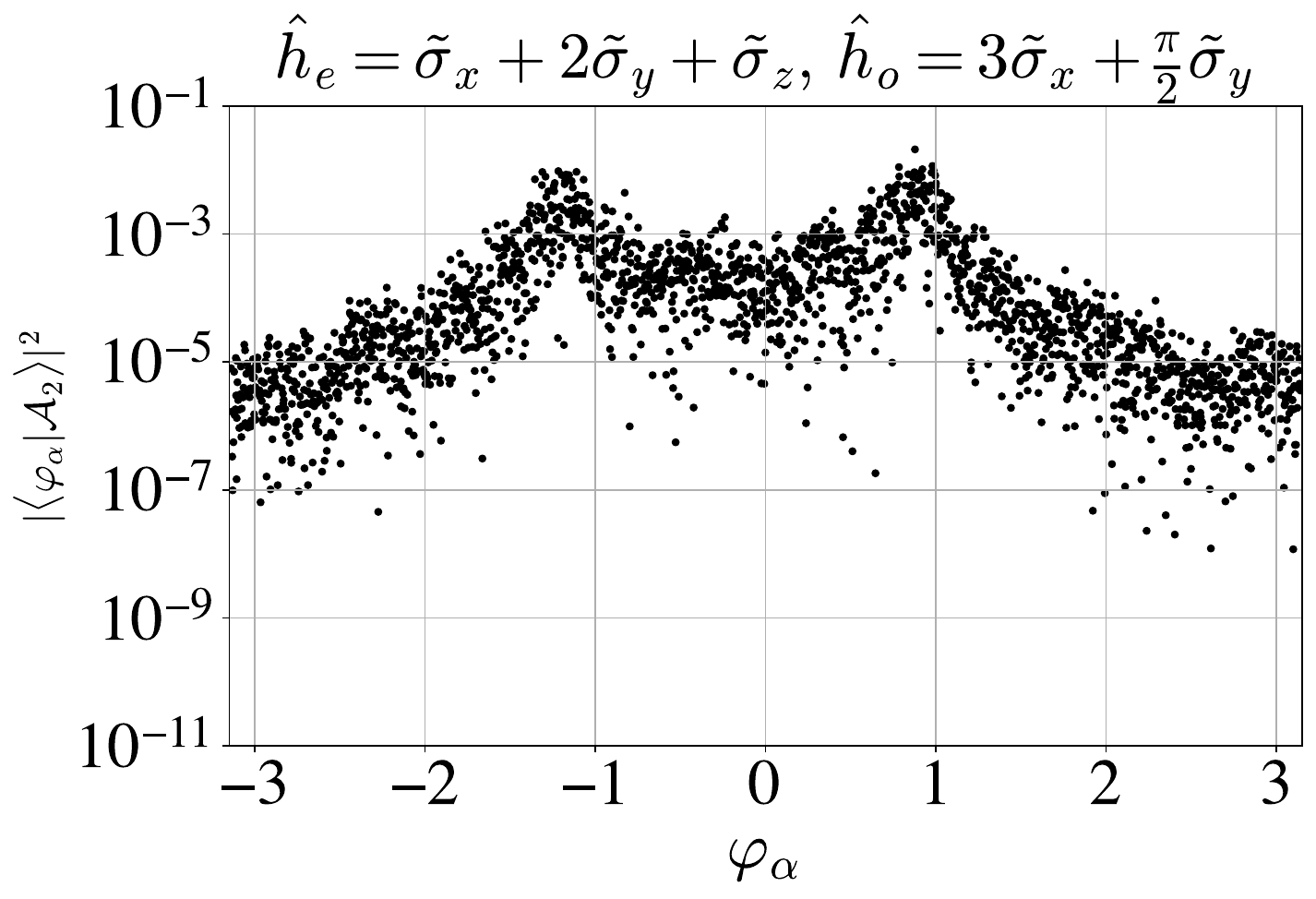}};
        \node[anchor=north west,font=\bfseries] at (b.north west) {(b)};
    \end{tikzpicture}

    \begin{tikzpicture}
        \node[inner sep=0] (c) {\includegraphics[width=0.38\linewidth]{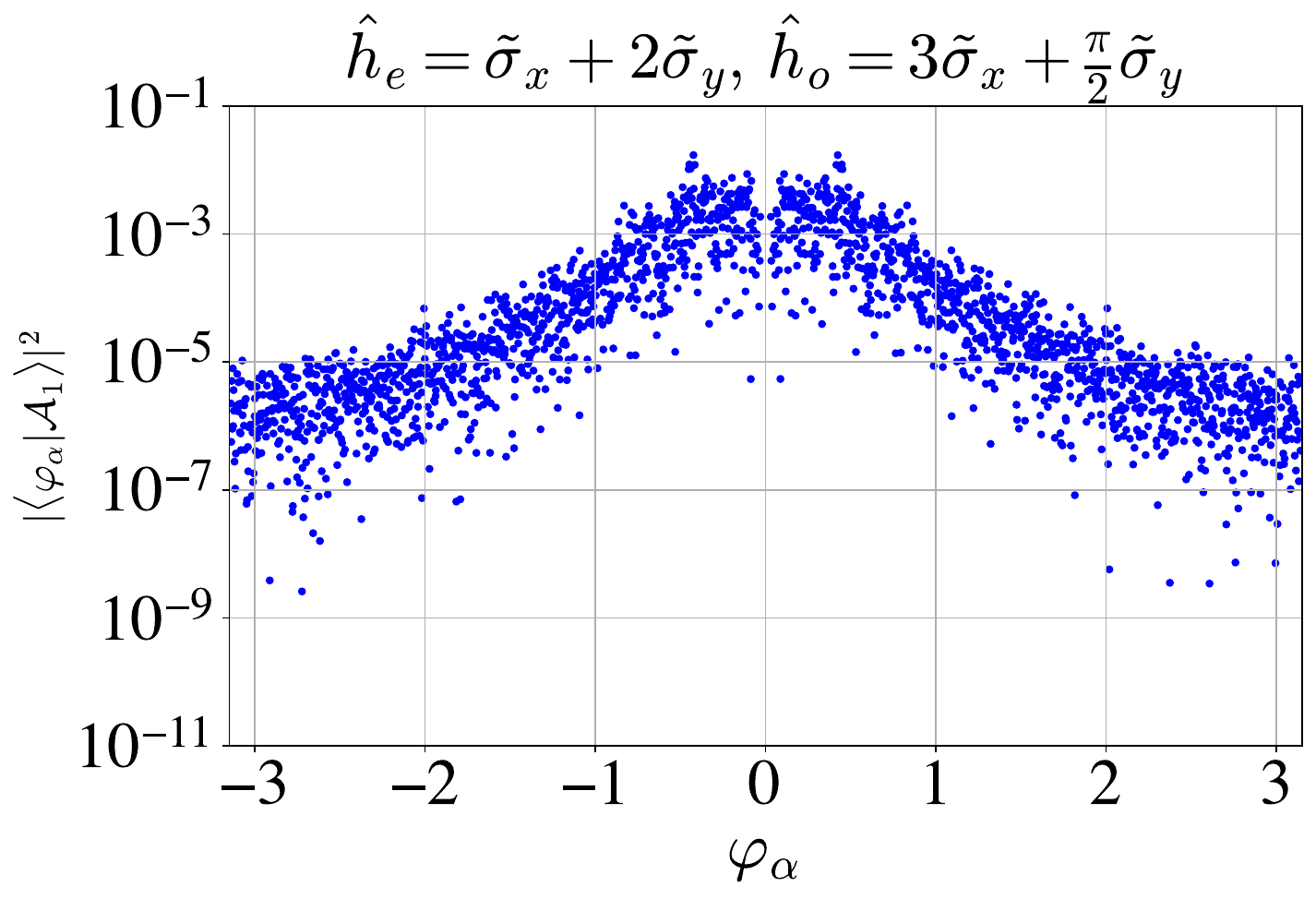}};
        \node[anchor=north west,font=\bfseries] at (c.north west) {(c)};
    \end{tikzpicture}
    \hspace{1cm}
    \begin{tikzpicture}
        \node[inner sep=0] (d) {\includegraphics[width=0.38\linewidth]{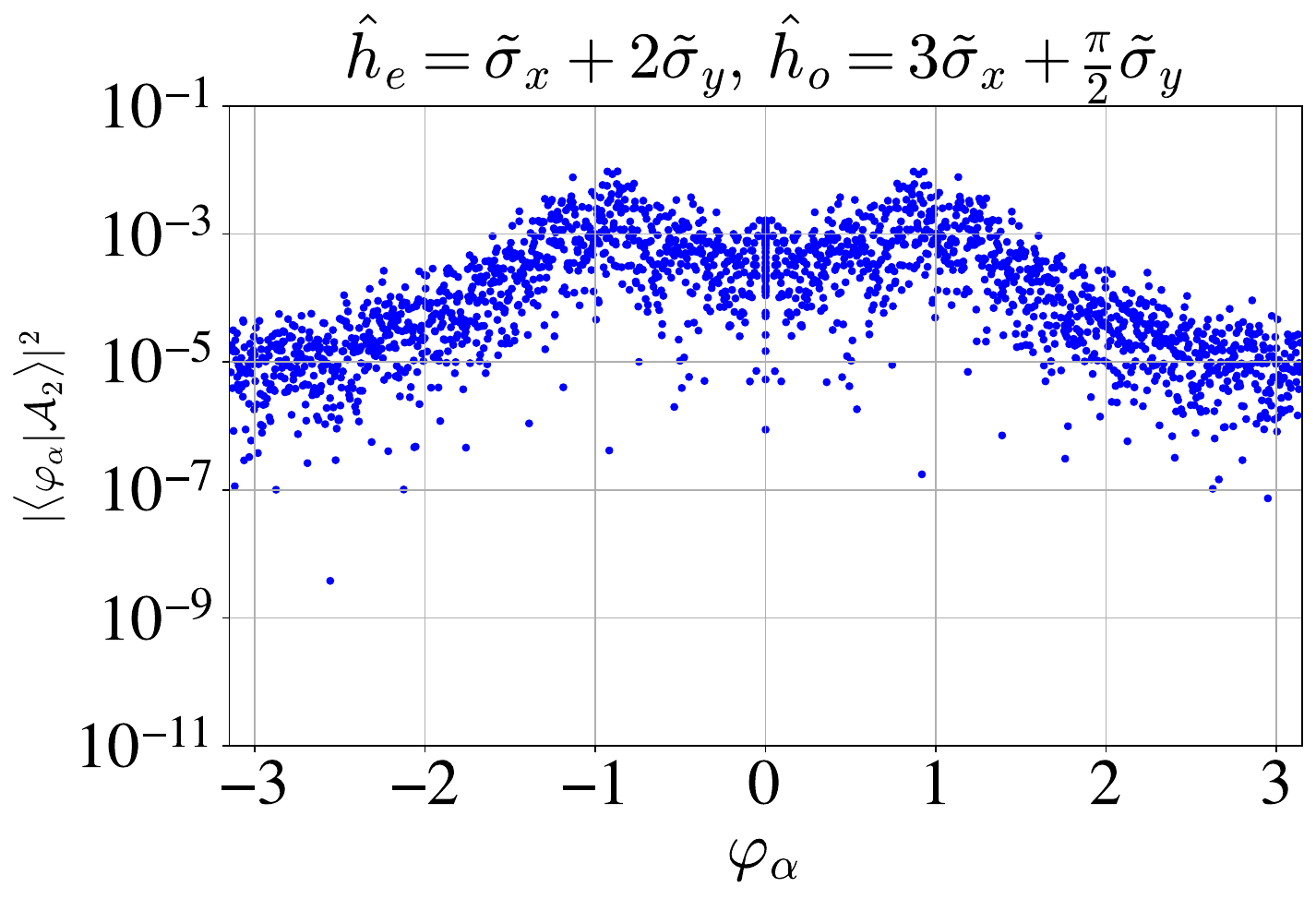}};
        \node[anchor=north west,font=\bfseries] at (d.north west) {(d)};
    \end{tikzpicture}

    \begin{tikzpicture}
        \node[inner sep=0] (e) {\includegraphics[width=0.38\linewidth]{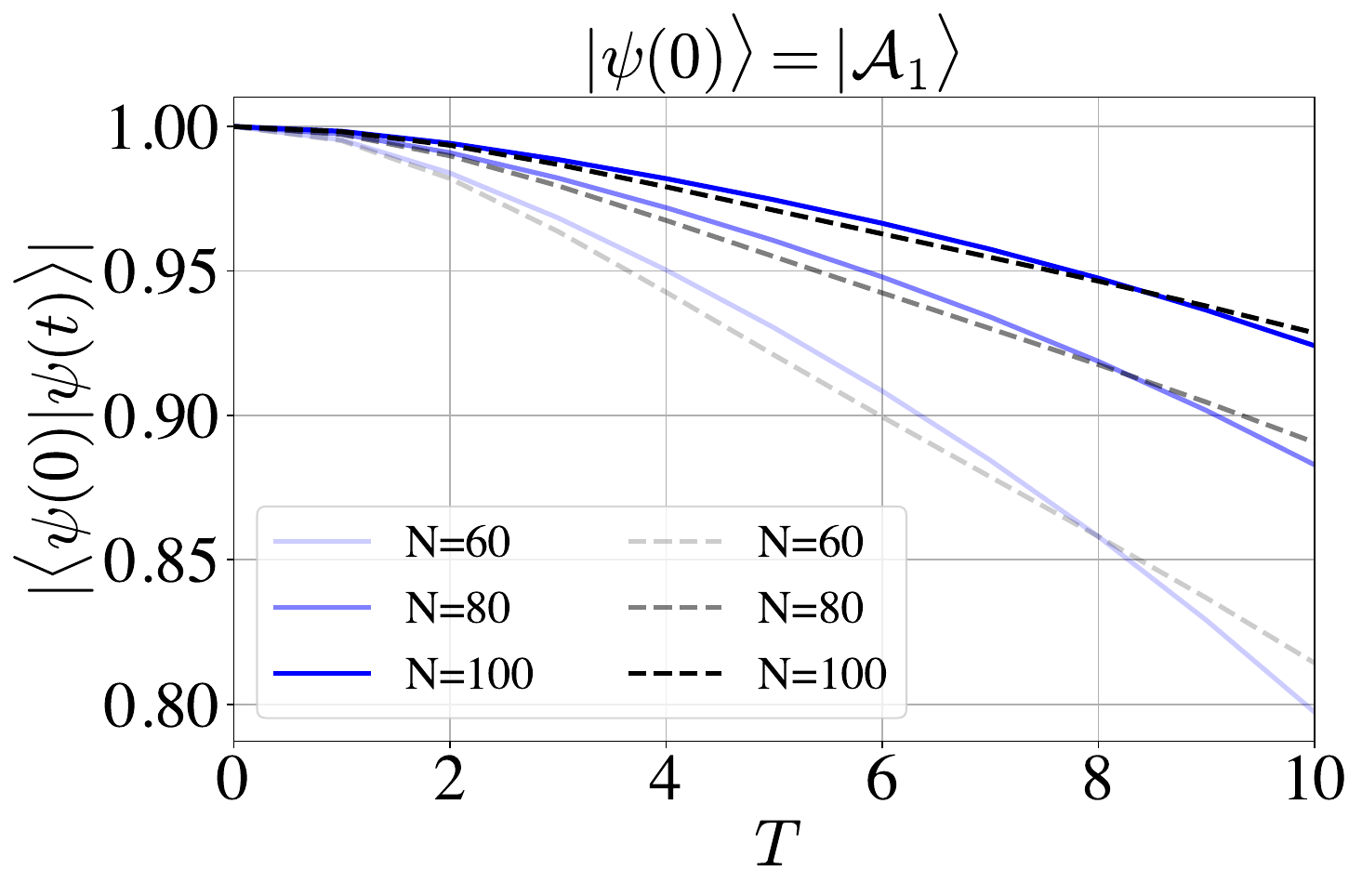}};
        \node[anchor=north west,font=\bfseries] at (e.north west) {(e)};
    \end{tikzpicture}
    \hspace{1cm}
    \begin{tikzpicture}
        \node[inner sep=0] (f) {\includegraphics[width=0.38\linewidth]{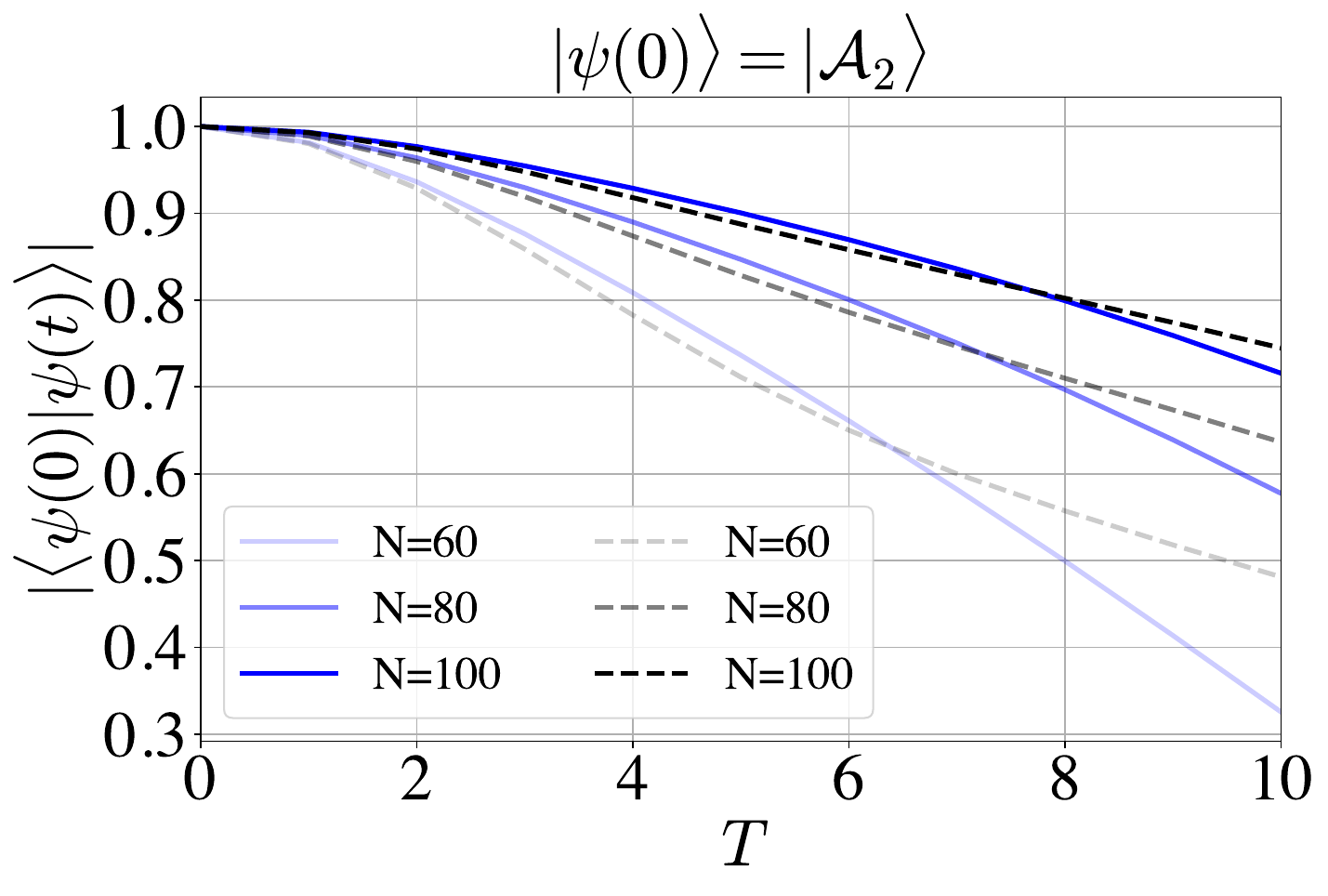}};
        \node[anchor=north west,font=\bfseries] at (f.north west) {(f)};
    \end{tikzpicture}

    \begin{tikzpicture}
        \node[inner sep=0] (g) {\includegraphics[width=0.38\linewidth]{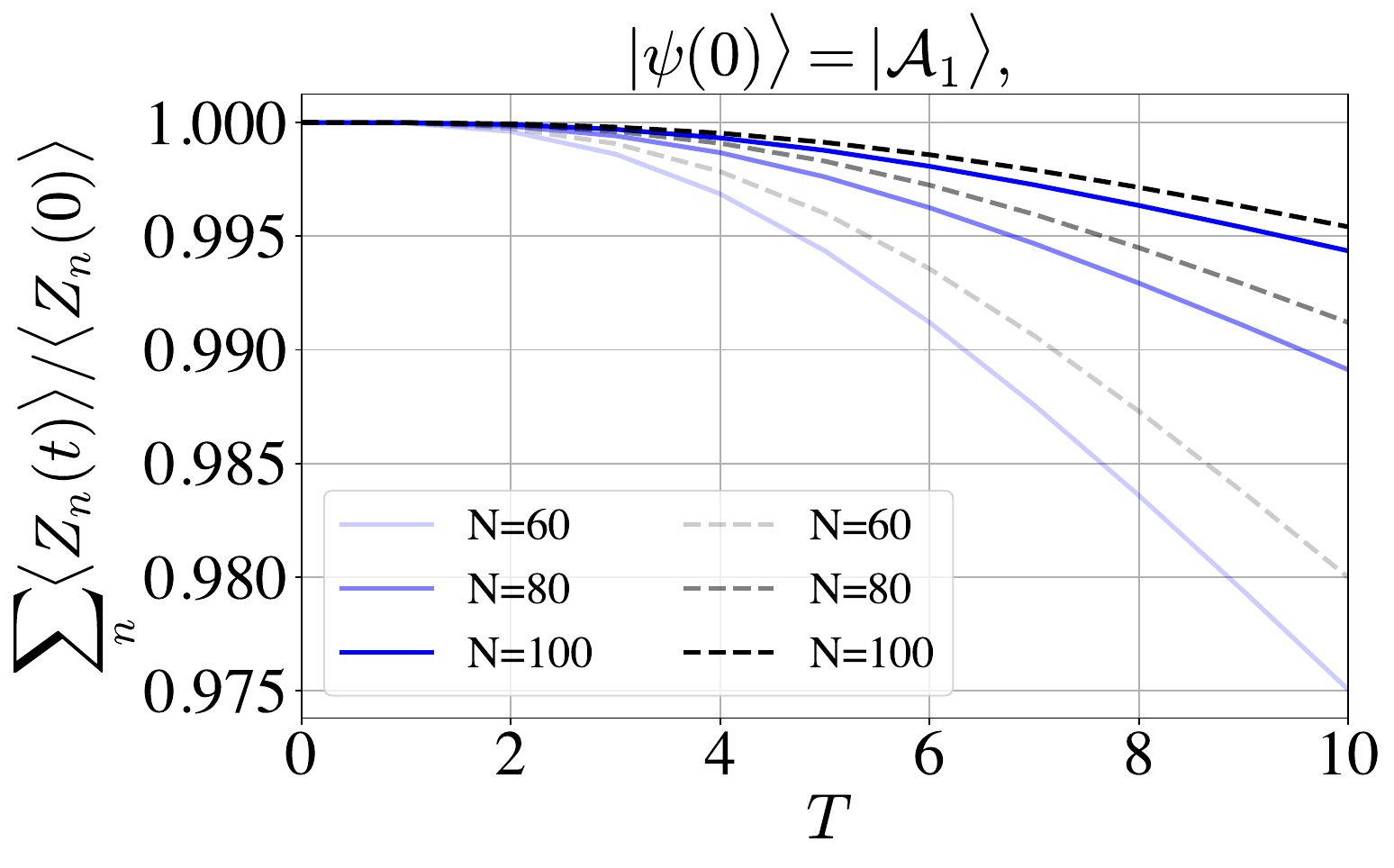}};
        \node[anchor=north west,font=\bfseries] at (g.north west) {(g)};
    \end{tikzpicture}
    \hspace{1cm}
    \begin{tikzpicture}
        \node[inner sep=0] (h) {\includegraphics[width=0.38\linewidth]{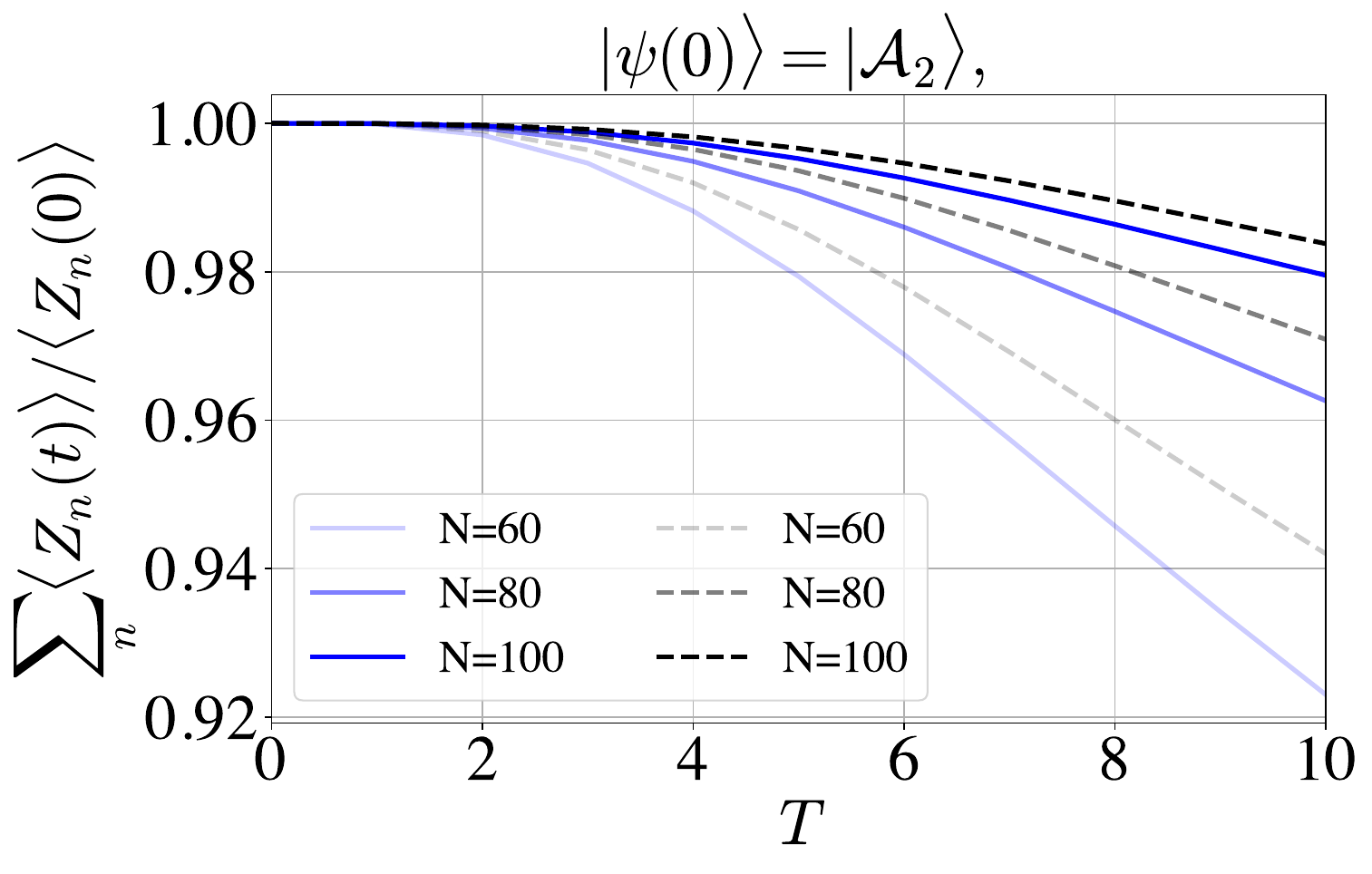}};
        \node[anchor=north west,font=\bfseries] at (h.north west) {(h)};
    \end{tikzpicture}

    \caption{Results for the Floquet circuit model with and without exponential degeneracy at $\varphi_\alpha = 0$. Panels \textbf{(a), (b)} show overlaps of $\{\ket{\mathcal{A}_1}, \ket{\mathcal{A}_2}\}$ with the Floquet eigenstates for $\hat{h}_e = \tilde{\sigma}_x + 2\tilde{\sigma}_y + \tilde{\sigma}_z$, $\hat{h}_o = 3\tilde{\sigma}_x + \frac{\pi}{2}\tilde{\sigma}_y$, respectively. Panels \textbf{(c), (d)} show the same overlaps for $\hat{h}_e = \tilde{\sigma}_x + 2\tilde{\sigma}_y$, $\hat{h}_o = 3\tilde{\sigma}_x + \frac{\pi}{2}\tilde{\sigma}_y$, where an exponential mid-spectrum degeneracy is present, as discussed in Supplementary material \ref{app: midspectrum degeneracy}. Parameters used to obtain the overlaps: $g=1$, $N=12$. Black (blue) denotes results without (with) degeneracy.
    The dynamical results are obtained via TEBD (with bond dimension $\chi=512$), using Julia ITensor  \cite{itensor}. The fidelity with respect to the initial state is shown in \textbf{(e), (f)} and total magnetisation in \textbf{(g), (h)} for $N \in \{60, \, 80, \, 100\}$. Increasing $N$ slows the decay of the fidelity as well as the magnetisation, a clear indication of AQMBS, similarly to the Hamiltonian model case.}
    \label{fig: aqmbs_dynamics_floquet}
\end{figure*}
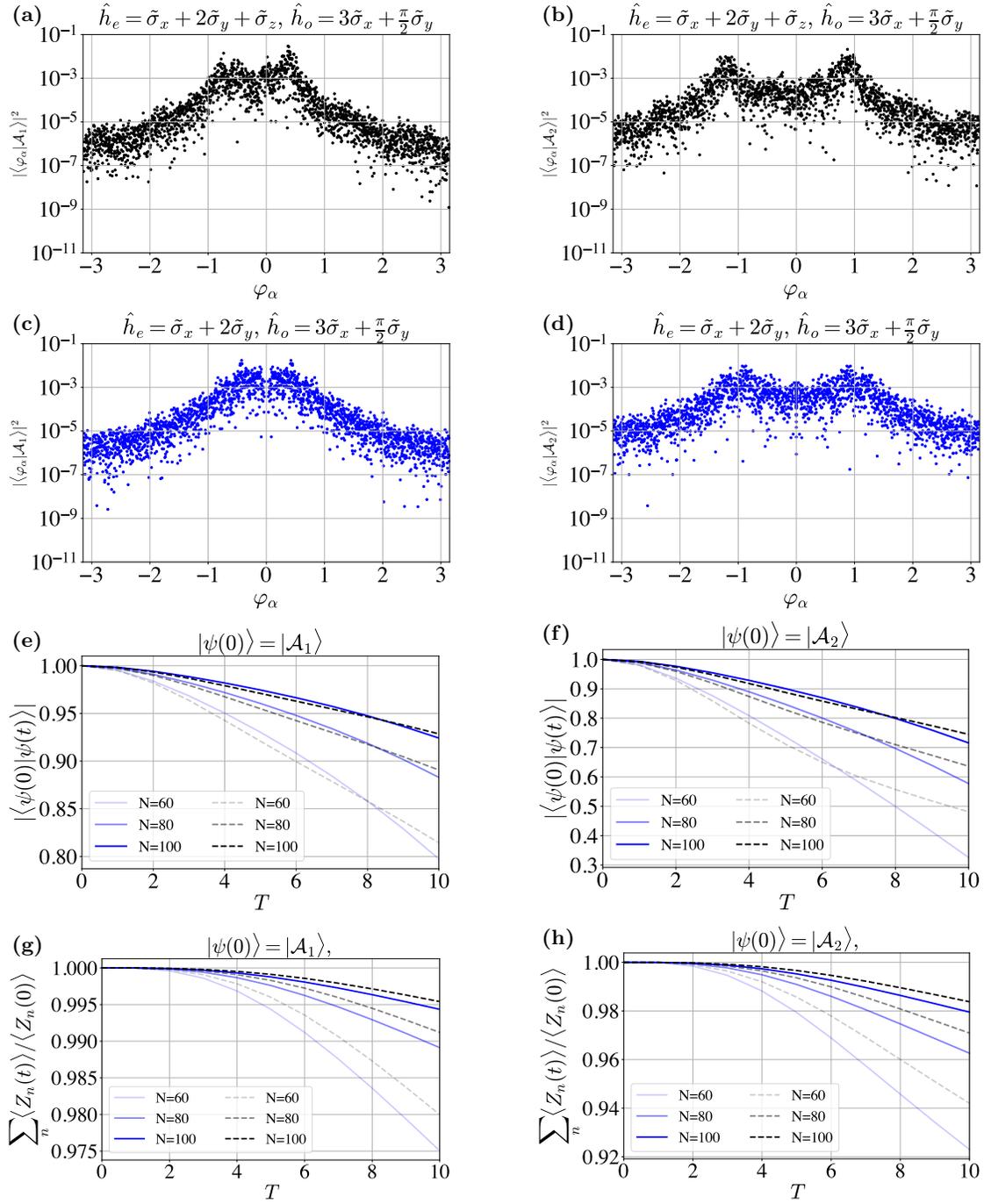

\end{document}